\documentclass{lmcs} 
\pdfoutput=1
\usepackage[utf8]{inputenc}

\usepackage{lastpage}
\lmcsdoi{19}{2}{13}
\lmcsheading{}{\pageref{LastPage}}{}{}%
{Mar.~15,~2022}{May~31,~2023}{}

\pdfoutput=1


\usepackage{hyperref}

\def\eg{{\em{e.g.}}}
\def\cf{{\em{cf.}}}
\def\ie{{\em{i.e.}}}

\usepackage{amsmath}
\usepackage{amsfonts}
\usepackage{amssymb}
\usepackage{ wasysym }
\usepackage{amsthm}
\usepackage{ stmaryrd }
\usepackage{hyperref}
\usepackage{cleveref}
\usepackage{fontawesome}
\usepackage{algpseudocode}
\usepackage[most]{tcolorbox}
\tcbuselibrary{skins}
\usepackage{wrapfig}
\usepackage{booktabs}
\usepackage{thm-restate}

\usepackage{listings}
\usepackage{tikz}
\usetikzlibrary{arrows,patterns}
\usepackage{bussproofs}

\usepackage{caption}
\usepackage{subcaption}
\usepackage{ltl}

\newcommand{\ltlN}{\LTLnext}
\newcommand{\ltlG}{\LTLglobally}
\newcommand{\ltlF}{\LTLeventually}
\newcommand{\ltlU}{\LTLuntil}
\newcommand{\ltlW}{\LTLweakuntil}

\usepackage{pifont}
\newcommand{\cmark}{\ding{51}}%
\newcommand{\xmark}{\ding{55}}%

\usepackage{MnSymbol}

\newcommand\xqed[1]{%
	\leavevmode\unskip\penalty9999 \hbox{}\nobreak\hfill
	\quad\hbox{#1}}
\newcommand\demo{\xqed{$\vartriangleleft$}}

\newtcolorbox{mybox}[1]{
    colback=white,
    colbacktitle=white,
    coltitle=black,
    colframe=black,
    fonttitle=\bfseries,
    enhanced,
    attach boxed title to top center={yshift=-2mm},
    title={\large#1},
    arc = 0pt,
}

\newtcolorbox{myboxi}[2]{
	colback=white,
	colbacktitle=white,
	coltitle=black,
	colframe=black,
	fonttitle=\bfseries,
	enhanced,
	attach boxed title to top center={yshift=-2mm},
	title={\large#2},
	arc = 0pt,
	#1
}

\newtcolorbox{formulaBox}[1]{
	colback=white,
	colbacktitle=white,
	coltitle=black,
	colframe=black,
	fonttitle=\bfseries,
	enhanced,
	attach boxed title to top left={yshift=-3mm, xshift=2mm},
	title={\large#1},
	arc = 0pt,
	top=0pt,
	before skip= 6pt,
	after skip=7pt,
	breakable,
}

\definecolor{pastelorange}{HTML}{F79A3D}
\definecolor{pastelblue}{HTML}{2874ae}
\definecolor{pastelgreen}{HTML}{61B940}
\definecolor{pastelviolet}{HTML}{A370AB}
\definecolor{darkblue}{HTML}{1C4987}
\definecolor{chestnut}{HTML}{A24516}
\definecolor{darkgreen}{HTML}{426A5A}

\newcommand{\EXPTIME}{\texttt{EXPTIME}}
\newcommand{\EXPSPACE}{\texttt{EXPSPACE}}

\newcommand{\PSPACE}{\texttt{PSPACE}}
\newcommand{\PTIME}{\texttt{PTIME}}
\newcommand{\NLOGSPACE}{\texttt{NLOGSPACE}}

\newcommand{\ltlnext}{\LTLnext}
\newcommand{\ltlg}{\LTLsquare}

\newcommand{\HyperATLS}{\texttt{HyperATL}$^*$}

\newcommand{\CTL}{\texttt{CTL}}
\newcommand{\CTLS}{\texttt{CTL}$^*$}

\newcommand{\LTL}{\texttt{LTL}}
\newcommand{\HyperLTL}{\texttt{HyperLTL}}

\newcommand{\ATL}{\texttt{ATL}}
\newcommand{\ATLS}{\texttt{ATL}$^*$}

\newcommand{\AHLTL}{\texttt{AHLTL}}

\newcommand{\HyperCTLS}{\texttt{HyperCTL}$^*$}

\newcommand{\out}[1]{\mathit{out}(#1)}

\newcommand{\calT}{\mathcal{T}}
\newcommand{\calG}{\mathcal{G}}
\newcommand{\calA}{\mathcal{A}}
\newcommand{\calL}{\mathcal{L}}
\newcommand{\calV}{\mathcal{V}}

\newcommand{\calX}{\mathcal{X}}

\newcommand{\nat}{\mathbb{N}}

\newcommand{\atomic}{\mathit{AP}}

\newcommand{\agents}{\Xi}
\newcommand{\agent}{\xi}

\newcommand{\refLemma}[1]{Lemma~\ref{lem:#1}}
\newcommand{\refTheo}[1]{Theorem~\ref{theo:#1}}
\newcommand{\refDef}[1]{Definition~\ref{def:#1}}
\newcommand{\refFig}[1]{Figure \ref{fig:#1}}

\newcommand{\refExample}[1]{Example~\ref{ex:#1}}
\newcommand{\refTable}[1]{Table~\ref{tab:#1}}
\newcommand{\refProp}[1]{Proposition~\ref{prop:#1}}

\newcommand{\quant}{\mathbb{Q}}

\newcommand{\ldot}{\mathpunct{.}}

\newcommand{\moves}{\mathcal{M}}

\newcommand{\myvar}[1]{{\color{chestnut!80!black} $#1$}}
\newcommand{\mycontrol}[1]{{\ttfamily#1}}
\newcommand{\myconst}[1]{{\ttfamily \color{black!70} #1}}

\newcommand{\pathsvars}{\mathcal{V}}

\newcommand{\veri}{\ensuremath{\mathfrak{V}}}
\newcommand{\refu}{\ensuremath{\mathfrak{R}}}

\begin{document}

\title[HyperATL$^\text{\scalebox{1.5}{$*$}}$: A Logic for Hyperproperties in Multi-Agent Systems]{HyperATL$^\text{\scalebox{1.7}{$*$}}$: A Logic for Hyperproperties \\in Multi-Agent Systems}

\author[R.~Beutner]{Raven Beutner\lmcsorcid{0000-0001-6234-5651}}	
\address{CISPA Helmholtz Center for Information Security, Germany}	

\author[B.~Finkbeiner]{Bernd Finkbeiner\lmcsorcid{0000-0002-4280-8441}}	

\thanks{This work was partially supported by the German Research Foundation (DFG) as part of the Collaborative Research Center ``Foundations of Perspicuous Software Systems'' (TRR 248, 389792660) and by the ERC Grants OSARES (No.~683300) and HYPER (No.~101055412). R.~Beutner carried out this work as a member of the Saarbrücken Graduate School of Computer Science.}


\begin{abstract}
	Hyperproperties are system properties that relate multiple computation paths in a system and are commonly used to, \eg, define information-flow policies.
	In this paper, we study a novel class of hyperproperties that allow reasoning about strategic abilities in multi-agent systems. 
	We introduce \HyperATLS{}, an extension of computation tree logic with path variables and strategy quantifiers. 
	Our logic supports quantification over paths in a system -- as is possible in hyperlogics such as \HyperCTLS{} -- but resolves the paths based on the strategic choices of a coalition of agents.
	This allows us to capture many previously studied (strategic) security notions in a unifying hyperlogic.
	Moreover, we show that \HyperATLS{} is particularly useful for specifying asynchronous hyperproperties, \ie, hyperproperties where the execution speed on the different computation paths depends on the choices of a scheduler.  
	We show that finite-state model checking of \HyperATLS{} is decidable and present a model checking algorithm based on alternating automata.
	We establish that our algorithm is asymptotically optimal by proving matching lower bounds.
	We have implemented a prototype model checker for a fragment of \HyperATLS{} that can check various security properties in small finite-state systems.
\end{abstract}

\maketitle

\section{Introduction}

Hyperproperties \cite{ClarksonS10} are system properties that relate multiple computation paths of a system.
Such properties are of increasing importance as they can, for example, characterize the information-flow in a system.
A prominent logic to express hyperproperties is \HyperLTL{}, which extends linear-time temporal logic (\LTL) with explicit path quantification \cite{ClarksonFKMRS14}.
\HyperLTL{} can, for instance, express generalized non-interference (GNI) \cite{McCullough88}, stating that the high-security input of a system does not influence the observable output:
\begin{align}\label{prop:GNIintro}
	\forall \pi\ldot \forall \pi'\ldot \exists \pi''\ldot\ltlg\big(\bigwedge_{a \in H} a_{\pi} \leftrightarrow a_{\pi''}\big) \land \ltlg\big(\bigwedge_{a \in O} a_{\pi'} \leftrightarrow a_{\pi''}\big)\tag{\textit{GNI}}
\end{align}
Here $H$ is a set of high-security input propositions and $O$ a set of outputs (for simplicity, we assume that there are no low-security inputs).
In our model, the system thus generates a set of computations paths over $H \cup O$, where each path corresponds to a possible input-output interaction with the system.
The \ref{prop:GNIintro} formula then states that for any pair of paths $\pi, \pi'$ there exists a third path $\pi''$ that agrees with the high-security inputs of $\pi$ and with the outputs of $\pi'$.
The existence of $\pi''$ guarantees that any observation on the outputs is compatible with every possible sequence of high-security inputs; non-determinism is the sole explanation for the output.

\begin{wrapfigure}{R}{0.55\textwidth}
	\begin{center}
		\small
		\begin{tikzpicture}[scale=1]
			
			\node[] at (1,3) (hatl) {\HyperATLS{} };

			\node[] at (0,2) (hctl) {\HyperCTLS};
			
			\node[] at (2,2) (atls) {\ATLS};
			
			\node[] at (-1,1) (hltl) {\texttt{HyperLTL}};
			
			\node[] at (1,1) (ctls) {\CTLS};
			
			\node[] at (3,1) (atl) {\texttt{ATL}};

			\draw[pastelblue, thick,dashed] (0.1, 3.4) -- node[above,sloped] {\small\color{pastelblue}Hyperproperties} (-2.6, 0.7) -- (-0.3, 0.7) --   (2.4, 3.4) -- (0.1, 3.4);

			\draw[chestnut, thick,dotted] (1.8, 3.3) -- node[above,sloped] {\small\color{chestnut}\quad Strategic Properties} (4.4, 0.7) -- (2.3, 0.7) -- (-0.3, 3.3) -- (1.8, 3.3) ;

			\node[] at (0,0) (ltl) {\texttt{LTL}};
			
			\node[] at (2,0) (ctl) {\texttt{CTL}};

			\draw[->, thick] (ltl) -- (ctls);
			\draw[->, thick] (ltl) -- (hltl);
			
			\draw[->, thick] (ctl) -- (ctls);
			\draw[->, thick] (ctl) -- (atl);
			
			\draw[->, thick] (hltl) -- (hctl);
			
			\draw[->, thick] (ctls) -- (hctl);
			\draw[->, thick] (ctls) -- (atls);
			
			\draw[->, thick] (atl) -- (atls);
			
			\draw[->, thick] (hctl) -- (hatl);
			\draw[->, thick] (atls) -- (hatl);
			
		\end{tikzpicture}
	\end{center}
	\caption{Hierarchy of expressiveness of temporal logics. An arrow $A \to B$ indicates that $A$ is a (syntactic) fragment of $B$. Logics in the blue, dashed area can express hyperproperties. Logics in the red, dotted area can express strategic properties in multi-agent systems. Logics that are interpreted on multi-agent systems (\ATL{}, \ATLS{}, and \HyperATLS{}) can also be applied to transition systems (the standard model for the remaining logics) by interpreting transition systems as 1-agent systems (see Remark \ref{rem:transVsCGS}); the reverse does not hold, \ie, logics that are interpreted on transitions systems cannot reason about strategic abilities in multi-agent systems.
	}\label{fig:expr}
\end{wrapfigure}

Existing hyperlogics (like \HyperLTL{}) consider a system as a set of paths and quantify (universally or existentially) over those paths.
In this paper, we introduce a novel class of hyperproperties that reason about \emph{strategic behavior} in a multi-agent system where the paths of the system are outcomes of games played on a game structure.
We introduce \HyperATLS{}, a temporal logic to express hyperproperties in multi-agent systems.
Our logic builds on the foundation laid by alternating-time temporal logic (\ATLS{}) \cite{AlurHK02}.\footnote{\ATLS{} is a temporal logic that extends \CTLS{} by offering selective quantification over paths that are possible outcomes of games \cite{AlurHK02}.
	The \ATLS{} quantifier $\llangle A \rrangle \varphi$ states that the players in $A$ have a joint strategy such that every outcome under that strategy satisfies $\varphi$.}
While strategy quantifiers in \ATLS{} can be nested (similar to \CTLS), the logic is unable to express hyperproperties, as the scope of each quantifier is limited to the current path.

In \HyperATLS{}, we combine quantification over strategic behavior with the ability to express hyperproperties. 
Syntactically, our logic combines the strategic quantifier of \ATLS{} but binds the outcome to a path variable:\footnote{Similar to logics such as \HyperLTL{}, we use path variables as a syntactic tool to refer to paths that are bound by outer quantifiers. For example, in the \ref{prop:GNIintro} formula, we use paths variables $\pi, \pi', \pi''$ and compare the paths bound to these variables in the body of the formula.}
The \HyperATLS{} formula $\llangle A \rrangle \pi\ldot \varphi$ specifies that the agents in $A$ have a strategy such that all outcomes under that strategy, when bound to the path variable $\pi$, satisfy $\varphi$.
Similar to \HyperLTL{}, quantification is resolved incrementally.
For example, $\llangle A \rrangle \pi\ldot  \llangle A' \rrangle \pi'\ldot \varphi$ requires the existence of strategy for the agents in $A$ such that for all possible outcomes under that strategy, when bound to $\pi$, the agents in $A'$ have a strategy such that all possible outcomes, when bounds to $\pi'$, satisfy $\varphi$.
In particular, the quantification over strategies for $A'$ takes place after $\pi$ is fixed.\footnote{This incremental elimination of quantification ensures that \HyperATLS{} is a proper extension of \HyperCTLS{} but is also crucial for decidable model checking. In fact, model checking for any logic that could express the existence of a strategy such that the set of outcomes under that strategy satisfies a hyperproperty, would already subsume the realizability problem for hyperproperties, which is known to be undecidable already for very simple fragments \cite{FinkbeinerHLST18}.}
We endow our logic with an explicit construct to resolve multiple games simultaneously (syntactically, we surround quantifiers by $[\cdot]$ brackets).
The formula $[\llangle A\rrangle \pi\ldot \llangle A' \rrangle \pi']~\varphi$ requires winning strategies for the agents in $A$ (for the first copy) and for $A'$ (for the second copy) in a game that progresses simultaneously, \ie, the players can observe the current state of both copies.

The resulting logic is very expressive and subsumes both the existing hyperlogic \HyperCTLS{} \cite{ClarksonFKMRS14} (the branching-time extension of  \HyperLTL{}) and the alternating-time logic \ATLS{} \cite{AlurHK02}.
The resulting expensiveness hierarchy is depicted in \Cref{fig:expr}.

Reasoning about strategic hyperproperties in multi-agent systems is useful in various settings, including information-flow control and asynchronous hyperproperties. 
Consider the following two examples that demonstrate how we can use \HyperATLS{} to express such strategic hyperproperties.

\paragraph{Application 1: Information-flow Control}

We consider a strategic information-flow control property.
Imagine a system where the non-determinism arises from a scheduling decision between two subprograms $P_1$ and $P_2$.
Each subprogram reads the next high-security input $h$ of the system. 
Suppose that $P_1$ assumes that $h$ is even and otherwise leaks information, while $P_2$ assumes that $h$ is odd and otherwise leaks information.
We can check \ref{prop:GNIintro} on the resulting system. 
In \HyperLTL{}, quantification is resolved incrementally, so the witness path $\pi''$ is chosen after $\pi$ and $\pi'$ are already fixed.
In particular, all future high-security input $h$ are already determined, so a leakage disproving path $\pi''$ can be constructed (by always scheduling the copy that does not leak information on the \emph{next} input); the system satisfies \ref{prop:GNIintro}.
By contrast, an \emph{actual} scheduler, who determines which subprogram handles the next input, cannot avoid information leakage.
\HyperATLS{} can express a stricter information-flow policy.
As a first step, we consider the system as a game structure with two players. 
Player $\agent_H$ chooses (in each step) the values for the high-security inputs, and player $\agent_N$ resolves the remaining non-determinism of the system (\ie, the nondeterminism not caused by input selection).
We give a concrete semantics into such a game structure in \Cref{sec:semantics}.
Consider the following \HyperATLS{} specification:
\begin{align*}
	\forall \pi\ldot\llangle \agent_N \rrangle \pi'\ldot \ltlg \big(\bigwedge_{a \in O} a_{\pi} \leftrightarrow a_{\pi'}\big)
\end{align*}
This formula requires that for every possible path $\pi$, the non-determinism player $\agent_N$ has a \emph{strategy} to produce identical outputs on $\pi'$ against all possible moves by $\agent_H$.
As we consider the system as a game structure, any strategy for $\agent_N$ does not know which moves $\agent_H$ will play in the future.
Any possible output of the system is thus achievable by a strategy only knowing the finite input history but oblivious to the future high-security inputs.  
The system sketched above does not satisfy this property.

A particular strength of this formulation is that we can encode additional requirements on the strategy for the scheduler. 
For example, if the internal non-determinism arises from the scheduling decisions between multiple components, we can, in addition, require fairness of the scheduling strategy. 
\demo

\begin{wrapfigure}{R}{0.35\textwidth}
	\vspace{-0.0cm}
	\begin{center}
		\begin{tikzpicture}
			\node[align=left] () at (0, 0) {
				\ttfamily
				\myvar{o} $\leftarrow$ \myconst{true}\\
				\mycontrol{while}(\myconst{true})\\
				\makebox[0.5cm]{} \myvar{h} $\leftarrow$ \mycontrol{Read}$_H$\\
				\makebox[0.5cm]{} \mycontrol{if} \myvar{h} \mycontrol{then}\\
				\makebox[1cm]{} \myvar{o} $\leftarrow$ $\neg$\myvar{o}\\
				\makebox[0.5cm]{} \mycontrol{else}\\
				\makebox[1cm]{} \myvar{temp} $\leftarrow$ $\neg$\myvar{o} \\
				\makebox[1cm]{} \myvar{o} $\leftarrow$ \myvar{temp}	
			};
		\end{tikzpicture}
	\end{center}
	\caption{Example program that violates (synchronous) observational determinism. }\label{fig:exProg}
\end{wrapfigure}

\paragraph{Application 2: Asynchronous Hyperproperties}

Most existing hyperlogics traverse the paths of a system synchronously.
However, in many applications (for example, when reasoning about software), we require an asynchronous traversal to, \eg,  account for the unknown speed of execution of software that runs on some unknown platform or to abstract away from intermediate (non-observable) program steps.
Strategic hyperproperties enable reasoning over asynchronous hyperproperties by considering the execution speed of a system as being controlled by a dedicated scheduling player, which we add via a system transformation.
Direct reasoning about asynchronicity is then replaced by reasoning about the strategic abilities of the scheduling player. 

As an example, consider the program in \refFig{exProg}. It continuously reads a high-security input $h$ and, depending on $h$, flips the output $o$ either directly or via a temporary variable.
Consider the \HyperLTL{} specification $\forall \pi\ldot \forall \pi'\ldot \ltlg \big( o_\pi \leftrightarrow o_{\pi'} \big)$. 
It expresses \emph{observational-determinism} (OD), \ie, it states that the output should be identical across all paths \cite{HuismanWS06}.
In \refFig{exProg}, the value of $o$ is flipped in each loop iteration, but the exact timepoint at which the flip occurs depends on the high-security input; the program does not satisfy OD in the synchronous \HyperLTL{} semantics. 
However, when executing the program, we might assume that an observer cannot detect this small timing difference and, instead, only observes the value of $o$ whenever it is changed.
We thus require a property that allows the paths of the program to be executed at different speeds to realign the output.
To reason about this in \HyperATLS{}, we extend the system with a scheduling agent $\mathit{sched}$ that can stutter the system.
That is, $\mathit{sched}$ can, in each step, decide if the system makes a step or remains in its current state (we give a concrete construction for including $\mathit{sched}$ in \Cref{sec:examples2}).
On the resulting multi-agent system (which now includes agent $\mathit{sched}$) we check the following \HyperATLS{} property:
\begin{align*}
	[\llangle \mathit{sched} \rrangle  \pi\ldot\llangle \mathit{sched} \rrangle \pi']~\mathit{fair}_{\pi} \land \mathit{fair}_{\pi'} \land \ltlg \big(o_{\pi} \leftrightarrow o_{\pi'}\big) 
\end{align*}
This formula requires that the scheduler has a strategy to align any two program paths such that the output agrees.
The additional requirement $\mathit{fair}_{\pi} \land \mathit{fair}_{\pi'} $ ensures that both copies are scheduled infinitely many times (see \Cref{sec:examples2} for details). 
By surrounding the quantifier with $[\cdot]$, both paths are resolved simultaneously (instead of incrementally), so the strategies in both copies can collaborate. 
The program in \refFig{exProg} (with an added asynchronous scheduler) satisfies this \HyperATLS{} property.
Any two paths in the original system can thus be aligned (by stuttering for any finite number of steps) such that the output agrees globally. 

This general style of asynchronous reasoning turns out to be remarkably effective:
Verification is possible (decidable) for the full logic, and the approach subsumes the largest known decidable fragment of a recent asynchronous extension of \HyperLTL{} \cite{BaumeisterCBFS21}.
\demo

\paragraph{Model Checking}

We show that model checking of \HyperATLS{} on concurrent game structures is decidable and present an automata-based model checking algorithm. 
Our algorithm incrementally reduces model checking to the emptiness of an automaton.
By using alternating automata, we encode the strategic behavior of a game structure within the transition function of the automaton.
To characterize its complexity, we partition \HyperATLS{} formulas based on the number of complex quantifiers (\ie, quantifiers where the set of agents is non-trivial) and simple quantifiers (\ie, quantifiers that reason about all possible paths, irrespective of the strategic behavior in the system).
For each fragment, we derive upper bounds on the model checking complexity, both in the size of the system and specification.
Different from \HyperLTL{} -- where each alternation results in an exponential blowup \cite{FinkbeinerRS15,Rabe16,BeutnerF23} -- the strategic quantification in \HyperATLS{} results in a \emph{double} exponential blowup with each complex quantifier.
Using a novel counter construction, we prove matching lower bounds on the \HyperATLS{} model checking problem (in the size of both specification and system).

\paragraph{Prototype Model Checker}

On the practical side, we present \texttt{hyperatlmc}, a prototype model checker for a fragment of \HyperATLS{}. 
The fragment supported by our tool does, in particular, include all alternation-free \HyperLTL{} formulas \cite{FinkbeinerRS15}, the model checking approach from \cite{CoenenFST19}, and the formulas arising when expressing asynchronous hyperproperties in \HyperATLS{}.

\paragraph{Contributions}

In summary, our contributions include the following:
\begin{itemize}
	\item We introduce \HyperATLS{}, a novel logic to express strategic hyperproperties in multi-agent systems.
	We demonstrate that \HyperATLS{} can express many existing information-flow policies and offers a natural formalism to express asynchronous hyperproperties (subsuming the largest known decidable fragment of the logic presented in \cite{BaumeisterCBFS21}).
	
	\item We give an automata-based model checking algorithm for \HyperATLS{} and analyze its complexity (both in system and specification size) based on the number and type of quantifiers. 
	
	\item We prove matching lower bounds on the \HyperATLS{} model checking problem via a novel counter construction.
	
	\item We present \texttt{hyperatlmc}, a prototype-model checker for a fragment of \HyperATLS{}, and use it to verify information-flow policies and asynchronous hyperproperties in small systems.
\end{itemize}

\noindent
This paper is an extended version of a preliminary conference version \cite{BeutnerF21}. 
Compared to the conference paper, this version contains detailed and streamlined proofs of the theoretical results.
Moreover, we extend our earlier complexity bounds by also analyzing \HyperATLS{} model-checking in the size of the system (in \cite{BeutnerF21} we only consider the complexity in the size of the specification) and derive uniform lower and upper bounds that are stated purely in the number and type of the quantifiers (without side conditions needed in \cite{BeutnerF21}).

\paragraph{Structure}

The remainder of this paper is structured as follows.
In \Cref{sec:prelim}, we introduce basic preliminaries and in \Cref{sec:hypteratl} we develop \HyperATLS{}. 
Afterward, in \Cref{sec:examples1,sec:examples2}, we present examples (ranging from information-flow policies to asynchronous hyperproperties) expressible in \HyperATLS{} and connect to existing asynchronous hyperlogics.
We discuss model checking on finite-state game structures in \Cref{sec:mc}.
In \Cref{sec:lb}, we show matching lower bounds on the model checking problem. 
Finally, in \Cref{sec:proto}, we report on our prototype model checker and discuss related work in \Cref{sec:related}.

\section{Preliminaries}\label{sec:prelim}

In this section, we introduce basic preliminaries on transition systems, game structures and alternating automata. 
For a set $X$, we write $X^*$ for the set of finite sequences over $X$, $X^+$ for the set of non-empty finite sequences, and $X^\omega$ for the set of in infinite sequences.
For an infinite sequence $u \in X^\omega$ and $i \in \nat$, we write $u(i) \in X$ for the $i$th element (starting at the $0$th) and $u[i, \infty] \in X^\omega$ for the infinite-suffix starting at position $i$.
For $u_1, \ldots, u_n \in X^\omega$, we define $\otimes(u_1, \ldots, u_n) \in (X^n)^\omega$ as the pointwise product, \ie, $\otimes(u_1, \ldots, u_n)(i) := (u_1(i), \ldots, u_n(i))$.
In case of only two sequences $u_1, u_2$, we write $u_1 \otimes u_2$ instead of $\otimes(u_1, u_2)$.
The pointwise product extends to finite sequences of the same length. 
We fix a finite set of atomic propositions $\atomic$ and define $\Sigma := 2^\atomic$. 

\paragraph{Transition Systems}

A \emph{transition system} is a tuple $\calT = (S, s_0, \delta, L)$ where $S$ is a finite set of states, $s_0 \in S$ is an initial state, $\delta \subseteq S \times S$ is a transition relation, and $L : S \to \Sigma$ is a labeling function. 
We assume that for every $s \in S$ there is at least one $s'$ such that $(s, s') \in \delta$.
A path in $\calT$ is an infinite sequence $p \in S^\omega$ such that $p(0) = s_0$ and for every $i \in \nat$, $(p(i), p(i+1)) \in \delta$.

\paragraph{Concurrent Game Structures}

As the basic model of multi-agent systems we use game structures. 
A \emph{concurrent game structure} (CGS) \cite{AlurHK02} is an extension of a transition system in which the transition relation is composed of the moves of individual agents (also called players).
Formally, a CGS is a tuple $\calG = (S, s_0, \agents, \moves, \delta, L)$.
The finite set of states $S$, the initial state $s_0 \in S$, and the labeling $L : S \to \Sigma$ are as in a transition system. 
Additionally, $\agents$ is a finite and non-empty set of agents (or players), $\moves$ is a finite and non-empty set of moves, and $\delta: S \times (\agents \to \moves) \to S$ is a transition function. 
We call a mapping $\sigma: \agents \to \moves$ a \emph{global} move vector.
Given a state and global move vector, the transition function $\delta$ determines a unique successor state.
For a set of agents $A \subseteq \agents$ we call a function $\sigma: A \to \moves$ a \emph{partial} move vector. 
For disjoint sets of agents $A_1, A_2 \subseteq \agents$ and partial move vectors $\sigma_i : A_i \to \moves$ for $i \in \{1, 2\}$ we define $\sigma_1 + \sigma_2 : A_1 \cupdot A_2 \to \moves$ as the move vector obtained as the combination of the individual choices. 
For $\sigma : A \to \moves$ and $A' \subseteq A$, we define $\sigma_{\mid A'} : A' \to \moves$ by restricting the domain of $\sigma$ to $A'$.

\begin{rem}\label{rem:transVsCGS}
	We can naturally interpret a transition systems as a $1$-player CGS in which the move of the unique player determines the successor state of the system. 
	Any temporal logic that specifies properties on game structures is thus also applicable in transition systems. 
	\demo
\end{rem}

\paragraph{Multi Stage Concurrent Game Structures}

In a concurrent game structure (as the name suggests), all agents choose their next move concurrently, \ie, without knowing what moves the other players have chosen. 
We introduce the concept of a \emph{multi-stage game structure} (MSCGS), in which the move selection proceeds in stages and agents can base their decision on the already fixed moves of (some of the) other agents. 
This is particularly useful when we, \eg, want to base a scheduling decision on the moves selected by the other agents.
Formally, a MSCGS is a CGS equipped with a function $d : \agents \to \mathbb{N}$, that orders the agents according to informedness. 
Whenever $d(\agent_1) < d(\agent_2)$, $\agent_2$ can base its next move on the move selected by $\agent_1$.
A CGS thus naturally corresponds to a MSCGS with $d = \mathbf{0}$, where $\mathbf{0}$ is the constant $0$ function.

\paragraph{Strategies in Game Structures}

A strategy in a game structure is a function that maps finite histories of plays in the game to a move in $\moves$. 
As the plays in an MSCGS progress in stages, each decision is based on the past sequence of states \emph{and} the fixed moves of all agents in previous stages.
Formally, a strategy for an agent $\agent \in \agents$ is a function 
\begin{align*}
    f_\agent : S^+ \times \big(\{\agent' \mid d(\agent') < d(\agent)\} \to \moves\big) \to \moves.
\end{align*}
Note that in case where $d = \mathbf{0}$, a strategy can be seen as a function $S^+ \to \moves$.

\begin{defi}
	Given a set of agents $A$, a set of strategies $F_A = \{f_\agent \mid \agent \in A\}$, and a state $s \in S$, we define $\out{\calG, s, F_A} \subseteq S^\omega$ as the set of all runs $u \in S^\omega$ such that 1) $u(0) = s$, and 2) for every $i \in \mathbb{N}$ there exists a global move vector $\sigma : \agents \to \moves$ with $\delta(u(i), \sigma) = u(i+1)$ and for all $\agent \in A$ we have $\sigma(\agent) = f_\agent(u[0, i], \sigma_{\mid  \{\agent' \mid d(\agent') < d(\agent)\}})$.
	\demo
\end{defi}

\paragraph{Alternating Automata}

For a set $X$, we write $\mathbb{B}^+(X)$ for the set of positive boolean formulas over $X$ with the standard propositional semantics.
Given $\Psi \in \mathbb{B}^+(X)$ and $B \subseteq X$ we write $B \models \Psi$ if the assignment obtained from $B$ by mapping all $x \in B$ to true and all $x \not\in B$ to false satisfies $\Psi$.
An \emph{alternating parity automaton} (APA) is a tuple $\calA = (Q, q_0, \Sigma, \rho, c)$ where $Q$ is a finite set of states, $q_0 \in Q$ is an initial state, $\Sigma$ is a finite alphabet, $\rho : Q \times \Sigma \to \mathbb{B}^+(Q)$ is a transition function, and $c : Q \to \mathbb{N}$ is a coloring of states.
A tree is a set $T \subseteq \mathbb{N}^*$ that is prefixed closed, \ie, $\tau \cdot n \in T$ implies $\tau \in T$.
We refer to elements in $\tau \in T$ as nodes and denote with $|\tau|$ the length of $\tau$ (or equivalently the depth of the node). 
For a node $\tau \in T$, we define $\mathit{children}(\tau)$ as the set of immediate children of $\tau$, \ie, $\mathit{children}(\tau) := \{\tau \cdot n \in T \mid n \in \mathbb{N}\}$.
An $X$-labeled tree is a pair $(T, r)$ where $T$ is a tree and $r : T \to X$ labels nodes with an element in $X$. 
A run of an APA $\calA = (Q, q_0, \Sigma, \rho, c)$ on a word $u \in \Sigma^\omega$ is a $Q$-labeled tree $(T, r)$ such that 1) $r(\epsilon) = q_0$, and 2) for all $\tau \in T$, $\{r(\tau') \mid \tau' \in \mathit{children}(\tau)\} \models \rho(r(\tau), u(|\tau|))$.
A run $(T, r)$ is accepting if, for every infinite path in $T$, the minimal color that occurs infinitely many times (as given by $c$) is even. 
We denote with $\calL(\calA)$ the set of words for which $\calA$ has an accepting run.
We call an alternating automaton $\calA$ non-deterministic (resp.~universal) if the codomain of the transition function $\rho$ consists of disjunctions (resp.~conjunction) of states.
If the codomain of $\rho$ consists of atomic formulas (\ie, formulas without boolean connectives that only consist of a single positive state atom), we call $\calA$ deterministic.\footnote{In a non-deterministic or universal automaton, we interpret $\rho$ as a function $Q \times \Sigma \to 2^Q$. In a deterministic automaton, we interpret $\rho$ as a function $Q \times \Sigma \to Q$. }
Alternating, non-deterministic, universal, and deterministic parity automata all recognize the same class of languages (namely $\omega$-regular ones) although they can be (double) exponentially more succinct.

\begin{thmC}[\cite{MiyanoH84,DrusinskyH94}]\label{theo:alt1}
    For every alternating parity automaton $\calA$ with $n$ states, there exists a non-deterministic parity automaton $\calA'$ with $2^{\mathcal{O}(n \log n)}$ states that accepts the same language.
    For every non-deterministic or universal parity automaton $\calA$ with $n$ states, there exists a deterministic parity automaton $\calA'$ with $2^{\mathcal{O}(n \log n)}$ states that accepts the same language.
\end{thmC}

\begin{thm}\label{theo:altneg}
    For every alternating parity automaton $\calA$ with $n$ states, there exists an alternating parity automaton $\overline{\calA}$ with $\mathcal{O}(n)$ states that accepts the complemented language.
    If $\calA$ is non-deterministic (resp.~universal), $\overline{\calA}$ is universal (resp.~non-deterministic).
\end{thm}

\section{HyperATL*}\label{sec:hypteratl}

In this section, we introduce \HyperATLS{}. 
Our logic extends \CTLS{} \cite{EmersonH86} by introducing path variables (similar to \HyperCTLS{} \cite{ClarksonFKMRS14}) and strategic quantification (similar to \ATLS{} \cite{AlurHK02}). 
Assume a fixed set of agents $\agents$ and let $\pathsvars$ be a finite set of path variables. 
\HyperATLS{} formulas are generated by the following grammar
\begin{align*}
     \varphi, \psi := \llangle A \rrangle \pi\ldot \varphi \mid a_\pi \mid \varphi \land \psi \mid \neg \varphi \mid \ltlnext \varphi \mid \varphi \ltlU \psi 
\end{align*}
where $\pi \in \pathsvars$ is a path variable, $a \in \atomic$ an atomic proposition, and $A \subseteq \agents$ a set of agents.
Quantification of the form $\llangle A \rrangle \pi\ldot \varphi$ binds a path to path variable $\pi$, and $a_\pi$ refers to the truth value of $a$ on the path bound to variable $\pi$. 
A formula is closed if all sub-formulas $a_\pi$ occur in the scope of a quantifier that binds $\pi$; throughout the paper we assume all formulas to be closed.
We use the usual derived boolean connectives $\lor, \to, \leftrightarrow$, the boolean constants true ($\top$) and false ($\bot$), and temporal operators \emph{eventually} ($\ltlF \varphi := \top \ltlU \varphi$), \emph{globally} ($\ltlG \varphi := \neg \ltlF \neg \varphi$), and \emph{weak until} ($\varphi \ltlW \psi := (\varphi \ltlU \psi) \lor \ltlG \varphi$).

The strategic quantifier $\llangle A \rrangle \pi\ldot \varphi$ postulates that the agents in $A$ have a joint strategy such that every outcome under that strategy, when bound to path variable $\pi$, satisfies $\varphi$.
Trivial agent sets, \ie, $A= \emptyset$ or $A = \agents$, correspond to classical existential or universal quantification.
We therefore write $\forall \pi\ldot \varphi$ instead of $\llangle \emptyset \rrangle \pi\ldot \varphi$ and $\exists \pi\ldot \varphi$ instead of $\llangle \agents \rrangle \pi\ldot \varphi$.  
For a single agent $\agent \in \agents$, we sometimes write $\llangle \agent \rrangle \pi\ldot \varphi$ instead of $\llangle \{\agent\} \rrangle \pi\ldot \varphi$.
We define $\llbracket A \rrbracket \pi\ldot \varphi := \neg \llangle A \rrangle \pi\ldot \neg \varphi$ which states that the agents in $A$ have \emph{no} strategy such that every outcome, when bound to $\pi$, avoids $\varphi$ (see \cite{AlurHK02}).
We call a quantifier $\llangle A \rrangle \pi$ \emph{simple} if the agent-set $A$ is trivial ($\emptyset$ or $\agents$) and otherwise \emph{complex}.

A \HyperATLS{} formula is \emph{linear} if it consists of an initial quantifier prefix followed by a quantifier-free formula, \ie, has the form $\quant_1 \pi_1 \ldots \quant_n \pi_n\ldot \psi$ where $\psi$ is quantifier-free and each $\quant_i$ is either $\llangle A \rrangle$ or $\llbracket A \rrbracket$ for some $A$.
Linear-\HyperATLS{} is thus the syntactic subfragment of \HyperATLS{} that is analogous to the definition of \HyperLTL{} as a syntactic fragment of \HyperCTLS{} \cite{ClarksonFKMRS14}.

\paragraph{Semantics}
The semantics of \HyperATLS{} is defined with respect to a game structure $\calG = (S, s_0, \agents, \moves, \delta, d, L)$ and a path assignment $\Pi$, which is a partial mapping $\Pi : \pathsvars \rightharpoonup S^\omega$.
For $\pi \in \pathsvars$ and path $p \in S^\omega$ we write $\Pi[\pi \mapsto p]$ for the assignment obtained by updating the value of $\pi$ to $p$.
We write $\Pi[i, \infty]$ to denote the path assignment defined by $\Pi[i, \infty](\pi) := \Pi(\pi)[i, \infty]$.
\begin{align*}
    \Pi &\models_\calG a_\pi &\text{iff  } \quad &a \in L(\Pi(\pi)(0))\\
    \Pi &\models_\calG \neg \varphi &\text{iff  }\quad &\Pi \not\models_\calG \varphi\\
    \Pi &\models_\calG \varphi \land \psi &\text{iff  } \quad &\Pi \models_\calG \varphi \text{ and } \Pi \models_\calG \psi\\
    \Pi &\models_\calG \ltlnext \varphi &\text{iff  } \quad &\Pi[1, \infty] \models_\calG \varphi\\
    \Pi &\models_\calG \varphi \ltlU \psi &\text{iff  } \quad &\exists i \geq 0\ldot  \Pi[i, \infty] \models_\calG \psi  \text{ and } \forall 0 \leq j < i\ldot \Pi[j, \infty] \models_\calG \varphi\\
    \Pi &\models_\calG \llangle A \rrangle \pi\ldot\varphi &\text{iff  } \quad &\exists F_A\ldot \forall p \in \mathit{out}(\calG, \Pi(\epsilon)(0), F_A)\ldot \Pi[\pi \mapsto p] \models_\calG \varphi
\end{align*}
Here $\Pi(\epsilon)$ refers to the path that was last added to the assignment (similar to the \HyperCTLS{} semantics \cite{ClarksonFKMRS14}).\footnote{If we assume the path variables quantified in a formula are distinct (which we can always ensure by $\alpha$-renaming), we can view a path assignment $\Pi$ as a finite list of pairs in $\calV \times S^\omega$, interpret $\Pi[\pi \mapsto p]$ as appending the pair $(\pi, p)$ to the list, and interpret $\Pi(\epsilon)$ as the path of the last pair in the list.
}
If $\Pi$ is the empty assignment, we define $\Pi(\epsilon)(0)$ as the initial state $s_0$ of $\calG$.
We say that $\calG$ satisfies $\varphi$, written $\calG \models \varphi$, if $\emptyset \models_\calG \varphi$ where $\emptyset$ is the empty path assignment. 

\begin{rem}
	Note that the games used to produce paths in a strategy quantification $\llangle A \rrangle \pi\ldot\varphi$ are \emph{local}, \ie, the outcome of the game is fixed (and bound to $\pi$) before $\varphi$ is evaluated further.
	The strategy for agents in $A$ is quantified after the outer paths (those bound to the path variables that are free in $\llangle A \rrangle \pi\ldot\varphi$) are already fixed.
	For example, in a formula of the form $\forall \pi\ldot \llangle A \rrangle \pi'\ldot \varphi$ the agents in $A$ know the already fixed path bound to $\pi$ but behave as a strategy w.r.t.~$\pi'$. 
	In particular, any formula using only simple quantification (\ie, only $\forall$ and $\exists$ quantifiers) corresponds directly to the (syntactically identical) \HyperCTLS{} property.  
	\demo
\end{rem}

\begin{prop}\label{prop:subsume}
    \HyperATLS{} subsumes \ATLS{}. When interpreting transitions systems as $1$-player CGSs (cf.~Rem.~\ref{rem:transVsCGS}), \HyperATLS{}  subsumes \HyperCTLS{} (and thus \HyperLTL).
    The resulting hierarchy is depicted in \refFig{expr}.
\end{prop}

\paragraph{Extension 1: Extended Path Quantification}

Oftentimes, it is convenient to compare different game structures with respect to a hyperproperty. 
For \emph{linear} \HyperATLS{} properties, we consider formulas with extended path quantification.
We write $\llangle A \rrangle_\calG \, \pi\ldot \varphi$ to quantify path $\pi$ via a game played in $\calG$.
For example, $\forall_{\calG}\,\pi\ldot\llangle A \rrangle_{\calG'}\,\pi'\ldot \ltlg (o_{\pi} \leftrightarrow o_{\pi'})$ states that for each path $\pi$ in $\calG$ the agents in $A$ have a strategy in $\calG'$ that produces only paths $\pi'$ which agree with $\pi$ on $o$ (where $o$ is a shared proposition between $\calG$ and $\calG'$).\footnote{Formally, we change the syntax of quantification from $\llangle A \rrangle \pi\ldot \varphi$ to $\llangle A \rrangle_\nu \pi\ldot \varphi$ where $\nu$ is a \emph{system variable}.
We evaluate the resulting formula no longer on a single  system $\calG$ but on a mapping $N$ that maps system variables to game structures.
Each quantifier $\llangle A \rrangle_\nu \pi\ldot \varphi$ is then resolved on system $N(\nu)$.
The interesting case in the semantics thus becomes
	\begin{align*}
		\Pi \models_N \llangle A \rrangle_\nu \pi\ldot\varphi \quad \text{iff  } \quad \exists F_A\ldot \forall p \in \mathit{out}(N(\nu), N(\nu)_0, F_A)\ldot \Pi[\pi \mapsto p] \models_N \varphi.
\end{align*}
Here, $N(\nu)_0$ is the initial state in game structure $N(\nu)$, and $F_A$ ranges over strategies for the agents in $A$ in game structure $N(\nu)$.
Note that this is only possible for linear formulas, as each play starts in the initial state of the game structure, irrespective of the current path assignment. 
See \cite[\S 5.4]{Rabe16} for details on the extended path quantification in the context of \HyperCTLS{}. 
}

\paragraph{Extension 2: Parallel Composition}

We extend \HyperATLS{} ~with a syntactic construct that allows multiple paths to be resolved in a single bigger game, where individual copies of the system progress in parallel.
Consider the following modification to the \HyperATLS{} ~syntax, where $k \geq 1$:
\begin{align*}
    \varphi, \psi := \big[\llangle A_1 \rrangle \pi_1\ldots\llangle A_k \rrangle \pi_k\big]~\varphi \mid a_\pi \mid \neg \varphi \mid \varphi \land \psi  \mid \ltlnext \varphi \mid \varphi \ltlU \psi
\end{align*}
When surrounding strategy quantifiers by $[\cdot]$, the resulting paths are the outcome of a game played on a bigger, parallel game of the structure. 
Consequently, the agents in each copy can base their decisions not only on the current state of their copy but on the combined state of all $k$ copies (which allows for a coordinated behavior among the copies).
For a player $\agent$ and CGS $\calG = (S, s_{0}, \agents, \moves, \delta, L)$, a \emph{$k$-fold strategy} for $\agent$ is a function $f_\agent : (S^k)^+ \to \moves$. 

\begin{defi}
	For a system $\calG$,  sets of $k$-fold strategies strategies $F_{A_1}, \ldots, F_{A_k}$ and states $s_1, \ldots, s_k$, we define $\mathit{out}(\calG, (s_1, \ldots, s_k), F_{A_1}, \ldots, F_{A_k})$ as all plays $u \in (S^k)^\omega$ such that 1) $u(0) = (s_1, \ldots, s_k)$, and 2) for every $i \in \mathbb{N}$ there exist global move vectors $\sigma_1, \ldots, \sigma_k$ such that $u(i+1) = \left(\delta(t_1, \sigma_1), \ldots, \delta(t_k, \sigma_k)\right)$ where $u(i) = (t_1, \ldots, t_k)$ and for every $j \in \{1, \ldots, k\}$, agent $\agent \in A_j$ and strategy $f_\agent \in F_{A_j}$, it holds that $\sigma_j(\agent) = f_\agent(u[0,i])$.
	\demo
\end{defi}

The definition of $k$-fold strategies and $\mathit{out}(\calG, (s_1, \ldots, s_k), F_{A_1}, \ldots, F_{A_k})$ extends naturally if we consider MSCGSs instead of CGSs.
We extend our semantics by the following judgment:
\begin{align*}
        \Pi &\models_\calG \big[\llangle A_1 \rrangle \pi_1\ldots\llangle A_k \rrangle \pi_k\big]~\varphi \quad \text{iff  } \quad \exists F_{A_1}, \ldots, F_{A_k}. \\
        &\forall (p_1, \ldots, p_k) \in \mathit{out}(\calG, (\Pi(\epsilon)(0), \ldots, \Pi(\epsilon)(0)), F_{A_1}, \ldots, F_{A_k}). \Pi[\pi_1 \mapsto p_1]\ldots[\pi_k \mapsto p_k] \models_\calG \varphi
\end{align*}%
Here we consider $\mathit{out}(\calG, (s_1, \ldots, s_k), F_{A_1}, \ldots, F_{A_k})$ as a subset of $(S^\omega)^k$ instead of $(S^k)^\omega$ using the natural correspondence (given by $\otimes^{-1}$).
Note that $[\llangle A \rrangle \pi]~\varphi$ is equivalent to $\llangle A \rrangle \pi\ldot \varphi$.

\section{Strategic Hyperproperties and Information-Flow Control}\label{sec:examples1}

Before discussing the automated verification of \HyperATLS{} properties, we consider example properties expressed in \HyperATLS{}.
We organize our examples into two categories. 
We begin with examples from information-flow control and highlight the correspondence with existing properties and security paradigms (this is done in this section). 
Afterward (in \Cref{sec:examples2}), we show that strategic hyperproperties are naturally suited to express asynchronous hyperproperties. 

\begin{rem}
	In our discussion of information-flow policies, we focus on game structures that result from reactive systems. 
	Let $H, L$, and $O$ be pairwise disjoint sets of atomic propositions denoting high-security inputs, low-security inputs, and outputs, respectively.
	We consider a system as a 3-player game structure comprising agents  $\agent_N, \agent_H$, and $\agent_L$ responsible for resolving non-determinism, selecting high-security, and selecting low-security inputs, respectively.
	In particular, the move from $\agent_H$ (resp.~$\agent_L$) determines the values of the propositions in $H$ (resp.~$L$) in the next step. Agent $\agent_N$ resolves the remaining non-determines in the system.
	We call a CGS of the above form a \emph{progCGS} (program-CGS).
	We will see a concrete transformation of programs into progCGSs in \Cref{sec:semantics}.
	For now, we rely on the reader's intuition. 
	We call a progCGS \emph{input-total} if, in each step, $\agent_H$ and $\agent_L$ can choose \emph{all} possible valuations for the input propositions in $H$ and $L$, respectively; we assume all progCGSs in this section to be input total. 
	\demo
\end{rem}

\subsection{Strategic Non-Interference}

In the introduction, we already saw that, in some cases, generalized non-interference \cite{McCullough88} is a too relaxed notion of security, as the witness path is fixed knowing the entire future input-output behavior.
Recall the definition of GNI (compared to the definition in the introduction, we now also support low-security inputs):
\begin{align}\label{prop:GNI}
	\forall \pi\ldot \forall \pi'\ldot \exists \pi''\ldot\ltlg\big(\bigwedge_{a \in H} a_{\pi} \leftrightarrow a_{\pi''}\big) \land \ltlg\big(\bigwedge_{a \in L \cup O} a_{\pi'} \leftrightarrow a_{\pi''}\big)\tag{\textit{GNI}}
\end{align}
In \HyperATLS{}, we express
\begin{align*}\label{prop:stratNI}
    \forall \pi\ldot \llangle \{\agent_N, \agent_L \} \rrangle \pi'\ldot \ltlg \big(\bigwedge_{a \in L \cup O} a_{\pi} \leftrightarrow a_{\pi'}\big). \tag{\textit{stratNI}}
\end{align*}
In this formula, a strategy for the $\agent_N$ and $\agent_L$ should construct a path $\pi'$ that agrees with the low-security inputs and outputs of $\pi$.
Note that $\agent_L$ only determines the value of the propositions in $L$, so any winning strategy for $\agent_L$ is ``deterministic'' in the sense that it needs to copy the low-security inputs from $\pi$. 
Agent $\agent_N$ needs to resolve the non-determinism but does not know the future high-security inputs on $\pi'$ (as those are chosen by $\agent_H$).
If there exists a winning \emph{strategy} for $\agent_N$ that avoids information leakage, there also exists a path (in the sense of \ref{prop:GNI}) that disproves leakage.

\begin{lem}\label{lem:stratNIandGNI}
    For any progCGS $\calG$, if $\calG \models \ref{prop:stratNI}$  then $\calG \models \ref{prop:GNI}$.
\end{lem}
\begin{proof}
    We show the contraposition and assume that $\calG \not\models \ref{prop:GNI}$.
    There thus exists paths $p_H$ and $p_{L, O}$ such that no path in $\calG$ agrees with the high inputs of $p_H$ and low-security inputs and outputs of $p_{L, O}$.
    We show that $\calG \not\models \ref{prop:stratNI}$. 
    For the universally quantified path bound to $\pi$ we choose $p_{L, O}$ and argue that $\agent_N, \agent_L$ have no winning strategy to construct $\pi'$.
    A spoiling strategy for $\agent_H$ is the one that, in each step, chooses the high-security inputs in accordance with $p_H$ (which is possible as the system is input-total).
    A winning strategy for $\agent_N$ and $\agent_L$ would then need to construct a path that agrees with $p_H$ on $H$ and with $p_{L, O}$ on $L \cup O$ which, by assumption, does not exist.  
\end{proof}

\subsection{Parallel Composition and $\forall\exists$-verification}\label{sec:coenen}

Our next examples make use of the parallel composition offered by quantification of the form $[\llangle A_1 \rrangle \pi_1\ldots\llangle A_k \rrangle \pi_k]$.
We consider the model checking algorithm for $\forall^*\exists^*$-\HyperLTL{} formulas introduced by Coenen et al.~\cite{CoenenFST19}.
The idea is to consider the verification of \HyperLTL{} formula $\forall \pi\ldot \exists \pi'\ldot \varphi$ as a game between the $\forall$-player and $\exists$-player.
The $\forall$-player moves through a copy of the state space (thereby producing a path $\pi$), and the $\exists$-player reacts with moves in a separate copy (thereby producing a path $\pi'$).
The $\exists$-player wins if $\pi'$ combined with $\pi$ satisfies $\varphi$, in which case the property holds.\footnote{The reverse implication does, in general, not hold as the $\exists$-player might require knowledge about the future behavior of the $\forall$-player. 
	 The game-based verification method can be made complete by adding \emph{prophecies}, \ie, hints for the $\exists$-player that provide limited information about the future behavior of the $\forall$-player \cite{BeutnerF22b}.} In \HyperATLS{}, we can express this game-based verification approach as a logical statement:
Formula $[\forall \pi\ldot\exists \pi']~\varphi$ requires a strategy that constructs $\pi'$ when played in parallel with a game that constructs $\pi$.
A system thus satisfies $[\forall \pi\ldot\exists \pi']~\varphi$ exactly if the property can be verified in the game-based approach from \cite{CoenenFST19}.

Phrased differently, while \cite{CoenenFST19} derives a verification \emph{method}, \HyperATLS{} can express both the original \HyperLTL{} property and its game-based verification method as \emph{formulas}; the correctness of the algorithm from \cite{CoenenFST19} becomes a logical implication in \HyperATLS{}.
\refLemma{gameBasedVeri} states a more general implication (by considering a property of the form $\forall \pi\ldot \llangle A \rrangle \pi'\ldot \varphi$ instead of $\forall \pi\ldot \exists \pi'\ldot \varphi$).

\begin{lem}\label{lem:gameBasedVeri}
    Let $\calG$ be any game structure and $\varphi$ be any \HyperATLS{} formula. 
    If $\calG \models [\forall \pi\ldot \llangle A \rrangle \pi']~\varphi$ then $\calG \models \forall \pi\ldot \llangle A \rrangle \pi'\ldot \varphi$.
\end{lem}
\begin{proof}
    Assume that $\calG \models [\forall \pi\ldot \llangle A \rrangle \pi']~\varphi$. Let $F_A = \{f_\agent : (S \times S)^+ \to \moves \mid \agent \in A\}$ be the set of strategies for the agents in $A$ that is wining, \ie, for every $(p, p') \in \out{\calG, (s_0, s_0), \emptyset, F_A}$ we have $[\pi \mapsto p, \pi' \mapsto p'] \models \varphi$.
    We show that $\calG  \models \forall \pi\ldot \llangle A \rrangle \pi'\ldot \varphi$. Let $\Pi = [\pi \mapsto p]$ be any path assignment for $\pi$ (which is universally quantified).
    We construct a winning strategy $f_\agent' : S^+ \to \moves$ for each $\agent \in A$. For $u \in S^+$ we define 
    \begin{align*}
        f_\agent'(u) := f_\agent(p[0, |u|-1] \otimes u).
    \end{align*}
    Strategy $f_\agent'$ disregards most of the already fixed path $p$ and queries $f_\agent$ on prefixes of $p$.
    Let $F_A' =  \{f_\agent' \mid \agent \in A\}$. 
    It is easy to see that for each $p' \in \mathit{out}(\calG, s_0, F_A')$ we have $p \otimes p' \in \mathit{out}(\calG, (s_0, s_0), \emptyset, F_A)$ and  so $\calG  \models \forall \pi\ldot \llangle A \rrangle \pi'\ldot \varphi$ as required.
\end{proof}

Using \refLemma{gameBasedVeri} (which generalizes easily to formulas of the form $\forall \pi_1 \ldots \forall \pi_k\ldot \llangle A \rrangle \pi'\ldot \varphi$), we can, for example, strengthen \ref{prop:GNI} by surrounding the quantifier prefix with $[\cdot]$ brackets. 
We can manually increase the lookahead to enable the strategy that is constructing path $\pi'$ to peek at future events on $\pi_1, \ldots, \pi_k$ by shifting the system.

\begin{defi}
	For a game structure $\calG = (S, s_\mathit{init}, \agents, \moves, \delta, L)$ and $n \geq 1$ we define
	\begin{align*}
		\mathit{shift}(\calG, n) := (S \cupdot \{s_0, \ldots, s_{n-1}\}, s_0, \agents, \moves, \delta', L')
	\end{align*}%
	where $s_0, \ldots, s_{n-1}$ are fresh states not already in $S$. The transition function $\delta'$ is defined by $\delta'(s, \sigma) := \delta(s, \sigma)$ for $s \in S$, $\delta'(s_i, \sigma) := s_{i+1}$ if $ i < n-1$ and $\delta'(s_{n-1}, \sigma) := s_\mathit{init}$.
	The labeling function $L'$ is defined by $L'(s) = L(s)$ if $s \in S$ and $L'(s_i) = \emptyset$.
	We define $\mathit{shift}(\calG, 0) := \calG$.
	\demo
\end{defi}

System $\mathit{shift}(\calG, n)$ shifts the behavior of $\calG$ by adding $n$ initial steps before continuing as in $\calG$. 
We express a shifted approximation of \ref{prop:GNI} as follows.
\begin{align}\label{prop:NstratGNI}
    \big[\forall_\calG \pi\ldot\forall_\calG \pi'\ldot\exists_{\mathit{shift}(\calG, n)} \pi'' \big]~\ltlg\big(\bigwedge_{a \in H} a_{\pi} \leftrightarrow \ltlN^n a_{\pi''}\big) \land \ltlg\big(\bigwedge_{a \in L \cup O}  a_{\pi'} \leftrightarrow \ltlN^n a_{\pi''}\big)\tag{\textit{aproxGNI$_n$}}
\end{align} 
We shift the behavior of the copy on which $\pi''$ is constructed by $n$ positions which we correct using $n$ $\ltlN$s in the body of the formula.
It is easy to see that if $\calG \models \textit{aproxGNI}_n$ for some $n$, then $\calG \models \ref{prop:GNI}$.\footnote{This implication is of practical relevance as \ref{prop:NstratGNI} sits at a fragment of \HyperATLS{} checkable in polynomial time (in the size of the system), whereas the $\forall^*\exists^*$-fragment of \HyperLTL{} (and thus \HyperATLS{}) is already \PSPACE-hard \cite{Rabe16}. }

\subsection{Simulation-based Non-Interference}\label{sec:simSec}

The previously discussed information-flow policies are path-based.
By contrast, \emph{simulation}-based definitions of non-interference require a lock-step security simulation that disproves leakage (see, \eg, \cite{SabelfeldS00, Sabelfeld03,MantelS10}).
Let $\calG = (S, s_0, \{\agent_N, \agent_L, \agent_H\}, \moves, \delta, L)$ be a progCGS.
For states $s, s' \in S$ and evaluations $i_L \in 2^L$ and $i_H \in 2^H$, we write $s \Rightarrow^{i_L}_{i_H} s'$ if $L(s') \cap L = i_L$ and  $L(s') \cap H = i_H$ and $s'$ is a possible successor of $s$ in $\calG$ (\ie, there exists a move vector $\sigma$ such that $\delta(s, \sigma) = s'$).

\begin{defi}\label{def:simSec}
	A \emph{security simulation} is a relation $R \subseteq S \times S$ such that whenever $(s, t) \in R$, we have 1) $s$ and $t$ agree on the output propositions, \ie, $L(s) \cap O = L(t) \cap O$, and  2) for any $i_L \in 2^L$ and $i_H, i'_H \in 2^H$ if $s \Rightarrow^{i_L}_{i_H} s'$ then there exists a $t'$ with $t \Rightarrow^{i_L}_{i'_H} t'$ and $(s', t') \in R$. 
	\demo
\end{defi}

\noindent
Note that this is not equivalent to the fact that $R$ is a simulation in the standard sense \cite{Milner80} as the second condition is asymmetric in the high-security inputs.
We call $\calG$ \emph{simulation secure} if there exists a security simulation $R$ with $(s_0, s_0) \in R$ \cite{Sabelfeld03,SabelfeldS00}.
It is easy to see that every input-total system that is \emph{simulation secure} satisfies \ref{prop:GNI}.
The converse does, in general, not hold.
In \HyperATLS{}, we can express simulation security. 
\begin{align}\label{prop:simNI}
    [\forall_\calG \, \pi\ldot \llangle \agent_N \rrangle_{\mathit{shift}(\calG, 1)} \, \pi']~\ltlg \big(\bigwedge_{a \in L } a_{\pi} \leftrightarrow \bigcirc a_{\pi'}\big) \to \ltlg \big(\bigwedge_{a \in O} a_{\pi} \leftrightarrow \bigcirc a_{\pi'}\big) \tag{\textit{simSec}}
\end{align}
Here we shift the path $\pi'$ by one position ($\mathit{shift}(\calG, 1)$), which we correct using the $\ltlN$.
The shifting allows the strategy for $\agent_N$ in the second copy to base its decision on an already fixed step in the first copy, \ie, it corresponds to a strategy with a fixed lookahead of $1$ step. 

\begin{lem}\label{lem:simSec}
    A progCGS $\calG$ is \emph{simulation secure} if and only if  $\calG \models \ref{prop:simNI}$.
\end{lem}
\begin{proof}
    We only sketch the high-level idea as the proof is similar to the well-known characterization of simulations and bisimulations as two-player games \cite{Stirling95} .
    For the first direction, we assume that $\calG$ is \emph{simulation secure} and let $R$ be a security simulation witnessing this.  
    The idea of the strategy for $\agent_N$ is to choose successors such that the parallel game between both copies is always in $R$-related states (after shifting).
    This is possible as long as the premise of \ref{prop:simNI} is not violated. 
   	By the definition of security simulations, this already implies $\ltlg (\bigwedge_{a \in O} a_{\pi} \leftrightarrow \bigcirc a_{\pi'})$.
    Note that the shifting is key, so the strategy for $\agent_N$ fixes a move \emph{after} the universally quantified path chooses a successor (as in the definition of security simulation). 
    For the second direction, assume $\agent_N$ has a winning strategy. We construct the relation $R$ by defining $(s, t) \in R$ whenever the pair $(s, t)$ occurs in the (shifted) parallel composition in any play for which the premise of \ref{prop:simNI} holds.
\end{proof}

\subsection{Non-Deducibility of Strategies}

As the last example in this section, we consider \emph{non-deducibility of strategies} ($\mathit{NDS}$) \cite{WittboldJ90}.
$\mathit{NDS}$ requires that every possible output is compatible with every possible input-\emph{strategy} (whereas \ref{prop:GNI} requires it to be compatible with every possible input-\emph{sequence}).
The subtle difference between sequences and strategies is important when a high-security input player can observe the internal state of a system. 
As a motivating example, consider the following (first introduced in \cite{WittboldJ90}):

\begin{exa}\label{ex:NDS}
    Suppose we have a system that reads a binary input $h$ from a high-security source and outputs $o$. 
    The system maintains a bit $b$ of information in its state, initially chosen non-deterministically.
    In each step, the system reads the input $h$, outputs $h \oplus b$ (where $\oplus$ is the xor-operation), non-deterministically picks a new value for $b$ and then repeats. 
    As $\oplus$ encodes an one-time pad, it is not hard to see that this system satisfies \ref{prop:GNI}:
    Given any input, any output is possible by resolving the non-deterministic choice of $b$ appropriately. 
    
    If the input player is, however, able to observe the system (in the context of \cite{WittboldJ90} the system shares the internal bit on a private channel), she can communicate an arbitrary sequence of bits to the low-security environment.
    Whenever she wants to send bit $c$, she inputs $h = c \oplus b$ where $b$ is the value of the internal bit (note that $(c \oplus b) \oplus b = c$).
    \demo
\end{exa}

Instead of requiring that every output sequence is compatible with all possible high-security input sequences, we require it to be compatible with all possible high-security input strategies. 
Phrased differently, there should \emph{not} be an output sequence such that a strategy for the input-player can \emph{avoid} this output.
\begin{align}\label{prop:NDS}
    \neg \exists \pi\ldot\llangle \agent_H \rrangle \pi'\ldot\ltlg (\bigwedge_{a \in L} a_{\pi} \leftrightarrow a_{\pi'}) \to \ltlF (\bigvee_{a \in O} a_{\pi} \not\leftrightarrow a_{\pi'}) \tag{\textit{NDS}}
\end{align}
This formula states that there does not exist a path $\pi$ such that $\agent_H$ has a strategy to avoid the output of $\pi$ (provided with the same low-security inputs).
The system sketched in \refExample{NDS} does not satisfy \ref{prop:NDS} (there, \eg, exists a strategy for $\agent_H$ that ensures that the output $o$ is always set to true).
Note that, due to determinacy of parity games, \ref{prop:NDS} is equivalent to \ref{prop:stratNI} on turn-based game structures.

\section{Strategic and Asynchronous Hyperproperties}\label{sec:examples2}

Most existing hyperlogics traverse the paths of a system synchronously.
However, especially when reasoning about software systems, one requires an asynchronous traversal to account, for example, for the unknown execution speed or to abstract away from intermediate (non-observable) program steps. 
In this section, we outline how strategic hyperproperties are useful to express such asynchronous hyperproperties.

The idea is to express asynchronous hyperproperties by viewing the stuttering of a system (\ie, whether a system progresses or remains in its current state) as being resolved by a dedicated player (which we call scheduling player).
Quantification over strategies of the scheduling player then naturally corresponds to asynchronous reasoning.  
That is, instead of reasoning about the (asynchronous) scheduling of a system directly, we reason about strategies for the scheduling agent. 
This style of asynchronous reasoning can express many properties while remaining fully decidable (as model checking of \HyperATLS{} is decidable, see \Cref{sec:mc}) and yields formulas that are automatically checkable (see \Cref{sec:proto}). 

\subsection{Scheduling Player and Stuttering Transformation  }

We call a player $\mathit{sched}$ an asynchronous scheduler if it can decide whether the system progresses (as decided by the other agents) or stutters.
Note that this differs from the asynchronous turn-based games as defined in \cite{AlurHK02}. In our setting, the scheduler does not control which of the player controls the next move but rather decides if the system as a whole progresses or stutters. 
In cases where the system does not already include an asynchronous scheduler, we can include a scheduler via a simple system transformation.

\begin{defi}\label{def:stutter}
	Given a game structure $\calG = (Q, q_0, \agents, \moves, \delta, d, L)$ over $\atomic$ and a fresh agent \emph{sched} not already included in $\agents$, define the stutter version of $\calG$, denoted $\calG_\mathit{stut}$, as the game structure over $\atomic \cupdot \{\mathit{stut}\}$ by $\calG_\mathit{stut} := (Q \times \{0, 1\}, (q_0, 0), \agents \cupdot \{\mathit{sched}\}, \moves \times \{\rightarrowtail, \downarrowtail\}, \delta', d', L')$ where 
	\begin{align*}
		\delta'\big((s, \_), \sigma\big) := \begin{cases}
			\begin{aligned}
				&(\delta(s,\mathit{proj}_1 \circ \sigma_{\mid \agents}), 0) \quad &&\text{if } (\mathit{proj}_2 \circ \sigma)(\mathit{sched}) = \rightarrowtail\\
				&(s, 1) \quad &&\text{if } (\mathit{proj}_2 \circ \sigma)(\mathit{sched}) = \downarrowtail
			\end{aligned}
		\end{cases}
	\end{align*}
	$L'((s, 0)) := L(s)$ and $L'(s, 1) := L(s) \cup \{\mathit{stut} \}$. Finally $d'(\agent) := d(\agent)$ for $\agent \in \agents$ and $d'(\mathit{sched}) := m + 1$ where $m$ is the maximal element in the codomain of $d$.
	\demo
\end{defi}

Here, $\mathit{proj}_i$ is the projection of the $i$th element in a tuple, $\circ$ denotes function composition, and $\_$ represents an arbitrary value in that position. 
In $\calG_\mathit{stut}$, the agents of the original game structure $\calG$, fix moves in $\moves$ and thereby determine the next state of the system.
In addition, the $\{\rightarrowtail, \downarrowtail\}$-decision of scheduling player determines if the move is actually executed ($\rightarrowtail$) or if the system remains in its current state ($\downarrowtail$). 
As $\mathit{sched}$ sits in the last stage of the MSCGS, the scheduling decision is based on the already fixed moves of the agents in $\agents$.
The extended state-space $Q \times \{0, 1\}$ is used to keep track of the stuttering, which becomes visible via the new atomic proposition $\mathit{stut}$.

\subsection{Observational Determinism}\label{sec:od}

As a warm-up,  we again consider the property of observational-determinism which states that the output along all paths is identical, \ie, $\forall \pi\ldot \forall \pi'\ldot\ltlg \big(\bigwedge_{a \in O} a_{\pi} \leftrightarrow a_{\pi'}\big)$.
We already argued that the example program in the introduction (in \refFig{exProg}) does not satisfy this property (if interpreted as a transition system in the natural way), as the output changes at different time points. 
To express an asynchronous version of OD, we reason about (a strategy for) the scheduling player on the transformed system.
\begin{align}\label{prop:ODa}
    [\llangle \mathit{sched} \rrangle \pi\ldot \llangle \mathit{sched}\rrangle \pi']~\mathit{fair}_{\pi} \land \mathit{fair}_{\pi'} \land \ltlg \Big(\bigwedge_{a \in O} a_{\pi} \leftrightarrow a_{\pi'}\Big)\tag{\textit{OD}$_\mathit{asynch}$}
\end{align}
where $\mathit{fair}_{\pi} := \ltlg \ltlF \neg \mathit{stut}_{\pi}$, asserts that the system may not be stuttered forever.
Note that we encapsulated the quantifiers by $[\cdot]$, thus resolving the games in parallel. 
Note that $\mathit{sched}$ only controls the stuttering and not the path of the underlying system.
The example from \refFig{exProg}, after the stuttering transformation, satisfies this formula, as the output can be aligned by the scheduling player.

\subsection{One-Sided Stuttering}

By resolving the stuttered paths incrementally (\ie, omitting the $[\cdot]$-brackets), we can also express \emph{one-sided stuttering}, \ie, allow only the second copy to be stuttered. 
As an example, assume $P^h$ is a program written in a \textbf{h}igher-level programming language and $P^l$ the complied program into a \textbf{l}ow-level language (\eg, assembly code).
Let $\calT^\mathit{h}$ and $\calT^\mathit{l}$ be transition systems of both programs, and consider the property that the low-level program exhibits the same output as the original program.
As the compiler breaks each program statement into multiple low-level instructions, the outputs will not match in a synchronous manner. 
Instead, the system $\calT^h$ may need to stutter for the low-level program to ``catch up''.
Using \refDef{stutter} we can express this as follows.
\begin{align*}
    \forall_{\calT^\mathit{l}} \, \pi\ldot \; \llangle \mathit{sched} \rrangle_{\calT^\mathit{h}_\mathit{stut}} \,\pi'\ldot \mathit{fair}_{\pi'} \land \ltlg \Big(\bigwedge_{a \in O} a_{\pi} \leftrightarrow a_{\pi'}\Big)
\end{align*}
\ie, for every execution $\pi$ of low-level program we can stutter the high-level program such that the observations align.

\subsection{Asynchronous HyperLTL}\label{sec:async}

We compare the strategic approach to asynchronicity  of \HyperATLS{} (in \Cref{sec:od}) with \emph{asynchronous HyperLTL} (\AHLTL{} for short) \cite{BaumeisterCBFS21}, a recent extension of \HyperLTL{} specifically designed to express asynchronous properties. 
\AHLTL{} is centered around the stuttering of a path.
A path $p'$ is a stuttering of $p$, written $p \trianglelefteq p'$, if it is obtained by stuttering each step in $p$ finitely often.
Formulas in \AHLTL{} quantify universally or existentially over stuttering of paths.
For example, the \AHLTL{} formula $\forall \pi_1 \ldots \forall \pi_n\ldot \mathbf{E}\ldot \varphi$ holds on a transition system $\calT$ (written $\calT \models_{\texttt{AHLTL}} \forall \pi_1 \ldots \forall \pi_n\ldot \mathbf{E}\ldot \varphi$) if for all path $p_1, \ldots, p_n$ in the $\calT$, there exists stutterings $p'_1, \ldots, p'_n$ (\ie, $p_i \triangleleft p_i'$ for all $i$) that (when bound to $\pi_1, \ldots, \pi_n$) satisfy $\varphi$.
Different from the asynchronous treatment in \HyperATLS{}, the stuttering in \AHLTL{} is thus quantified \emph{after} all paths are fixed.

Finite-state model checking of \AHLTL{} is undecidable, already for formulas of the form $\forall \pi_1\ldot \forall \pi_2\ldot\mathbf{E}\ldot \varphi$ \cite{BaumeisterCBFS21}.
The largest known fragment of \AHLTL{} with decidable model checking problem are formulas of the form $\forall \pi_1 \ldots \forall \pi_n\ldot \mathbf{E}\ldot \varphi$ where $\varphi$ is an \emph{admissible formula}.
An admissible formula has the form 
\begin{align*}
    \varphi = \varphi_{\mathit{state}} \land \big(\bigwedge_{i=1}^n \varphi_{\mathit{stut}}^i\big) \land \varphi_\mathit{phase}
\end{align*}
where $\varphi_{\mathit{state}}$ is a state-formula, \ie, uses no temporal operators, each $\varphi_{\mathit{stut}}^i$ is a stutter invariant formula that only refers to a single path variable, and $\varphi_\mathit{phase}$ is a phase formula which is a conjunction of formulas of the form $\ltlg \bigwedge_{a \in P} (a_{\pi_i} \leftrightarrow a_{\pi_j})$.
The phase formula requires the two paths $\pi_i$ and $\pi_j$ to traverse the same sequence (phases) of ``colors'' (as defined by $P$).
See \cite[\S 4.1]{BaumeisterCBFS21} for a more detailed discussion.

By replacing the stuttering quantifier $\mathbf{E}$ in \AHLTL{} with strategy quantification in \HyperATLS{} we obtain a sound approximation for formulas of the form $\forall \pi_1 \ldots \forall \pi_n\ldot \mathbf{E}\ldot \varphi$.

\begin{thm}\label{theo:altl}
    Assume a transition system $\calT$ and \AHLTL{} formula $\forall \pi_1 \ldots \forall \pi_n\ldot \mathbf{E}\ldot \varphi$. If
    \begin{align}\label{eq:thm1}
        \calT_\mathit{stut} \models [\llangle \mathit{sched} \rrangle \pi_1\ldots\llangle \mathit{sched} \rrangle \pi_n]~\varphi \land \bigwedge_{i = 1}^n \mathit{fair}_{\pi_i}
    \end{align}
    then 
    \begin{align}\label{eq:thm2}
        \calT \models_{\texttt{AHLTL}} \forall \pi_1 \ldots \forall \pi_n\ldot \mathbf{E}\ldot \varphi.
    \end{align}
    If $n = 2$ and $\varphi$ is an admissible formula, \ref{eq:thm1} and \ref{eq:thm2} are equivalent.
\end{thm}
\begin{proof}
    Let $\calT = (S, s_0, \delta, L)$.
    We first show that \ref{eq:thm1} implies \ref{eq:thm2}.
    Assume \ref{eq:thm1} and let $f_\mathit{sched}^i : (S^n)^+ \to \moves \times \{\rightarrowtail, \downarrowtail\}$ for $1 \leq i \leq n$ be a winning strategy for the scheduler.
    To show \ref{eq:thm2}, let $p_1, \ldots, p_n$ be any paths in $\calT$.
    For each $1 \leq i \leq n$ we define a stuttered version $p_i'$ such that $p_i \trianglelefteq p_i'$.
    As intermediate steps we define $p_{1, m}, \ldots, p_{n, m} \in S^m$ and $c_{1, m}, \ldots, c_{n, m} \in \nat$ for each $m \in \nat_{\geq 1}$ by recursion on $m$.
    For $m = 1$, we set $p_{i, 1} = p_i(0)$, \ie, a path of length $1$ and $c_{i, 1} = 1$ for each $1 \leq i \leq n$.
    For $m > 1$, define $b_{i, m}:= \mathit{proj}_2 \circ f^i_\mathit{sched}(\otimes(p_{1, m-1}, \ldots, p_{n, m-1}))$ for each $1 \leq i \leq n$.
    If $b_{i, m} = \rightarrowtail$, we define $c_{i, m} := c_{i, m-1} + 1$ and if $b_{i, m} = \downarrowtail$ we define $c_{i, m} := c_{i, m-1}$.
    We then set $p_{i, m} := p_{i, m} \cdot p_i(c_{i, m})$ (where $\cdot$ denotes sequence concatenation).
    Define $p_i' \in S^\omega$ as the limit of $\{p_{i, m}\}_{m \in \nat_{\geq 1}}$ (which exists as $p_{i, m}$ is a prefix of $p_{i, m+1}$ for every $m$).
    It is easy to see that $p_i \trianglelefteq p_i'$ (the fairness assumption in \ref{eq:thm1} ensures that a path is not stuttered forever).
    Moreover $[\pi_1 \mapsto p'_1, \ldots, \pi_n \mapsto p'_n] \models \varphi$ holds as  $\{f_\mathit{sched}^i\}_{i=1}^n$ is winning, and so $[\pi_1 \mapsto p_1, \ldots, \pi_n \mapsto p_n] \models_{\texttt{AHLTL}}  \mathbf{E}\ldot \varphi$ as required.
    
    For the second direction, assume that \ref{eq:thm2} holds and that $\varphi$ is admissible. 
    Let $\varphi = \varphi_{\mathit{state}} \land \big(\bigwedge_{i=1}^l \varphi_{\mathit{stut}}^i\big) \land \varphi_\mathit{phase}$.
    State formula $\varphi_{\mathit{state}}$ only refers to the initial states (as it is free of temporal operators) and $\varphi_{\mathit{stut}}^i$ is, by assumption, stutter invariant, so any fair scheduling chosen by $\mathit{sched}$ satisfies both $\varphi_{\mathit{state}}$ and $\bigwedge_{i=1}^l \varphi_{\mathit{stut}}^i$.
    Let $\varphi_\mathit{phase} = \ltlg \bigwedge_{a \in P} (a_{\pi_1} \leftrightarrow a_{\pi_2})$ be the phase formula.
    We say two states $s, s' \in S$ are \emph{about to change phase} if for some $a \in P$ we have $a \in L(s) \not\leftrightarrow a \in L(s')$.
    In this case, we write $\mathit{change}(s, s')$.
    The strategy for the scheduler has access to both the current state $s_1, s_2$ of both copies and the next state $s'_1, s'_2$ in both copies (as the scheduler sits at the last stage of $\calT_\mathit{stut}$).
    The joint strategy for the scheduler in both copies behaves as follows:
   	If $\mathit{change}(s_1, s_1') \leftrightarrow \mathit{change}(s_2, s_2')$ holds, \ie, either both or none of the copies are about to change phase, it schedules \emph{both} copies. 
    Otherwise, it only schedules the copy that is \emph{not} about to change phase.
    This ensures that phase changes occur synchronized in both copies.
    As we assumed \ref{eq:thm2}, there is a stuttering for all paths in the system such that $\varphi_\mathit{phase}$ holds, \ie, all paths traverse the same phases, albeit at possibly different speeds.
    It is easy to see that the strategy defined above creates an alignment into identical phases for any two paths of the system.
    Consequently, any play compatible with this strategy satisfies $\varphi_\mathit{phase}$ (and therefore also $\varphi$), and so \ref{eq:thm1} holds as required.
\end{proof}

\refTheo{altl} gives a sound approximation of the (undecidable) \AHLTL~model checking that is exact for admissible formulas.
As \HyperATLS{} model checking is decidable (see \Cref{sec:mc}) and the stuttering construction $\calT_\mathit{stut}$ is effectively computable, we derive an alternative proof of the decidability result from \cite{BaumeisterCBFS21}.
In summary, \HyperATLS{} subsumes the largest (known) decidable fragment of \AHLTL{}, while enjoying decidable model checking for the full logic (see the following \Cref{sec:mc}).
Moreover, the \HyperATLS{} formula \ref{eq:thm1} constructed in \refTheo{altl} falls in the fragment supported by our model checker (see \Cref{sec:proto}).

\section{Automata-Based Model Checking}\label{sec:mc}

In this section, we present an automata-based model checking algorithm for \HyperATLS{}, \ie, given a formula $\dot{\varphi}$ (we use the dot to refer to the original formula and use $\varphi, \psi$ to refer to sub-formulas of $\dot{\varphi}$) and a game structure $\calG = (S, s_0, \agents, \moves, \delta, d, L)$ we decide if $\calG \models \dot{\varphi}$.
Before discussing our verification approach, let us briefly recall \ATLS{} model checking \cite{AlurHK02} and why the approach is not applicable to \HyperATLS{}.
In \ATLS, checking if $\llangle A \rrangle \varphi$ holds in some state $s$ can be reduced to the non-emptiness check of the intersection of two tree automata. 
One accepting all possible trees that can be achieved via a strategy for players in $A$ starting in $s$, and one accepting all trees whose paths satisfy the path formula $\varphi$ \cite{AlurHK02}.
In our hyperlogic, this is not possible.
When checking $\llangle A \rrangle \pi\ldot \varphi$, we cannot construct an automaton accepting all trees that satisfy $\varphi$, as the satisfaction of $\varphi$ depends on the paths assigned to the outer path-quantifiers (which are not yet fixed).

Instead, we construct an automaton that accepts all \emph{path assignments} for the outer quantifiers for which there exists a winning strategy for the agents in $A$ (similar to the model checking approach for \HyperCTLS{} \cite{FinkbeinerRS15}).
Different from the approach for \HyperCTLS{}, we cannot resolve path quantification via an existential or universal product construction and instead encode the strategic behavior of $\calG$ within the transition function of an \emph{alternating} automaton.

In the following, we 1) define a notion of equivalence between formulas and automata and discuss the overall model checking algorithm (in \Cref{sec:equiv}), 2) give an inductive construction of an equivalent automaton (in \Cref{sec:construction}), 3) prove the construction correct (in \Cref{sec:correctness}), and 4) discuss the complexity of our algorithm (in \Cref{sec:upperBounds}).

\subsection{$\calG$-Equivalence and Model Checking Algorithm}\label{sec:equiv}

Recall that for paths $p_1, \ldots, p_n \in S^\omega$, $\otimes(p_1, \ldots, p_n) \in (S^n)^\omega$ denotes the pointwise product (also called the zipping).
Assume that some \HyperATLS{} formula $\psi$ contains free path variables $\pi_1, \ldots, \pi_n$ (in our algorithm $\psi$ is a sub-formula of $\dot{\varphi}$ that occurs under path quantifiers that bind $\pi_1, \ldots, \pi_n$).
We say that an automaton $\calA$ over alphabet $\Sigma_\psi := S^n$ is \emph{$\calG$-equivalent} to $\psi$, if for any paths $p_1, \ldots, p_n$ it holds that 
\begin{align*}
	[\pi_i \mapsto p_i]_{i=1}^n \models_\calG \psi \quad \text{iff} \quad \otimes(p_1, \ldots, p_n) \in \calL(\calA).
\end{align*}
That is, $\calA$ accepts a zipping of paths exactly if the path assignment constructed from those paths satisfies the formula; $\calA$ summarizes all path assignments for the free variables that satisfy a formula.

Now let $\dot{\varphi}$ be the formula to be checked.
Our model checking algorithm progresses in an (inductive) bottom-up manner and constructs an automaton $\calA_\psi$ that is $\calG$-equivalent for each subformula $\psi$ of $\dot{\varphi}$ (we give the construction in the next section).
Consequently, we obtain an automaton $\calA_{\dot{\varphi}}$ over alphabet $\Sigma_{\dot{\varphi}} = S^0$ that is $\calG$-equivalent to $\dot{\varphi}$.
By definition of $\calG$-equivalence, $\calA_{\dot{\varphi}}$ is non-empty iff $\emptyset \models_\calG \dot{\varphi}$ iff $\calG \models \dot{\varphi}$.
As emptiness of alternating parity automata is decidable \cite{MiyanoH84,BokerKR10} we can decide whether $\calG \models \dot{\varphi}$.

\subsection{Construction of $\calG$-Equivalent Automata}\label{sec:construction}

In the following, we give a construction of a $\calG$-equivalent automaton for each syntactic construct of \HyperATLS{}.
The most interesting case is the construction for a formula $\varphi = \llangle A \rrangle \pi\ldot \psi$ where we construct an automaton $\calA_\varphi$ over $\Sigma_\varphi = S^{n}$ from an automaton $\calA_\psi$ over $\Sigma_\psi = S^{n+1}$ by a suitable product construction with $\calG$ that takes the strategic behavior in the game structure into account.
We split the construction into the cases of logical and temporal operators (in \Cref{fig:ltltran}), simple quantification (in \Cref{fig:existsConstruct}), and complex quantification (in \Cref{fig:stratConstruct}).

\begin{figure}[!t]
	\small
	\begin{minipage}[b]{0.49\textwidth}
		\begin{myboxi}{equal height group=A}{$\varphi = a_{\pi_i}$}
			\centering
			$\calA_\varphi := (\{q_\mathit{init}\}, q_\mathit{init}, \Sigma_\varphi, \rho, \mathbf{0})$, where
			\begin{align*}
				&\rho(q_\mathit{init}, [s_1, \ldots, s_n]) := \begin{cases}
					\begin{aligned}
						\;&\top \quad &&\text{if } a \in L(s_i)\\
						&\bot \quad &&\text{if } a \not\in L(s_i)
					\end{aligned}
				\end{cases}
			\end{align*}
		\end{myboxi}
	\end{minipage}
	\hfill
	\begin{minipage}[b]{0.49\textwidth}
		\begin{myboxi}{equal height group=A}{$\varphi = \neg \psi_1$}
			\centering
			$\calA_\varphi := \overline{\calA}_{\psi_1}$ where $\overline{\calA}_{\psi_1}$ is obtained by \refTheo{altneg}.
		\end{myboxi}
	\end{minipage}
	
	\begin{mybox}{$\varphi = \psi_1 \land \psi_2$}
		\centering
		$\calA_\varphi := (Q_1 \cupdot Q_2 \cupdot \{q_\mathit{init}\}, q_\mathit{init}, \Sigma_\varphi, \rho, c_1 \cupdot c_2 \cupdot [q_\mathit{init} \mapsto 0])$, where
		\begin{align*}
			\rho(q, [s_1, \ldots, s_n]) := \begin{cases}
				\begin{aligned}
					\;&\rho_1(q_{0, 1}, [s_1, \ldots, s_n]) \land \rho_2(q_{0, 2}, [s_1, \ldots, s_n]) \quad &&\text{if } q = q_\mathit{init}\\
					&\rho_i(q, [s_1, \ldots, s_n]) \quad &&\text{if } q \in Q_i
				\end{aligned}
			\end{cases}
		\end{align*}
	\end{mybox}

	\begin{mybox}{$\varphi = \ltlN \psi_1$}
		\centering
		$\calA_\varphi := (Q_1 \cupdot\{q_\mathit{init}\}, q_\mathit{init}, \Sigma_\varphi, \rho, c_1 \cupdot [q_\mathit{init} \mapsto 0])$, where
		\begin{align*}
			\rho(q, [s_1, \ldots, s_n]) := \begin{cases}
				\begin{aligned}
					\;&q_{0, 1} \quad &&\text{if } q = q_\mathit{init}\\
					&\rho_1(q, [s_1, \ldots, s_n]) \quad &&\text{if } q  \in Q_1
				\end{aligned}
			\end{cases}
		\end{align*}
	\end{mybox}

	\begin{mybox}{$\varphi = \psi_1 \ltlU \psi_2$}
		\centering
		$\calA_\varphi := (Q_1 \cupdot Q_2 \cupdot \{q_\mathit{init}\}, q_\mathit{init}, \Sigma_\varphi, \rho, c_1 \cupdot c_2 \cupdot [q_\mathit{init} \mapsto 1])$, where
		\begin{align*}
			\rho(q, [s_1, \ldots, s_n]) := \begin{cases}
				\begin{aligned}
					\;&\rho_2(q_{0, 2}, [s_1, \ldots, s_n]) \lor\big(\rho_1(q_{0, 1}, [s_1, \ldots, s_n]) \land q_\mathit{init}\big) \quad &&\text{if } q = q_\mathit{init}\\
					&\rho_i(q, [s_1, \ldots, s_n]) \quad &&\text{if } q \in Q_i
				\end{aligned}
			\end{cases}
		\end{align*}
	\end{mybox}
	
	\caption{Automaton construction for boolean and temporal operators. 
		Here, $\calA_{\psi_i} = (Q_i, q_{0, i}, \Sigma_{\psi_i}, \rho_i, c_i)$ for $i \in \{1, 2\}$ are inductively constructed alternating automata for sub-formulas $\psi_1$ and $\psi_2$.
		We assume that $Q_1$ and $Q_2$ are disjoint sets of states and $q_\mathit{init}$ is a fresh state.
		For two colorings $c_1 : Q_1 \to \nat$ and $c_2 : Q_2 \to \nat$, $c_1 \cupdot c_2 : Q_1 \cupdot Q_2 \to \nat$ denotes the combined coloring.
		We write $[q \mapsto n]$ for the function $\{q\} \to \nat$ that maps $q$ to $n$. }\label{fig:ltltran}
\end{figure}

\paragraph{Boolean And Temporal Operators}

For the boolean combinators and temporal operators, our construction follows the standard translation from \LTL{} to alternating automata (see, \eg, \cite{MullerSS88} or \cite{FinkbeinerRS15} for details).
We give the construction in \refFig{ltltran}.

\paragraph{Simple Quantification}

We now consider the case where $\varphi = \llangle A \rrangle \pi\ldot \psi$ and focus on the case where the quantifier is \emph{simple}.
Assume that $A = \agents$, \ie, $\varphi = \exists \pi\ldot \psi$.
The construction of $\calA_\varphi$ is similar to the one in \cite{FinkbeinerRS15,BeutnerF23} by building a product of $\calA_\psi$ and $\calG$.
We give the construction in \refFig{existsConstruct}.
Here $\calA_\psi^{\mathit{ndet}}$ is a \emph{non-deterministic} automaton equivalent to $\calA_\psi$, which we can obtain (with an exponential blowup) via \refTheo{alt1}.
The automaton $\calA_\varphi$ guesses a path in $\calG$ and tracks the acceptance of $\calA_\psi$ on this path combined with the input word over $S^n$, \ie, every accepting run of $\calA_\varphi$ on $\otimes(p_1, \ldots, p_n)$ guesses a path $p$ in $\calG$ such that $\calA_\psi$ accepts $\otimes(p_1, \ldots, p_n, p)$.
Note that in this case $\calA_\varphi$, is again a non-deterministic automaton.
The case where $A = \emptyset$ (\ie, $\varphi = \forall \pi\ldot \psi$) can be handled using complementation:
As $\forall \pi\ldot \varphi \equiv \neg \exists \pi\ldot \neg \varphi$ we can combine the construction for existential quantification (in \refFig{existsConstruct}) with the construction for negation in \Cref{fig:ltltran}.
Importantly, in cases where $A = \emptyset$, the automaton $\calA_\varphi$ is universal.

\begin{figure}
	\begin{mybox}{$\varphi = \exists \pi\ldot \psi$}
		\small
		\begin{align*}
			\calA_\varphi := (S \times Q \cupdot \{q_{\mathit{init}}\}, q_{\mathit{init}}, \Sigma_\varphi, \rho', c')
		\end{align*}
		where $c'(s, q) := c(q)$ and $c'(q_\mathit{init})$ can be chosen arbitrarily. The nondeterministic transition function $\rho' : \big(S \times Q \cupdot \{q_{\mathit{init}}\}\big) \times \Sigma_\varphi \to 2^{S \times Q \cupdot \{q_{\mathit{init}}\}}$ is defined via
		\begin{align*}
			&\rho'(q_{\mathit{init}}, [s_1, \ldots, s_n]) := \\
			&\quad\quad\quad\{(s', q') \mid q' \in \rho(q_0, [s_1, \ldots, s_n, s_n^\circ]) \land \exists \sigma : \agents \to \moves\ldot \delta(s_n^\circ, \sigma) = s'\} \\[0.3cm]
			&\rho'((s, q), [s_1, \ldots, s_n]) := \\
			&\quad\quad\quad\{(s', q') \mid q' \in \rho(q, [s_1, \ldots, s_n, s]) \land \exists \sigma : \agents \to \moves\ldot \delta(s, \sigma) = s'\}
		\end{align*}
		where we define $s_n^\circ := s_n$ if $n \geq 1$ and otherwise $s_n^\circ := s_0$ (where $s_0$ is the initial state of $\calG$).
	\end{mybox}
	
	\caption{Construction of a $\calG$-equivalent automaton for $\varphi = \exists \pi \ldot \psi$. Here, $\calA_\psi^{\mathit{ndet}} = (Q, q_0, \Sigma_\psi, \rho, c)$ with $\rho : Q \times \Sigma_\psi \to 2^Q$ is a non-deterministic automaton that is equivalent to the inductively constructed automaton $\calA_\psi$ for $\psi$.}\label{fig:existsConstruct}
\end{figure}

\begin{figure}
    \begin{mybox}{$\varphi = \llangle A \rrangle \pi\ldot \psi$}
    	\small
        \begin{align*}
            \calA_\varphi := (S \times Q \cupdot \{q_{\mathit{init}}\}, q_{\mathit{init}}, \Sigma_\varphi, \rho', c')
        \end{align*}
        where $c'(s, q) := c(q)$ and $c'(q_\mathit{init})$ can be chosen arbitrarily. The transition function $\rho' : \big(S \times Q \cupdot \{q_{\mathit{init}}\}\big) \times \Sigma_\varphi \to \mathbb{B}^+\big(S \times Q \cupdot \{q_{\mathit{init}}\}\big)$ is defined via
        \begin{align*}
        	&\rho'\big(q_\mathit{init}, [s_1, \ldots, s_n]\big) := \\
        	&\quad\quad\bigvee\limits_{\sigma_0 : A_0\to \moves} \bigwedge\limits_{\sigma'_0 : \overline{A}_0 \to \moves} \ldots \bigvee\limits_{\sigma_m : A_m \to \moves} \bigwedge\limits_{\sigma'_m : \overline{A}_m \to \moves}   
        	\bigg(\delta\Big(s^\circ_n, \sum_{i=0}^{m} (\sigma_i  + \sigma'_i)\Big), \rho\big(q_0, [s_1, \ldots, s_n, s^\circ_n]\big)   \bigg)\\[0.3cm]
        	&\rho'\big((s, q), [s_1, \ldots, s_n]\big) := \\
        	&\quad\quad\bigvee\limits_{\sigma_0: A_0 \to \moves} \bigwedge\limits_{\sigma'_0 : \overline{A}_0 \to \moves} \ldots \bigvee\limits_{\sigma_m : A_m \to \moves} \bigwedge\limits_{\sigma'_m : \overline{A}_m \to \moves}   
        	\bigg(\delta\Big(s, \sum_{i=0}^{m} (\sigma_i  + \sigma'_i)\Big), \rho\big(q, [s_1, \ldots, s_n, s]\big)   \bigg)
        \end{align*}%
        where we define $s_n^\circ := s_n$ if $n \geq 1$ and otherwise $s_n^\circ := s_0$ (where $s_0$ is the initial state of $\calG$).
        The sets $A_i$ and $\overline{A}_i$ are defined as $A_i := A \cap d^{-1}(i)$ and $\overline{A}_i := (\agents \setminus A) \cap d^{-1}(i)$ and $m$ is the maximal element in the codomain of $d$.
    \end{mybox}
    
    \caption{Construction of a $\calG$-equivalent automaton for $\varphi = \llangle A \rrangle \pi \ldot \psi$. Here, $\calA^\mathit{det}_\psi = (Q, q_0, \Sigma_\psi, \rho, c)$ with $\rho : Q \times \Sigma_\psi \to Q$ is a deterministic automaton that is equivalent to the inductive constructed alternating automaton $\calA_\psi$ for $\psi$.
 }\label{fig:stratConstruct}
\end{figure}

\paragraph{Strategic Quantification }

Lastly, we consider the case of (proper) complex quantification, \ie, the case where $\varphi = \llangle A \rrangle \pi\ldot \psi$ and $A \neq \emptyset$ and $A \neq \agents$.\footnote{Note that the construction in \refFig{stratConstruct} subsumes the construction in \refFig{existsConstruct}. 
	We give an explicit construction for the case of simple quantification (in \refFig{existsConstruct}) as the resulting automaton is exponentially smaller, giving tight complexity results (see \Cref{sec:upperBounds}). }
In our construction, $\calA_\varphi$ encodes the strategic behavior of the agents in $A$.
We achieve this by encoding the strategic play of the game structure within the transition function of $\calA_\varphi$.
We give the construction in \refFig{stratConstruct}.
Here $\calA^\mathit{det}_\psi$ is a \emph{deterministic} automaton equivalent to $\calA_\psi$ which we obtain (with a double exponential blowup) via \refTheo{alt1}.
The transition function encodes the strategic behavior by disjunctively choosing moves for players in $A$, followed by a conjunctive treatment of all adversarial players. 
The stages of $\calG$ naturally correspond to the order of the move selection (as captured in the sets $A_i$ and $\overline{A}_i$), giving an alternating sequence of disjunctions and conjunctions. 
In the case where the MSCGS $\calG$ is a CGS, \ie, $d = \mathbf{0}$, the transition function has the form of a (positive) DNF (a boolean formula of the form $\bigvee \bigwedge$), where the moves of agents in $A$ are considered disjunctively and the moves by all other agents conjunctively.  
Our construction can be extended to handle formulas of the form $[\llangle A_1 \rrangle \pi_1\ldots\llangle A_k \rrangle \pi_k]~\psi$ by joining the stage across $k$ copies of the game structure.
For the dual strategic quantifier $ \llbracket A \rrbracket \pi\ldot \psi $ we can again make use of the fact that $\llbracket A \rrbracket \pi\ldot \psi \equiv \neg \llangle A \rrangle\ldot \neg \psi$ and combine the construction in \refFig{stratConstruct} with that for negation in \Cref{fig:ltltran}.

\subsection{Correctness}\label{sec:correctness}

The correctness of our construction in \Cref{sec:construction} is stated by the following proposition.

\begin{restatable}{prop}{equiv}\label{prop:MCequiv}
	For any \HyperATLS{} formula $\varphi$, $\calA_\varphi$ is $\calG$-equivalent to $\varphi$.
\end{restatable}

The proof of \refProp{MCequiv} goes by induction on $\varphi$ following the construction of $\calA_\varphi$.
For the logical and temporal connectives (in \Cref{fig:ltltran}) and pure existential quantification (in \Cref{fig:existsConstruct}), the statement is obvious (see, \eg, \cite{MullerSS88} for the logical and temporal connectives and \cite{FinkbeinerRS15} for the case of simple quantification).
A proof for the case where $\varphi = \llangle A \rrangle \pi\ldot \psi$  (in \Cref{fig:stratConstruct}) can be found in \Cref{app:proofEquiv}.

\subsection{Complexity Upper Bounds}\label{sec:upperBounds}

The complexity of our model checking algorithm hinges on the size of the automaton $\calA_{\dot{\varphi}}$.
The constructions in \Cref{fig:ltltran} only increase the size of the automaton by a polynomial amount, but the constructions for path quantification in \Cref{fig:existsConstruct,fig:stratConstruct} increase the number of states exponentially. 
We observe a difference in blowup between simple quantification and complex quantification.
The former requires (in general) a conversion of the alternating automaton $\calA_\psi$ to a non-deterministic automaton ($\calA_\psi^\mathit{ndet}$ in \Cref{fig:existsConstruct}) causing an exponential blowup, whereas the latter requires a full determinization ($\calA_\psi^\mathit{det}$ in \Cref{fig:stratConstruct}) causing a \emph{double} exponential blowup.

To capture the size of the automaton and the resulting complexity of our algorithm, we define $\mathit{Tower}_c(k, n)$ as a tower of $k$ exponents (with base $c$), \ie,
\begin{align*}
	\mathit{Tower}_{c}(0, n) &:= n\\
	\mathit{Tower}_{c}(k+1, n) &:= c^{\mathit{Tower}_{c}(k, n)}
\end{align*}
For $k \geq 1$, we define $k$-\EXPSPACE{} as the class of languages recognized by a deterministic (or, due to Savitch's theorem \cite{Savitch70}, equivalently, non-deterministic) Turing machine (TM) with space $\mathit{Tower}_{c}(k, n)$ for some fixed $c \in \nat$.
Analogously, we define $k$-\EXPTIME{} as the class of languages recognizable by a \emph{deterministic} TM in time $\mathit{Tower}_{c}(k, n)$ for some fixed $c$.
We define $0$-\EXPTIME{} := \PTIME{} and $0$-\EXPSPACE{} := \PSPACE{}. 
For $k < 0$, we define $k$-\EXPSPACE{} := \NLOGSPACE{}.

\paragraph{Complexity Based on Prefix-Cost}

\begin{figure}
	
	\begin{minipage}{0.58\textwidth}
		\begin{subfigure}{\linewidth}
			\begin{align*}
				d_{\mathit{spec}}(\exists \pi\ldot \psi) &:= 1\\
				d_{\mathit{spec}}(\forall \pi\ldot  \psi) &:= 1\\
				d_{\mathit{spec}}(\llangle A \rrangle\pi\ldot \psi) &:= 2\\
				d_{\mathit{spec}}(\llbracket A \rrbracket\pi\ldot \psi) &:= 2\\
				d_{\mathit{spec}}(\quant_1 \, \pi\ldot \quant_2\, \pi'\ldot \varphi) &:= d_{\mathit{spec}}(\quant_2\,\pi'\ldot \varphi) + q(\quant_1, \quant_2) 
			\end{align*}
			\subcaption{Prefix-cost w.r.t.~the size of the \emph{specification}. Here, $\psi$ is the quantifier-free body of the formula. The definition of the cost function $q$ is given in \Cref{fig:qunatifierCost}.}\label{fig:mesSpec}
		\end{subfigure}
		\begin{subfigure}{\linewidth}
			\begin{align*}
				d_{\mathit{sys}}(\quant\,\pi\ldot \psi) &:= 0\\
				d_{\mathit{sys}}(\quant_1 \, \pi\ldot \quant_2\, \pi'\ldot \varphi) &:= d_{\mathit{sys}}(\quant_2\,\pi'\ldot \varphi) + q(\quant_1, \quant_2) 
			\end{align*}
			\subcaption{Prefix-cost w.r.t.~the size of the \emph{system}. Here, $\psi$ is the quantifier-free body of the formula. The definition of the cost function $q$ is given in \Cref{fig:qunatifierCost}.}\label{fig:mesSys}
		\end{subfigure}
	\end{minipage}%
	\hfill
	\begin{minipage}{0.38\textwidth}
		\begin{subfigure}{\linewidth}
			\centering 
			
			\def\arraystretch{1.5}
			\begin{tabular}{|c|c|c|c|c|}
				\cline{2-5}
				\multicolumn{1}{c|}{}& $\exists$  & $\forall$ & $\llangle A \rrangle$  & $\llbracket A \rrbracket$ \\
				\hline
				$\exists$ & $0$ & $1$  & $1$ & $1$  \\
				\hline
				$\forall$ & $1$  & $0$  & $1$ & $1$ \\
				\hline
				$\llangle A \rrangle$ & $1$  & $1$  & $2$ & $2$  \\
				\hline
				$\llbracket A \rrbracket$ & $1$  & $1$  & $2$ & $2$ \\
				\hline
			\end{tabular}
			
			\subcaption{Cost associated with each quantifier alteration. Given $\quant_1, \quant_2 \in \{\exists, \forall, \llangle A \rrangle, \llbracket A \rrbracket\}$, the cost $q(\quant_1, \quant_2)$ is given in column $\quant_1$, row $\quant_2$. Note that the cost matrix symmetric.}\label{fig:qunatifierCost}
		\end{subfigure}
	\end{minipage}

	\caption{Definition of prefix-cost w.r.t.~the size of the specification and the size of the system. Here $\llangle A \rrangle$ and $\llbracket A \rrbracket$ represent a \emph{proper} complex quantifier, \ie, a quantifier where $A \neq \emptyset$ and $A \neq \agents$. }
\end{figure}

We characterize the cost of our algorithm for \HyperATLS{} model checking.
We distinguish between the complexity in the size of the specification (the length of the formula, \ie, the number of nodes in the AST) and the size of the system (the number of states).
Our complexity analysis is parametric in the structure of the quantifier prefix of a formula.\footnote{
If we consider arbitrary \HyperATLS{} formulas, model checking is non-elementary in both the size of the specification and the size of the system (as \HyperATLS{} subsumes \HyperLTL{}; see \cite{Rabe16} for details).
If we instead consider formulas with a fixed quantifier structure, we can derive elementary complexity bounds (in terms of specification size and system size) for all formulas with the fixed quantifier structure.
}
We focus our discussion on \HyperATLS{} formulas that are \emph{linear}.\footnote{To stay as flexible as possible, we include $\llbracket A \rrbracket$ as a first-class quantifier (instead of a derived one).
	Focusing on linear formulas allows for a simpler characterization of the model checking complexity.
	We discuss the case of non-linear formulas later in Remark \ref{rem:nonlinear}.
}
To differentiate formulas based on the structure of their quantifier prefix, we assign each formula $\varphi$ two quantities:
the specification-based prefix-costs $d_{\mathit{spec}}(\varphi)$ and system-based prefix-costs $d_{\mathit{sys}}(\varphi)$.
Both are defined inductively in \Cref{fig:mesSpec,fig:mesSys}. 
Our measures generalize the \emph{alteration-depth} of a \HyperLTL{} formula.\footnote{The alteration-depth denotes the number of quantifier alternations (between $\exists$ and $\forall$) in the quantifier prefix and is used to characterize the \HyperLTL{} model checking complexity \cite{FinkbeinerRS15,Rabe16}.
	For any (linear) \HyperATLS{} formula $\varphi$ that uses only simple quantification (so $\varphi$ is also a \HyperLTL{} formula), $d_{\mathit{sys}}(\varphi)$ gives a (in some cases tight) upper bound the alternation-depth of $\varphi$.
	Note that in this case, $d_{\mathit{spec}}(\varphi) = d_{\mathit{sys}}(\varphi) + 1$.}

\begin{thm}\label{theo:upb}
	Model checking of a linear \HyperATLS{} formula $\varphi$ is
	\begin{enumerate}
		\item in $d_{\mathit{spec}}(\varphi)$-\EXPTIME{} in the size of the specification, and
		\item in $d_{\mathit{sys}}(\varphi)$-\EXPTIME{} in the size of the system
	\end{enumerate}
\end{thm}
\begin{proof}
	Let formula $\varphi$ and system $\calG$ be given. Abbreviate $d_{\mathit{spec}}:= d_{\mathit{spec}}(\varphi)$ and $d_{\mathit{sys}}:= d_{\mathit{sys}}(\varphi)$.
	Using the constructions in \Cref{fig:existsConstruct,fig:stratConstruct}, we observe that the size of $\calA_\varphi$ is at most $\mathit{Tower}_c(d_{\mathit{spec}}, |\varphi|)$ and $\mathit{Tower}_c(d_{\mathit{sys}}, |\calG|)$ for some base $c$ that depends only on  $d_{\mathit{spec}}$ and $d_{\mathit{sys}}$.
	To argue this, we consider all variations of consecutive types of quantifiers and their cost assigned in \Cref{fig:qunatifierCost}. 
	For example, for a formula $\exists \pi\ldot \llangle A \rrangle \pi'\ldot \varphi$, the alternating automaton $\calA_{\llangle A \rrangle \pi'\ldot \varphi}$ needs to be translated to a non-deterministic automaton when performing the construction in \Cref{fig:existsConstruct}, which incurs a single exponential blowup.
	For a formula $\llangle A \rrangle \pi\ldot \llangle A' \rrangle \pi'\ldot \varphi$, automaton $\calA_{\llangle A' \rrangle \pi'\ldot \varphi}$ needs to be determinized, causing a double exponential blowup. 
	For a formula $\llangle A \rrangle \pi\ldot \exists \pi'\ldot \varphi$, the automaton $\calA_{\exists \pi'\ldot \varphi}$ is already nondeterministic, so we can perform a determinization with a single exponential blowup. 
	The cost for all possible combinations of two consecutive quantifier types match with the cost given in \Cref{fig:qunatifierCost}.
	After eliminating all quantifiers, we end up with a alternating automaton over the \emph{singleton alphabet}, for which we can decide emptiness in polynomial time; the bounds follow.
	Note that $d_{\mathit{spec}}(\varphi)$ and $d_{\mathit{sys}}(\varphi)$ only differ in the cost associated with the innermost quantifier. 
	The automaton construction for this last quantifier is linear in the size of the system but exponential (in case the innermost quantifier is simple) or double exponential (in case the innermost quantifier is complex) in the size of the formula. 
\end{proof}

\begin{rem}
	As complex quantification requires a full determinization, simple quantification between complex quantifiers can, in some cases, have no impact on the complexity. 
	For example, for formulas of the form $\llangle A \rrangle \pi\ldot \exists \pi'\ldot \llangle A' \rrangle \pi''. \psi$ and $\llangle A \rrangle \pi\ldot \llangle A' \rrangle \pi'\ldot \psi$ our model checking algorithm follows the same asymptotic complexity. 
	\demo
\end{rem}

\paragraph{Complexity Based on the Quantifier Type}

We obtain a simpler complexity characterization if we only consider the number and type of each quantifier (\ie, we ignore the order in which they occur).
We again consider the complexity in the size of the specification (in \refProp{upperBound}) and the size of the system (in \refProp{upperBoundModel}).

\begin{prop}\label{prop:upperBound}
	Model checking of a linear \HyperATLS{} formula with $k$ complex and $l$ simple quantifiers is
	\begin{enumerate}
		\item in $(2k + l)$-\EXPTIME{}, and \label{enum:ubf1}
		\item in $(2k+l-1)$-\EXPSPACE{} if $l \geq 1$ \label{enum:ubf2}
	\end{enumerate}
	when measured in the size of the specification.
\end{prop}
\begin{proof}
	Let $\varphi$ be any formula using $k$ complex and $l$ simple quantifiers, and let $\calG$ be the game structure.
	Point (\ref{enum:ubf1}) follows directly from \refTheo{upb}, as we can easily check that $2k + l \geq d_{\mathit{spec}}(\varphi)$.
	For point (\ref{enum:ubf2}) let $\varphi = \quant_1\,\pi_1 \ldots \quant_{k+l}\, \pi_{k+l}\ldot \varphi$ where $\psi$ is quantifier-free.
	Let $i$ be the smallest index such that $\quant_i$ is simple (which exists as $l \geq 1$).
	In case $i = 1$ (\ie, the outermost quantifier is simple) we may assume w.l.o.g~that $\quant_1 = \exists$ (as otherwise, we check the negated formula).
	So $\calA_\varphi$ is a non-deterministic automaton, and the $(2k+l-1)$-\EXPSPACE{} upper bound follows, as we can check the emptiness of a non-deterministic automaton in \NLOGSPACE{} \cite{VardiW94}.
	If $l > 1$, it is easy to see that $2k + l - 1\geq d_{\mathit{spec}}(\varphi)$, so we get an even better upper bound of $(2k+l-1)$-\EXPTIME{} via \refTheo{upb} and thereby also the desired $(2k+l-1)$-\EXPSPACE{} bound.
\end{proof}

\begin{prop}\label{prop:upperBoundModel}
	Model checking of a linear \HyperATLS{} formula with $k$ complex and $l$ simple quantifiers is 
	\begin{enumerate}
		\item in $(2k+l-2)$-\EXPSPACE{}, and \label{enum:ubm1}
		\item in $(2k+l-2)$-\EXPTIME{} if $k \geq 1$, and \label{enum:ubm2}
		\item in $(2k+l-3)$-\EXPSPACE{} if $k \geq 1$ and the outermost quantifier is simple (so necessarily $l \geq 1$) \label{enum:ubm4}
	\end{enumerate} 
	when measured in the size of the system.
\end{prop}
\begin{proof}
	Let $\varphi$ be any formula $k$ complex and $l$ simple quantifiers and $\calG$ the game structure.
	We begin with point (\ref{enum:ubm2}). It is easy to see that if $k \geq 1$, then $2k + l - 2 \geq d_{\mathit{sys}}(\varphi)$, so (\ref{enum:ubm2}) follows from \refTheo{upb}.
	For point (\ref{enum:ubm4}) we may assume that the outermost quantifier is existential (otherwise, we check the negated formula). 
	We observe that if $k \geq 1$, the size of $\calA_\varphi$ is at most $\mathit{Tower}_c(2k+l-2, |\calG|)$ (as in the proof for point (\ref{enum:ubm2})) and, additionally, $\calA_\varphi$ is non-deterministic, so emptiness can be checked in \NLOGSPACE{}, giving the desired $(2k+l-3)$-\EXPSPACE{} bound.
	It remains to show point (\ref{enum:ubm1}). 
	If $k \geq 1$, we get the even better bound of $(2k+l-2)$-\EXPTIME{} from point (\ref{enum:ubm2}).
	In case $k = 0$ (so all quantifiers are simple), we again assume that the outermost quantifier is existential.
	It is easy to see that the size of $\calA_\varphi$ is at most $\mathit{Tower}_c(2k+l-1, |\calG|)$ and non-deterministic, so the bound follows from the \NLOGSPACE{} emptiness check.  
	Note that the case where $k = 0$ corresponds directly to \HyperLTL{} model checking \cite{FinkbeinerRS15,Rabe16,BeutnerF23}.
\end{proof}

\begin{rem}\label{rem:nonlinear}
	So far, we have focused on linear \HyperATLS{} formulas as this allows for a precise yet succinct analysis. 
	Analogous to point (\ref{enum:ubf1}) of \refProp{upperBound} we can easily see that we can check an arbitrary (possibly non-linear) formula $\varphi$ with $k$ complex and $l$ simple quantifiers in $(2k+l)$-\EXPTIME{} in the size of the specification and $(2k+l-1)$-\EXPTIME{} in the size of the system (by simply analyzing the size of $\calA_\varphi$).
	Deriving more precise bounds by generalizing the specification-based and system-based prefix-cost to non-linear formulas is challenging. 
	\demo
\end{rem}

\vspace{-1mm}
\section{Lower Bounds for Model Checking}\label{sec:lb}

In this section, we establish lower bounds on the \HyperATLS{} model checking problem.
Our lower bounds show that strategic quantification in the context of hyperproperties results in a logic that is strictly harder (w.r.t.~model checking) than both a hyperlogic without strategic quantification (such as \HyperLTL{}) and a non-hyper logic with strategic quantification (such as \ATLS{}).
We establish the following bounds:

\begin{restatable}{thm}{lb}\label{theo:lb}
	Model checking of a linear \HyperATLS{} ~formula with $k$ complex and $l$ simple quantifiers is $(2k + l - 1)$-\EXPSPACE-hard in the size of the specification.
\end{restatable}

\begin{restatable}{thm}{lbb}\label{theo:lb2}
	Model checking of a linear \HyperATLS{} ~formula with $k$ complex and $l$ simple quantifiers is $(2k + l - 3)$-\EXPSPACE-hard in the size of the system.
\end{restatable}

\begin{figure}
	
	\begin{subfigure}{\linewidth}
		\centering
		\begin{tikzpicture}
			\node[] at (-1.25,1.1) () {$k = 0$}; 
			\node[] at (2,1.1) () {$k = 1$}; 
			\node[] at (6.5,1.1) () {$k \geq 2$};

			\draw[black, very thick, rounded corners=2pt,pattern=north east lines,pattern color=pastelblue!40] (1, 0.25) -- (8.75, 0.25) -- (8.75, -3.75) -- (-2.25, -3.75) -- (-2.25, -1.75) -- (-0.25, -1.75) -- (-0.25, 0.25) -- (1, 0.25);

			\draw [black, very thick, rounded corners=2pt,fill=white] (0, 0) rectangle (4,-1.5);
			\node[anchor=center] () at (2, -0.75) {\small $2$-\EXPTIME{}-complete};
			
			\node[anchor=center,align=left] () at (6.5, -0.75) {\small in $(2k+l)$-\EXPTIME{}\\ \small$(2k + l-1)$-\EXPSPACE{}-hard};
			
			\draw [black, very thick, rounded corners=2pt, fill =white] (-2, -2) rectangle (8.5,-3.5);
			\node[anchor=center,align=left] () at (3.25, -2.75) {\small $(2k+l-1)$-\EXPSPACE{}-complete};

			\draw[black!50, very thick] (-2.75, -3.75) -- (-2.75, 0.75) -- (8.75, 0.75);

			\draw [black!50, very thick] (-0.25, 0.9) -- (-0.25,0.6);
			\draw [black!50, very thick] (4.25, 0.9) -- (4.25,0.6);
			
			\draw [black!50, very thick] (-2.9, -1.75) -- (-2.6,-1.75);
			
			\node[rotate=90] at (-3.1,-0.75) () {$l = 0$}; 
			
			\node[rotate=90] at (-3.1,-2.75) () {$l \geq 1$}; 
		\end{tikzpicture} 
		
		\subcaption{Complexity in the size of the specification. $(2k+l)$-\EXPTIME{} containment and $(2k + l-1)$-\EXPSPACE{}-hardness holds for the entire fragment (the blue, striped area).}\label{fig:boundsSpec}
	\end{subfigure}\\[0.5cm]
	\begin{subfigure}{\linewidth}
		\centering
		\begin{tikzpicture}
			
			\draw[black, very thick, rounded corners=2pt,pattern=north east lines,pattern color=pastelblue!40] (9, 0.25) -- (13.25, 0.25) -- (13.25, -4.25) -- (4.75, -4.25) -- (4.75, 0.25) -- (9, 0.25);
			
			\draw[rounded corners=2pt,fill=white] (6,-2.75) -- (13,-2.75) -- (13,-4) -- (5,-4) -- (5,-2.75) -- (6,-2.75);
			
			\draw[black, very thick, rounded corners=2pt,pattern=dots,pattern color=chestnut!60]  (6,-2.75) -- (13,-2.75) -- (13,-4) -- (5,-4) -- (5,-2.75) -- (6,-2.75);
			
			\draw [black, very thick, rounded corners=2pt,fill=white] (5, 0) rectangle (8.5,-1.5);
			\node[anchor=center] () at (6.75, -0.75) {\small \PTIME{}-complete};
			
			\node[anchor=center,align=left] () at (11, -0.75) {\small in $(2k+l-2)$-\EXPTIME{}\\ \small$(2k+l-3)$-\EXPSPACE{}-hard};
			
			\draw [black, very thick, rounded corners=2pt,fill=white] (0, -2) rectangle (4.5,-4);
			\node[anchor=center] () at (2.25, -3) {\small $(l-2)$-\EXPSPACE{}-complete};

			\node[anchor=center,align=left] () at (9, -3.375) {\footnotesize (Outermost quantifier is simple) \\ \small $(2k+l-3)$-\EXPSPACE{}-complete};
			
			\draw[black!50, very thick] (-0.75, -4.25) -- (-0.75, 0.75) -- (13.25, 0.75);

			\draw [black!50, very thick] (4.75, 0.9) -- (4.75,0.6);
			\draw [black!50, very thick] (8.75, 0.9) -- (8.75,0.6);
			
			\draw [black!50, very thick] (-0.9, -1.75) -- (-0.6,-1.75);
			
			\node[] at (2.25,1.1) () {$k = 0$}; 
			\node[] at (6.75,1.1) () {$k = 1$}; 
			\node[] at (11,1.1) () {$k \geq 2$}; 
			
			\node[rotate=90] at (-1.1,-0.75) () {$l = 0$}; 
			
			\node[rotate=90] at (-1.1,-3) () {$l \geq 1$}; 
		\end{tikzpicture} 
		
		\subcaption{Complexity in the size of the system.    $(2k+l-2)$-\EXPTIME{} containment and $(2k + l-3)$-\EXPSPACE{}-hardness holds for the blue, striped area. $(2k+l-3)$-\EXPSPACE{}-completeness only holds for formulas where the outermost quantifier is simple (the red, dotted area).	}\label{fig:boundsSys}
	\end{subfigure}

	\caption{Upper and lower bounds on the complexity of model checking linear \HyperATLS{} formulas with $k$ complex and $l$ simple quantifiers.
		The bounds are given in the size of the specification (\Cref{fig:boundsSpec}) and size of the system (\Cref{fig:boundsSys}).
	}
\end{figure}


\paragraph{Complexity in the Specification Size}

If we consider the \HyperATLS{} model checking complexity in terms of the specification size, \refProp{upperBound} and \refTheo{lb}  span the landscape depicted in \Cref{fig:boundsSpec}.
For all prefix structures, we get an upper bound of $(2k+l)$-\EXPTIME{} and a lower bound of $(2k+l-1)$-\EXPSPACE{}-hardness.
In two cases, we can improve the upper or lower bounds further and get tight results:
If $k = 1$ and $l =0$, we get a better $2$-\EXPTIME{} lower bound from the \ATLS{} model checking \cite{AlurHK02} and thus $2$-\EXPTIME{}-completeness.
In case $l \geq 1$, we get a matching $(2k+l-1)$-\EXPSPACE{} upper bound and thus $(2k+l-1)$-\EXPSPACE{}-completeness (note that this subsumes the already known \HyperLTL{} bounds in case $k = 0$ \cite{Rabe16}).

\paragraph{Complexity in the System Size}

If we consider the complexity in the size of the system, \refProp{upperBoundModel} and \refTheo{lb2} span the landscape depicted in \Cref{fig:boundsSys}.
In case $k = 0$ (where the formula is a \HyperLTL{} formula), we get $(l-2)$-\EXPSPACE{}-completeness \cite{Rabe16}.
In the case where $k \geq 1$ (\ie, the formula is a ``proper'' \HyperATLS{} formula), we get a upper bound of $(2k+l-2)$-\EXPTIME{} and a lower bound of $(2k+l-3)$-\EXPSPACE{}-hardness.
In two cases we can tighten the results:
In case $k = 1$ and $l=0$, we get a better lower bound from the \ATLS{} model checking and thus \PTIME{}-completeness \cite{AlurHK02}.
In the cases where $l \geq 1$ and the outermost quantifier is simple, we get an improved upper bound resulting in $(2k+l-3)$-\EXPSPACE{}-completeness.

\subsection{Proof Preliminaries}

The remainder of this section is devoted to a proof of \refTheo{lb} and \refTheo{lb2}.
Readers less interested in a formal proof can skip to \Cref{sec:proto}.

Our proof encodes the acceptance of space-bounded Turing machines (TM).
It builds on ideas used for the \HyperLTL{} lower bounds shown by Rabe \cite{Rabe16} (adopting earlier ideas from Stockmeyer \cite{stockmeyer1974complexity}) but uses a novel construction to achieve a doubly exponential increase using complex quantification.
The main idea of our construction is to design a \HyperATLS{} formula that requires a player to output a yardstick, which is a formula that specifies a fixed distance between two points along a path.
We can then encode the acceptance of a TM by using the yardstick to compare consecutive configurations of the TM.
We recommend having a look at the \HyperLTL{} lower bound shown by Rabe \cite[\S 5.6]{Rabe16}.

\paragraph{Precise Tower Length}

We use a slightly larger tower of exponents.
For $k, n \in \mathbb{N}$ define $\mathcal{C}(k, n)  \in \nat$ as follows:
\begin{align*}
	\mathcal{C}(0, n) &:= n\\
	\mathcal{C}(k+1, n) &:= 2^{2^{\mathcal{C}(k, n) }} \cdot 2^{\mathcal{C}(k, n)} \cdot \mathcal{C}(k, n)
\end{align*}
It is easy to see that for every $k \geq 1$, we have $\mathcal{C}(k, n)  \geq \mathit{Tower}_2(2k, n)$ for every $n$.
We design a formula with $k$ complex and $0$ simple quantifier that specifies a yardstick of length $\mathcal{C}(k, n)$.\footnote{In our proof, we encode the acceptance of $\mathit{Tower}_2(\cdot, n)$-space bounded TMs, \ie, we fix the base to $2$.
	Our construction easily extends to an arbitrary (but fixed) base $c$ by using $\lceil \log_2 c \rceil$ propositions for our counter construction.
	We stick with $c=2$ to keep the notation simple.
}
This will later allow us to encode the acceptance of a $\mathcal{C}(k, n)$-space bounded TM. 
In our construction, the size of the game structure is constant, and the size of the formula depends on $n$.

\subsection{Direct Counter Verification}\label{sec:counterConst1}

We first consider the case where $k = 1$ and construct a formula that ensures a yardstick of length $\mathcal{C}(1, n)$.
The idea is to describe a counter with $2^n$ many bits that is incremented in each step and resets to $0$ once the maximal value is reached.
Consequently, there are $2^{2^n}$ counter configurations between two resets of the counter.

\paragraph{Structure of a Counter}

\begin{figure}
	\begin{center}
		\scalebox{0.84}{
			\begin{tikzpicture}
				\node[circle, draw,inner sep=4pt,label=below:{\scriptsize$\{\bullet,\triangle,\blacktriangle\}$}] at (0,0) (n11) {};
				\node[circle, draw,inner sep=4pt] at (0.6,0) (n12) {};
				
				\node[] at (1.4,0) (n13) {\small$\dots$};
				
				\node[circle, draw,inner sep=4pt] at (2.2,0) (n1n) {};
				
				\draw[thick,-] (n11) -- (n12);
				\draw[thick,-] (n12) -- (n13);
				\draw[thick,-] (n13) -- (n1n);

				\node[circle, draw,inner sep=4pt,label=below:{\scriptsize$\{\bullet\}$}] at (2.8,0) (n21) {};
				\node[circle, draw,inner sep=4pt] at (3.4,0) (n22) {};
				
				\node[] at (4.2,0) (n23) {\small$\dots$};
				
				\node[circle, draw,inner sep=4pt] at (5,0) (n2n) {};
				
				\draw[thick,-] (n21) -- (n22);
				\draw[thick,-] (n22) -- (n23);
				\draw[thick,-] (n23) -- (n2n);

				\draw[thick, -] (n1n) -- (n21);
				
				\draw[thick,-] (n2n) -- (5.4, 0);

				\node[] at (6,0) () {\small$\ldots$};

				\node[circle, draw,inner sep=4pt,label=below:{\scriptsize$\{\bullet\}$}] at (7,0) (n31) {};
				\node[circle, draw,inner sep=4pt] at (7.6,0) (n32) {};
				
				\node[] at (8.4,0) (n33) {\small$\dots$};
				
				\node[circle, draw,inner sep=4pt] at (9.2,0) (n3n) {};
				
				\draw[thick,-] (6.6,0) -- (n31);
				
				\draw[thick,-] (n31) -- (n32);
				\draw[thick,-] (n32) -- (n33);
				\draw[thick,-] (n33) -- (n3n);

				\node[circle, draw,inner sep=4pt,label=below:{\scriptsize$\{\bullet,\triangle\}$}] at (9.8,0) (n41) {};
				\node[circle, draw,inner sep=4pt] at (10.4,0) (n42) {};
				
				\node[] at (11.2,0) (n43) {\small$\dots$};
				
				\node[circle, draw,inner sep=4pt] at (12,0) (n4n) {};
				
				\draw[thick,-] (n3n) -- (n41);
				
				\draw[thick,-] (n41) -- (n42);
				\draw[thick,-] (n42) -- (n43);
				\draw[thick,-] (n43) -- (n4n);
				
				\draw[thick,-] (n4n) -- (12.4,0);
				
				\node[] at (13,0) () {\small$\ldots$};
				
				\node[circle, draw,inner sep=4pt,label=below:{\scriptsize$\{\bullet,\triangle,\blacktriangle\}$}] at (14,0) (n51) {};
				\node[circle, draw,inner sep=4pt] at (14.6,0) (n52) {};
				
				\node[] at (15.4,0) (n53) {\small$\dots$};

				\node[circle, draw,inner sep=4pt] at (16.2,0) (n5n) {};
				\draw[thick,-] (13.6,0) -- (n51);
				
				\draw[thick,-] (n51) -- (n52);
				\draw[thick,-] (n52) -- (n53);
				\draw[thick,-] (n53) -- (n5n);

				\draw[draw=black, thick] (-0.2, 0.3) -- (-0.2, 0.5);
				\draw[draw=black,thick] (-0.2, 0.4) --node[above] {$n$} (2.4, 0.4);
				\draw[draw=black,thick] (2.4, 0.5) -- (2.4, 0.3);

				\draw[draw=black, thick] (2.6, 0.3) -- (2.6, 0.5);
				\draw[draw=black,thick] (2.6, 0.4) --node[above] {$n$} (5.2, 0.4);
				\draw[draw=black,thick] (5.2, 0.5) -- (5.2, 0.3);

				\draw[draw=black, thick] (-0.2, 0.9) -- (-0.2, 1.1);
				\draw[draw=black,thick] (-0.2, 1) --node[above] {$n \cdot 2^n$} (9.4, 1);
				\draw[draw=black,thick] (9.4, 0.9) -- (9.4, 1.1);

				\draw[draw=black, thick] (-0.2, 1.6) -- (-0.2, 1.8);
				\draw[draw=black, thick] (-0.2, 1.7) --node[above] {$n \cdot 2^n \cdot 2^{2^n}$} (13.6, 1.7);
				\draw[draw=black,thick] (13.6, 1.6) -- (13.6, 1.8);

				\node[draw=chestnut, thick] at (0,-1) (n11) {};
				\node[draw, thick, chestnut] at (0.6,-1) (n11) {};
				\node[draw, thick, chestnut] at (2.2,-1) (n11) {};
				
				\node[draw, thick, chestnut] at (2.8,-1) (n11) {};
				\node[draw, thick, chestnut] at (3.4,-1) (n11) {};
				\node[draw, thick, chestnut] at (5,-1) (n11) {};

				\node[draw, thick, chestnut] at (7,-1) (n11) {};
				\node[draw, thick, chestnut] at (7.6,-1) (n11) {};
				\node[draw, thick, chestnut] at (9.2,-1) (n11) {};

				\node[draw, thick, chestnut] at (9.8,-1) (n11) {};
				\node[draw, thick, chestnut] at (10.4,-1) (n11) {};
				\node[draw, thick, chestnut] at (12,-1) (n11) {};
				
				\node[draw, thick, chestnut] at (14,-1) (n11) {};
				\node[draw, thick, chestnut] at (14.6,-1) (n11) {};
				\node[draw, thick, chestnut] at (16.2,-1) (n11) {};

				\node[] at (-0.7,-1) (n11) {\color{chestnut}$\mathbf{a}$};
				
				\node[] at (-0.7,-1.5) (n11) {\color{pastelblue}$\mathbf{b}$};
				
				\node[draw=pastelblue, thick] at (0,-1.5) (n11) {};
				
				\node[draw=pastelblue, thick]at (2.8,-1.5) (n11) {};
				
				\node[draw=pastelblue, thick]at (7,-1.5) (n11) {};

				\node[draw=pastelblue, thick] at (9.8,-1.5) (n11) {};
				\node[draw=pastelblue, thick] at (14,-1.5) (n11) {};
			\end{tikzpicture}
		}
	\end{center}
	
	\caption{The basic structure of a correct counter of length $\mathcal{C}(1, n)$. Proposition $\bullet$ holds every $n$ steps and separates two $a$-configurations. Proposition $\triangle$ holds every $n \cdot 2^n$ step and separates two $b$-configurations. $\blacktriangle$ holds every $n \cdot 2^n \cdot 2^{2^n}$ steps.
		Propositions $a$ and $b$ give the bits of $a$-counter and $b$-counter, respectively.
		The $a$ proposition is relevant for the $a$-counter at all positions, marked by red boxes.
		The relevant positions for the $b$-counter are only those where $\bullet$ holds, marked by blue boxes. 
		Consequently, there are $2^n$ relevant positions for the $b$-counter between any two occurrences of $\triangle$.}\label{fig:countStruct}
\end{figure}

To ensure the correctness of the counter with $2^n$ bits, we use a second counter with $n$ bits. 
The bits of the larger counter (with $2^n$ bits) are given via atomic proposition $b$ (called the $b$-counter), and the bits of the smaller counter (with $n$ bits) by proposition $a$ (called the $a$-counter).
Together, the counter uses atomic propositions $\{\bullet, \triangle,\blacktriangle, a, b\}$, where $\bullet,\triangle,\blacktriangle$ are separating constructs and $a, b$ are propositions that give the value of the counter-bits (if, \eg, proposition $a$ is set, we interpret this as a $1$-bit and otherwise as a $0$-bit of the $a$-counter).
A \emph{correct} counter has the form depicted in \Cref{fig:countStruct}.
Proposition $\bullet$ occurs every $n$ steps and separates two configurations of the $a$-counter.
The $a$-counter should continuously count from $0$ to $2^n-1$ (in binary) and then restart at $0$ (we use a big-endian encoding where the least significant bit is the last position of a count). 
The $\triangle$ occurs every $n \cdot 2^n$ steps and marks the position where the $a$-counter resets to $0$, \ie, $\triangle$ holds whenever the following $a$-configuration is $0$.
Proposition $b$ is used for the second counter with $2^n$ many bits.
Each $b$-configuration is separated by $\triangle$ (\ie, one $b$-configuration has exactly the length in which the $a$-counter counts from $0$ to $2^n-1$).
The relevant bits for the $b$-counter are only those positions where $\bullet$ holds, so each $b$-configuration has $2^n$ relevant bits (marked as blue boxes in \Cref{fig:countStruct}).
The $b$-counter should count from $0$ to $2^{2^n}-1$ (in binary) and then restart at $0$. 
We mark the reset of the $b$-counter by $\blacktriangle$, \ie, $\blacktriangle$ holds whenever the following $b$-configuration is $0$.
Consequently, $\blacktriangle$ holds every $\mathcal{C}(1, n) = n \cdot 2^n \cdot 2^{2^n}$ steps.
An explicit prefix of a correct counter in the case where $n = 2$ is depicted in \Cref{fig:exampleCounter}.

\begin{figure}
	\small
	\centering
	\def\arraystretch{1.5}
	\begin{tabular}{c|ccccccccccccccccccc}
		$\bullet$ & $\blacksquare$ & $\square$ & $\blacksquare$ & $\square$ & $\blacksquare$ & $\square$ & $\blacksquare$ & $\square$  & $\blacksquare$ & $\square$ & $\blacksquare$ & $\square$ & $\blacksquare$ & $\square$ & $\blacksquare$ & $\square$ & $\blacksquare$ & $\square$ & $\blacksquare$  \\
		$\triangle$ & $\blacksquare$ & $\square$ & $\square$ & $\square$ & $\square$ & $\square$ & $\square$ & $\square$ & $\blacksquare$ & $\square$ & $\square$ & $\square$ & $\square$ & $\square$ & $\square$ & $\square$ & $\blacksquare$ & $\square$ & $\square$\\
		$\blacktriangle$ & $\blacksquare$ & $\square$ & $\square$ & $\square$ & $\square$ & $\square$ & $\square$ & $\square$ & $\square$ & $\square$ & $\square$ & $\square$ & $\square$ & $\square$ & $\square$ & $\square$ & $\square$ & $\square$ & $\square$ \\
		$a$ & $\square$ & $\square$ & $\square$  & $\blacksquare$ & $\blacksquare$ & $\square$ & $\blacksquare$ & $\blacksquare$ &$\square$ & $\square$ & $\square$  & $\blacksquare$ & $\blacksquare$ & $\square$ & $\blacksquare$ & $\blacksquare$ & $\square$ & $\square$ & $\square$\\
		$b$ & $\square$ &  &  $\square$ &  & $\square$ &  &  $\square$ &  &  $\square$ &  &  $\square$ &  &  $\square$ &  &  $\blacksquare$ &  &  $\square$ &  & $\square$\\
		\hline
		$c(a)$& \multicolumn{2}{c|}{$0$} &\multicolumn{2}{c|}{$1$}   &  \multicolumn{2}{c|}{$2$} & \multicolumn{2}{c|}{$3$} & \multicolumn{2}{c|}{$0$} & \multicolumn{2}{c|}{$1$}&\multicolumn{2}{c|}{$2$} &\multicolumn{2}{c|}{$3$} & \multicolumn{2}{c|}{$0$} & \\
		\cline{2-20}
		$c(b)$&  \multicolumn{8}{c|}{$0$} & \multicolumn{8}{c|}{$1$} && &\\
		\cline{2-20}
	\end{tabular}
	
	\caption{A prefix of a correct counter in the case where $n = 2$. The trace is read from left to right, where each column gives the evaluation of each atomic proposition at that step. $\blacksquare$ means that the proposition holds, and $\square$ that it does not hold.
		The last two rows give the value of each $a$-configuration and $b$-configuration, respectively. Each $a$-configuration has length $2$ ($= n$), and each $b$-configuration has length $8$ ($= n \cdot 2^n$) with $4$ ($ = 2^n$) relevant bits. The positions where proposition $b$ is irreverent for the counter are left blank.  } \label{fig:exampleCounter}
\end{figure}

\paragraph{Counter as a Game}

We interpret the construction of the counter as a game between a verifier (\veri) and a refuter (\refu).
The verifier selects, in each step, the evaluation of the propositions in $\{\bullet, \triangle,\blacktriangle, a, b\}$ that form the counter.
Meanwhile, the refuter can challenge the correctness of the counter (we make this precise below).
We design a specification and game structure such that the \emph{only} winning strategy for \veri{} is to produce a correct counter.
Consequently, on any play compatible with any winning strategy for $\veri$, $\blacktriangle$ holds every $\mathcal{C}(1, n)$ steps.

\paragraph{Game Structure}

\begin{figure}
	\centering
	\begin{tikzpicture}
		\node[draw, thick, circle, label=west:{\footnotesize$s_\mathit{init}$}] at (-1,0) (ns) {};
		
		\node[draw, thick, circle, label=north:{\footnotesize$s_1$}] at (3,-0.75) (nostart) {};
		
		\node[draw,thick, circle, label=south:{\footnotesize$s_2$}] at (0,0.75) (startloop) {};
		
		\node[draw,thick, circle, label=south:{\footnotesize$s_3$}, label=north:{\tiny$\{\mathit{errorStart}\}$}] at (1.5,0.75) (start) {};
		
		\node[draw,thick, circle, label=south:{\footnotesize$s_4$}] at (3,0.75) (startloopi) {};
		
		\node[draw,thick, circle, label=south:{\footnotesize$s_{\{\bullet, \triangle, \blacktriangle\}}$}, label=north:{\tiny$\{\bullet, \triangle, \blacktriangle\}$}] at (4.7,0) (enum) {};
		
		\node[] at (5.8,-0.5) (target1) {$\cdots$};
		\node[] at (5.8,0) (target2) {$\cdots$};
		\node[] at (5.8,0.5) (target3) {$\cdots$};
		
		\draw[->,thick] (enum) -- (target1);
		\draw[->,thick] (enum) -- (target2);
		\draw[->,thick] (enum) -- (target3);
		
		\draw[->,thick] (ns)+(-0.5, -0.5) -- (ns);
		\draw[->,thick] (ns) |- (startloop);
		\draw[->,thick] (ns) |- (nostart);
		
		\draw[->,thick] (start) -- (startloopi);
		
		\path (startloopi) edge [loop above, thick] node {} (startloopi);
		
		\path (startloop) edge [loop above,thick] node {} (startloop);
		\path (nostart) edge [loop below,thick] node {} (nostart);
		
		\draw[dashed, black!40, thick,rounded corners=2pt] (4, 0) -- (4, -1) -- (7, -1) -- (7, 1) -- (4, 1) -- (4, 0);
		
		\draw[->,thick] (startloop) -- (start);
		\draw[->,thick] (nostart) -- (enum.west);
		\draw[->,thick] (startloopi) -- (enum.west);
	\end{tikzpicture}
	
	\caption{The game structure $\calG_\mathit{counter}$ that is used to produce the counter. The part surrounded by the dashed box, generates the actual counter and includes all states of the form $s_O$ for some $O \subseteq \{\bullet, \triangle,\blacktriangle, a, b, \mathit{error}\}$. The remaining states are used to determine the time point at which the counter should start. The verifier decides the successor whenever in $s_4$, and the refuter decides the successor for states $s_\mathit{init}, s_1, s_2$ (state $s_3$ has a unique successor). We omit the label in case it is empty.}\label{fig:gameStructure}
\end{figure}

In our game structure, $\veri$  sets the values of the propositions in\linebreak$\{\bullet, \triangle,\blacktriangle, a, b\}$ and $\refu$ can challenge the correctness by setting proposition $\mathit{error}$ and $\mathit{errorStart}$.
We allow \refu{} to postpone the start of the counter.
Consider the following game structure over atomic propositions $\atomic := \{\bullet, \triangle,\blacktriangle, a, b, \mathit{error}, \mathit{errorStart}\}$:

\begin{defi}
	Define the CGS $\calG_\mathit{counter} := (S, s_\mathit{init}, \{\veri, \refu\}, \moves, \delta, L)$ where 
	\begin{align*}
		S := \{s_O \mid O \subseteq \{\bullet, \triangle,\blacktriangle, a, b, \mathit{error}\} \} \cup \{s_\mathit{init}, s_1, s_2, s_3, s_4\}
	\end{align*}
	the moves are given by $\moves := 2^{\{\bullet, \triangle,\blacktriangle, a, b, \mathit{error}\}}  \times \mathbb{B}$.
	Transition function $\delta : S \times (\{\veri, \refu\} \to \moves) \to S$ is defined by 
	\begin{align*}
		\delta(s_{O}, [\veri \mapsto (x, \_), \refu \mapsto (y, \_)]) &:= s_{(x \cap \{\bullet, \triangle,\blacktriangle, a, b\}) \cup (y \cap \{\mathit{error}\})} \\[0.2cm]
		\delta(s_{\mathit{init}}, [\veri \mapsto \_, \refu \mapsto (\_, b)]) &:= \begin{cases}
			\begin{aligned}
				&s_1 \quad \quad\quad\;\;&&\text{if } b = \top\\
				&s_2 \quad &&\text{if } b = \bot
			\end{aligned}
		\end{cases}\\[0.2cm]
		\delta(s_1, [\veri \mapsto \_, \refu \mapsto (\_, b)]) &:= \begin{cases}
			\begin{aligned}
				&s_1 \quad &&\text{if } b = \top\\
				&s_{\{\bullet, \triangle, \blacktriangle\}} \quad &&\text{if } b = \bot
			\end{aligned}
		\end{cases}\\[0.2cm]
		\delta(s_2, [\veri \mapsto \_, \refu \mapsto (\_, b)]) &:= \begin{cases}
			\begin{aligned}
				&s_2 \quad \quad\quad\;\; &&\text{if } b = \top\\
				&s_3 \quad &&\text{if } b = \bot
			\end{aligned}
		\end{cases}\\[0.2cm]
		\delta(s_3, \_) &:= s_4 \\
		\delta(s_4, [\veri \mapsto (\_, b), \refu \mapsto \_]) &:= \begin{cases}
			\begin{aligned}
				&s_4 \quad &&\text{if } b = \top\\
				&s_{\{\bullet, \triangle, \blacktriangle\}} \quad &&\text{if } b = \bot
			\end{aligned}
		\end{cases}
	\end{align*}
	where $\_$ denotes an arbitrary value. The labeling $L : S \to 2^{\{\bullet, \triangle,\blacktriangle, a, b, \mathit{error}\}}$ is defined by
	\begin{align*}
		L(s_O) &= O\\
		L(s_\mathit{init}) = L(s_1) = L(s_2) = L(s_4) &= \emptyset\\
		L(s_3) &= \{\mathit{errorStart}\} \tag*{\demo}
	\end{align*}
\end{defi}

\noindent
Once a state $s_O$ for some $O \subseteq \{\bullet, \triangle,\blacktriangle, a, b, \mathit{error}\}$ is reached, the verifier can determine which of the propositions in $\{\bullet, \triangle,\blacktriangle, a, b\}$ hold at the next step and refuter decides if $\mathit{error}$ holds (the first case in the definition of $\delta$).
This part of the state space is responsible for generating the actual counter.
The remaining states ($s_\mathit{init}, s_1, s_2, s_3$, and $s_4$) are used to determine \emph{when} the counter should start.
The structure is sketched in \Cref{fig:gameStructure}. 
States $s_\mathit{init}, s_1$, and $s_2$ are controlled by \refu{}, \ie, the move selected by \refu{} determines the successor state, whereas state $s_4$ is controlled by \veri{} (see the definition $\delta$).
If \refu{} moves to $s_1$, the start of the actual counter can be delayed by looping in $s_1$. 
If \refu{} moves to $s_2$, the $\mathit{errorStart}$ proposition occurs at some point, and afterward, $\veri$ can decide when the counter starts by looping in $s_4$ (this will be of importance to verify the construction in the case $k > 1$).

\paragraph{Specification}

We enforce that $\veri$ produces a correct counter once state $s_{\{\bullet, \triangle, \blacktriangle\}}$ is reached for the first time. 
Consider the \HyperATLS{} specification $\mathit{correct}_1$ defined as follows:

\begin{formulaBox}{$\mathit{correct}_1$}
	\small
	\begin{align}
		(\neg \blacktriangle_{\pi_1} )\ltlW \bigg( &\blacktriangle_{\pi_1}\land \nonumber\\
		& \ltlg \Big(\blacktriangle_{\pi_1} \to \triangle_{\pi_1}) \land (\triangle_{\pi_1} \to \bullet_{\pi_1}) \land (\bullet_{\pi_1} \leftrightarrow \ltlnext^n \bullet_{\pi_1}  )\Big) \, \land \label{eq:base1}\\
		&\ltlg\Big(\triangle_{\pi_1} \leftrightarrow  \big(\neg a_{\pi_1} \land  \ltlN ((\neg a_{\pi_1}) \ltlU  \bullet_{\pi_1} )\big) \Big) \, \land \label{eq:base2}\\
		&\ltlg\Big(\blacktriangle_{\pi_1} \leftrightarrow \big(\neg b_{\pi_1} \land \ltlN\big((\bullet_{\pi_1} \rightarrow \neg b_{\pi_1}) \ltlU  \triangle_{\pi_1}\big)\big) \Big) \, \land \label{eq:base3}\\
		&\ltlg\Big( (a_{\pi_1} \leftrightarrow \ltlnext^n a_{\pi_1}) \leftrightarrow \ltlN  \mathit{before}(\neg a_{\pi_1}, \bullet_{\pi_1})  \Big) \, \land \label{eq:base4}\\
		&\Big(\big(\mathit{exactlyOnce}(\mathit{error}_{\pi_1}) \land \ltlg(\mathit{error}_{\pi_1} \to \bullet_{\pi_1}) \big)\rightarrow \mathit{falseAlarm}_1\Big) \bigg) \label{eq:base5}
	\end{align}
\end{formulaBox}

\noindent
The initial weak until accounts for the offset before the counter is started, \ie, the counter should be correct once $\blacktriangle_{\pi_1}$ holds for the first time (which is the case when state $s_{\{\bullet, \triangle, \blacktriangle\}}$ is reached for the first time).
We discuss the conjuncts \ref{eq:base1} - \ref{eq:base5}  in detail:

\begin{enumerate}[leftmargin=1cm]
	\item[(\ref{eq:base1})] Formula \ref{eq:base1} states that the basic behavior of the separating constructs is correct: The implications between $\blacktriangle, \triangle$ and $\bullet$ hold, and $\bullet$ occurs every $n$ steps.
	\item[(\ref{eq:base2})] Formula \ref{eq:base2}, ensures that the $\triangle$ proposition is set correctly. $\triangle$ should hold exactly if the next $a$-configuration is $0$, \ie, $a$ does not hold until the next $a$-configuration begins (which is marked by $\bullet$).
	
	\item[(\ref{eq:base3})] Similar to Formula \ref{eq:base2}, Formula \ref{eq:base3} ensures that $\blacktriangle$ is set exactly if the next $b$-configuration is $0$, so $b$ does not hold until the next $b$-configuration begins (marked by $\triangle$).
	Recall that the bits of the $b$-counter are only those where $\bullet$ holds. 
	
	\item [(\ref{eq:base4})]  Formula \ref{eq:base4} ensures the correctness of the $a$-counter.
	Here 
	\begin{align*}
		\mathit{before}(\varphi, \psi) := (\neg \psi) \ltlU (\varphi \land \neg \psi)
	\end{align*}
	expresses that $\varphi$ should hold at some time strictly before $\psi$ holds for the first time.
	To encode that the $a$-counter is incremented (or reset), we use the following fact.
	Assume we are given two $m$-bit counters $\alpha = \alpha_0, \ldots, \alpha_m$ and $\beta = \beta_0, \ldots, \beta_m$ (with big-endian encoding, \ie, $\alpha_m, \beta_m$ are the least significant bits).
	Let $c(\alpha), c(\beta) \in \{0, \ldots, 2^m-1\}$ give the value of the counters.
	Then
	\begin{align}\label{eq:counter}
		\Big[c(\beta) = c(\alpha) + 1 \, \mod 2^m\Big] \quad \text{iff} \quad \Big[\forall i\ldot (\alpha_i = \beta_i) \leftrightarrow (\exists j > i. \alpha_j = 0)\Big].
	\end{align}
	Formula \ref{eq:base4} thus encodes that the value of the $a$-counter is incremented by using $\ltlN^n$ to compare the same position across two consecutive $a$-configurations.
	
	\item [(\ref{eq:base5})] Lastly, Formula \ref{eq:base5} ensures the correctness of the $b$-counter. 
	Here 
	\begin{align*}
		\mathit{exactlyOnce}(\psi) := (\neg \psi) \ltlU (\psi \land \ltlnext \ltlg(\neg \psi))
	\end{align*}
	expresses that $\psi$ holds exactly once.
	Different from Formula \ref{eq:base3}, we cannot encode the correctness directly via $\ltlN$s as the relevant positions of the $b$-counter are exponentially many steps apart.
	Instead, we use \refu{}s ability to raise the $\mathit{error}$ proposition. 
	To challenge the correctness, \refu{} should set $\mathit{error}$ at the \emph{first} bit in the $b$-counter that is incorrect.
	After the challenge occurred, we check if \refu{} spotted a genuine mistake in the counter (so $\mathit{falseAlarm}_1$ should state that the challenge was a false alarm; the counter at the position pointed to by $\mathit{error}$ was correct).
	To check the challenge we need to compare the position marked with $\mathit{error}$ with the same position  in the \emph{previous} $b$-configuration (which is the position $n \cdot 2^n$ steps earlier).
	The crux is that once the $\mathit{error}$ proposition is set, we can identify this position by using the $a$-counter (which we can assume to be correct).
	Consider the formula $\mathit{prevPos}$ defined as follows:
	
	\begin{formulaBox}{$\mathit{prevPos}$}
		\small
		\begin{align}
			&\bullet_{\pi_1} \land \label{eq:prev1} \\
			&\bigg((\neg \triangle_{\pi_1}) \ltlU \Big(\triangle_{\pi_1} \land \ltlnext\big( (\neg \triangle_{\pi_1}) \ltlU \mathit{error}_{\pi_1} \big)\Big)\bigg) \; \land \label{eq:prev2} \\ 
			&\bigg(\bigwedge_{i = 0}^{n-1} \Big(\big(\ltlnext^i a_{\pi_1}\big) \leftrightarrow \ltlg \big(\mathit{error}_{\pi_1} \rightarrow \ltlnext^i a_{\pi_1}\big)\Big)\bigg) \label{eq:prev3}
		\end{align}
	\end{formulaBox}
	
	\noindent
	Formula $\mathit{prevPos}$ holds exactly once (assuming that the $a$-counter is correct and $\mathit{error}_{\pi_1}$ occurs exactly once), and identifies the same position in the $b$-configuration that precedes that in which $\mathit{error}_{\pi_1}$ holds. 
	There are three conditions that identify this (unique) position:
	\begin{enumerate}[leftmargin=1cm]
		\item[(\ref{eq:prev1})] The position aligns with $\bullet_{\pi_1}$ (\ie, a position where $b$ is relevant).
		\item[(\ref{eq:prev2})] The position lies within the $b$-configuration that preceeds the $b$-configuration in which $\mathit{error}_{\pi_1}$ holds, \ie, $\triangle_{\pi_1}$ holds exactly once before $\mathit{error}_{\pi_1}$ holds.
		\item[(\ref{eq:prev3})] The position corresponds to the same bit of the $b$-counter. Each bit of the $b$-counter is uniquely characterized by the $a$-configuration that starts at that bit.
		To check that we are at the same position, the $a$-configuration should thus be the same as the $a$-configuration at the position pointed to by $\mathit{error}$.
	\end{enumerate}
	
	\noindent
	Using $\mathit{prevPos}$, we can then state that $\mathit{error}$ does not mark a genuine error in the $b$-counter. Here we again make use of \Cref{eq:counter}.
	We define formula $\mathit{falseAlarm}_1$ (which is used in formula $\mathit{correct}_1$) as follows:
	
	\begin{formulaBox}{$\mathit{falseAlarm}_1$}
		\small
			\begin{align*}
			\begin{split}
				&\Big(\ltlg \big(\mathit{prevPos} \rightarrow b_{\pi_1}\big) \leftrightarrow  \ltlg \big(\mathit{error}_{\pi_{1}} \rightarrow b_{\pi_1}\big)\Big)\\ 
				&\leftrightarrow \Big(\ltlg \big(\mathit{prevPos} \rightarrow  \ltlN \mathit{before}(\bullet_{\pi_1} \land \neg b_{\pi_1}, \triangle_{\pi_1} )\big)\Big)
			\end{split}
		\end{align*}
	\end{formulaBox}
	
	\noindent
	The bit of the $b$-counter at the position where $\mathit{prevPos}$ and the bit where $\mathit{error}_{\pi_{i}}$ holds (note that both positions are unique) should agree iff in the $b$-configuration pointed to by $\mathit{prevPos}$ there exists a $0$-bit at a less significant position (\cf~\Cref{eq:counter}).
\end{enumerate}

\noindent
It is easy to see that $\llangle \veri \rrangle \pi_1\ldot \mathit{correct}_1$ holds on $\calG_\mathit{counter}$. 
The only winning strategy for $\veri$ is to output a correct counter (as soon as state $s_{\{\bullet, \triangle, \blacktriangle\}}$ is reached).
In particular, on any play produced by a winning strategy for $\veri$, $\blacktriangle_{\pi_1}$ occurs exactly every $\mathcal{C}(1, n)$ steps (once $s_{\{\bullet, \triangle, \blacktriangle\}}$ is reached for the first time).

\subsection{Counter Verification Using Smaller Counter}\label{sec:counterConst2}

To obtain a yardstick of length $\mathcal{C}(k, n)$ for $k > 1$, we use the same counter structure in \Cref{fig:countStruct} but set the length (number of bits) of the $a$-counter to be $\mathcal{C}(k-1, n)$.
To verify the correctness of the counter, we then use a smaller yardstick of length $\mathcal{C}(k-1, n)$.
The final formula has the form
\begin{align*}
	\textstyle\llangle \veri \rrangle \pi_k\ldot \ldots \llangle \veri \rrangle \pi_1\ldot \bigwedge_{i=1}^k \mathit{correct}_i.
\end{align*}
We assert that every winning strategy for $\veri$ constructing path $\pi_i$ encodes a counter of length $\mathcal{C}(i, n)$ and use the counter on $\pi_{i-1}$ (which we can inductively assume to be correct) to ensure its correctness.
The construction in \Cref{sec:counterConst1} gives the base case for $k = 1$.
To verify the correctness of the counter on $\pi_i$ (for $i > 1$), we make use of the two modes available to \refu{} in $\calG_\mathit{counter}$ (see \Cref{fig:gameStructure}). 
By moving to state $s_1$, \refu{} can start the counter at any time (we will use this to verify that the placement of $\bullet$ and the $a$-counter is correct).
By moving to $s_2$, \refu{} can set the $\mathit{errorStart}$ proposition at any time, after which \veri{} can postpone the start of the counter.
We use this to verify the correctness of the $b$-counter. 

For $i > 1$, define formula $\mathit{correct}_i$ as follows:

\begin{formulaBox}{$\mathit{correct}_i$}
	\small
	\begingroup
	\allowdisplaybreaks
	\begin{align}
		(\neg \blacktriangle_{\pi_i}) \ltlW \bigg(&\blacktriangle_{\pi_i} \land\nonumber \\
		&\ltlg \Big((\blacktriangle_{\pi_i} \rightarrow \triangle_{\pi_i}) \land (\triangle_{\pi_i} \rightarrow \bullet_{\pi_i}) \Big) \, \land \label{eq:ind1}\\
		&\ltlG\Big(\blacktriangle_{\pi_{i-1}} \land \bullet_{\pi_i} \rightarrow (\neg \bullet_{\pi_i} \land \neg \blacktriangle_{\pi_{i-1}}) \ltlU (\bullet_{\pi_i} \land \blacktriangle_{\pi_{i-1}}) ) \Big) \, \land \label{eq:ind11}\\
		&\ltlg\Big(\triangle_{\pi_i} \leftrightarrow \big( \neg a_{\pi_i} \land  \ltlN( (\neg a_{\pi_i}) \ltlU  \bullet_{\pi_i}) \big) \Big) \, \land \label{eq:ind2}\\
		&\ltlg\Big(\blacktriangle_{\pi_i} \leftrightarrow \big(\neg b_{\pi_i} \land \ltlN\big((\bullet_{\pi_i} \rightarrow \neg b_{\pi_i}) \ltlU  \triangle_{\pi_i}\big) \big)\Big) \, \land \label{eq:ind3}\\
		\begin{split}
			&\ltlg\Big( \big(a_{\pi_i} \land \blacktriangle_{\pi_{i-1}} \land \mathit{next}(\blacktriangle_{\pi_{i-1}}, a_{\pi_i})\big) \rightarrow  \ltlN \mathit{before}( \neg a_{\pi_i}, \bullet_{\pi_i})  \Big) \, \land \\
			&\ltlg\Big( \big(\neg a_{\pi_i} \land \blacktriangle_{\pi_{i-1}} \, \land \mathit{next}(\blacktriangle_{\pi_{i-1}}, \neg a_{\pi_i})\big) \rightarrow  \ltlN \mathit{before}(\neg a_{\pi_i}, \bullet_{\pi_i})  \Big) \,\land \\
			&\ltlg\Big( \big(\neg a_{\pi_i} \land \blacktriangle_{\pi_{i-1}} \land \mathit{next}(\blacktriangle_{\pi_{i-1}}, a_{\pi_i})\big) \rightarrow \ltlN \big(a_{\pi_i} \ltlU \bullet_{\pi_i}\big) \Big) \, \land \\
			&\ltlg\Big( \big(a_{\pi_i} \land \blacktriangle_{\pi_{i-1}} \land \mathit{next}(\blacktriangle_{\pi_{i-1}}, \neg a_{\pi_i})\big) \rightarrow  \ltlN\big(a_{\pi_i} \ltlU  \bullet_{\pi_i}\big) \Big) \, \land \label{eq:ind4}
		\end{split}\\
		&\Big(\mathit{exactlyOnce}(\mathit{error}_{\pi_i}) \land \mathit{exactlyOnce}(\mathit{errorStart}_{\pi_{i-1}})  \rightarrow \mathit{falseAlarm}_i \Big)\bigg) \label{eq:ind5}
	\end{align}
	\endgroup
\end{formulaBox}

\noindent
We again discuss each conjunct separately.
Let $\tilde{n} := \mathcal{C}(i-1, n)$ be the length of the smaller yardstick on $\pi_{i-1}$

\begin{enumerate}[leftmargin=1.1cm]
	\item[(\ref{eq:ind1})] Formula \ref{eq:ind1} ensures the basic implications between $\bullet, \triangle$ and $\blacktriangle$.
	\item[(\ref{eq:ind11})] Formula \ref{eq:ind11} ensures that $\bullet_{\pi_i}$ holds every $\tilde{n}$ steps.
	Different from Formula \ref{eq:base1} we cannot express this using $\ltlN$s. Instead, we use the smaller counter that will be generated on $\pi_{i-1}$.
	On $\pi_{i-1}$, \refu{} can start the counter at any possible time (by looping in $s_1$).
	Formula \ref{eq:ind11} now states that if $\blacktriangle_{\pi_{i-1}}$ holds at the same time as $\bullet_{\pi_i}$, then $\bullet_{\pi_i}$ should hold the next time $\blacktriangle_{\pi_{i-1}}$ holds (which is $\tilde{n}$ steps apart if the counter on $\pi_{i-1}$ is correct).
	As \refu{} can start the counter on $\pi_{i-1}$ at any time and $\pi_{i-1}$ is resolved \emph{after} $\pi_i$ is fixed, any winning strategy for \veri{} on $\pi_i$ must set $\bullet_{\pi_i}$ exactly $\tilde{n}$ steps apart.
	
	\item[(\ref{eq:ind2})] Formula \ref{eq:ind2} ensures that $\triangle_{\pi_i}$ holds iff the next $a$-configuration is $0$.
	
	\item[(\ref{eq:ind3})] Formula \ref{eq:ind3} ensures that $\blacktriangle_{\pi_i}$ holds iff the next $b$-configuration is $0$.
	
	\item[(\ref{eq:ind4})] Formula \ref{eq:ind4} ensures the correctness of the $a$-counter. 
	We split this into 4 separate conditions, and, similar to Formula \ref{eq:ind11}, use the counter on $\pi_{i-1}$ to compare positions that are $\tilde{n}$ steps apart.
	Here 
	\begin{align*}
		\mathit{next}(\varphi, \psi) := (\neg \varphi) \ltlU (\varphi \land \psi)
	\end{align*}
	expresses that $\psi$ should hold at the next occurrence of $\varphi$.
	For example, the first line of Formula \ref{eq:ind4} covers the following case: If both $a_{\pi_i}$ and $\blacktriangle_{\pi_{i-1}}$ hold now and the next time $\blacktriangle_{\pi_{i-1}}$ holds, $a_{\pi_i}$ also holds (so the current bit of the current $a$-configuration does not change), then there should be a $0$-bit before the end of the current $a$-configuration (\cf~\Cref{eq:counter}).
	Similarly, if the value of $a_{\pi_i}$ now and the next time $\blacktriangle_{\pi_{i-1}}$ holds is different, then there should not be a $0$-bit before the end of the current $a$-configuration (expressed in the last two cases of Formula \ref{eq:ind4}).
	As \refu{} can start the counter on $\pi_{i-1}$ at any time, any bits in consecutive $a$-configurations on $\pi_i$ can be compared, so the only winning strategy for \veri{} on $\pi_i$ is to produce a correct $a$-counter. 
	
	\item[(\ref{eq:ind5})] Ensuring the correctness of the $b$-counter is more challenging.
	We again let \refu{} challenge the correctness of the $b$-counter by setting $\mathit{error}$.
	However, different from Formula \ref{eq:base5}, we cannot directly identify the same position in the previous $b$-configuration (note that the construction of $\mathit{prevPos}$ depends on $n$).
	Instead, if $\mathit{error}_{\pi_i}$ is set by \refu{}, \refu{} is also responsible for identifying the same position in the \emph{previous} $b$-configuration by setting $\mathit{errorStart}_{\pi_{i-1}}$ on $\pi_{i-1}$ (which he can do by moving the game producing $\pi_{i-1}$ to state $s_2$, see \Cref{fig:gameStructure}).
	We can then compare the two bits of the $b$-counter pointed to by $\mathit{errorStart}_{\pi_{i-1}}$ and $\mathit{error}_{\pi_i}$.
	As the position where $\mathit{errorStart}_{\pi_{i-1}}$ is set is determined by \refu{} we additionally need to check that the position is correct, \ie, corresponds to the same position within the previous $b$-configuration. 
	We define formula $\mathit{falseAlarm}_i$ as follows:
	
	\begin{formulaBox}{$\mathit{falseAlarm}_i$}
		\small
		\begin{align*}
			\mathit{counterIsCorrect}_i \lor \mathit{wrongConfiguration}_i \lor \mathit{wrongBit}_i
		\end{align*}
	\end{formulaBox}

	\noindent
	We discuss each disjunct of $\mathit{falseAlarm}_i$ separately. Note that only one of these disjuncts needs to hold in order to show that that the supposed error identified by $\refu$ is not genuine. 
	
	\begin{itemize}
		\item $\mathit{counterIsCorrect}_i$ expresses that the two positions pointed to by $\mathit{errorStart}_{\pi_{i-1}}$ and $\mathit{error}_{\pi_i}$ are correct, \ie, \refu{} did not point to an actual error in the $b$-counter. We define it as follows:
		
		\begin{formulaBox}{$\mathit{counterIsCorrect}_i$}
			\small
			\begin{align*}
				\begin{split}
					&\Big(\ltlg \big(\mathit{errorStart}_{\pi_{i-1}} \rightarrow b_{\pi_i}\big) \leftrightarrow  \ltlg \big(\mathit{error}_{\pi_{i}} \rightarrow b_{\pi_i}\big)\Big)\\ 
					&\leftrightarrow \Big(\ltlg \big(\mathit{errorStart}_{\pi_{i-1}} \rightarrow  \ltlN \mathit{before}(\bullet_{\pi_i} \land \neg b_{\pi_i}, \triangle_{\pi_i} )\big)\Big)
				\end{split}
		\end{align*}
	\end{formulaBox}
	
		\noindent
			The formula is similar to $\mathit{falseAlarm}_1$, but uses $\mathit{errorStart}_{\pi_{i-1}}$ instead of $\mathit{prevPos}$ to point to the same position in the previous $b$-configuration.

		\item $\mathit{wrongConfiguration}_i$ expresses that $\mathit{errorStart}_{\pi_{i-1}}$ and $\mathit{error}_{\pi_{i}}$ do not occur in two consecutive $b$-configurations on $\pi_i$, \ie, there is not exactly one $\triangle_{\pi_i}$ between both.
		
		\begin{formulaBox}{$\mathit{wrongConfiguration}_i $}
			\small
			\begin{align*}
				\ltlG\Big(\mathit{errorStart}_{\pi_{i-1}} \rightarrow \neg \Big((\neg \triangle_{\pi_i} )\ltlU \big(\triangle_{\pi_i} \land \ltlN ((\neg \triangle_{\pi_i}) \ltlU \mathit{error}_{\pi_i})\big)\Big)\Big) 
			\end{align*}
		\end{formulaBox}
		
		\item $\mathit{wrongBit}_i$ expresses that $\mathit{errorStart}_{\pi_{i-1}}$ and $\mathit{error}_{\pi_{i}}$ do not point to the same bit in the two consecutive $b$-configurations.
		We again use the fact that a bit of the $b$-counter is precisely characterized by the $a$-configuration that starts at the bit.
		If \refu{} moved the game producing $\pi_{i-1}$ to $s_3$ (which he did as $\mathit{errorStart}_{\pi_{i-1}}$ occurs), \veri{} can loop in state $s_4$ and decide when to start the counter.
		To show that $\mathit{errorStart}_{\pi_{i-1}}$ and $\mathit{error}_{\pi_{i}}$ point to different $b$-bits on $\pi_i$, \veri{} should loop in $s_4$ and find a bit position at which the $a$-configuration that starts at position $\mathit{errorStart}_{\pi_{i-1}}$ and the $a$-configuration that starts at position $\mathit{error}_{\pi_{i}}$ differ.		
		 We define $\mathit{wrongBit}_i$ as follows:
		 
		 \begin{formulaBox}{$\mathit{wrongBit}_i$}
		 	\small
		 	\begin{align}
		 		&\ltlG\Big(\mathit{errorStart}_{\pi_{i-1}} \rightarrow \mathit{before}(\blacktriangle_{\pi_{i-1}}, \bullet_{\pi_i})  \Big) \, \land \label{eq:wrongBit1}\\
		 		\begin{split}
		 			&\neg\bigg(\Big(\ltlg \big(\mathit{errorStart}_{\pi_{i-1}} \rightarrow \mathit{next}(\blacktriangle_{\pi_{i-1}}, a_{\pi_{i}})\big) \Big) \\
		 			&\quad\quad\leftrightarrow \Big(\ltlg \big(\mathit{error}_{\pi_i} \rightarrow \mathit{next}(\blacktriangle_{\pi_{i-1}}, a_{\pi_{i}}) \big)\Big)\bigg)\label{eq:wrongBit2}
		 		\end{split}
		 	\end{align}
		 \end{formulaBox}
	 
	 	\noindent
		 Formula \ref{eq:wrongBit1} ensures that $\veri$ starts the counter soon enough by leaving $s_4$, \ie, after $\mathit{errorStart}_{\pi_{i-1}}$ is set, the counter on $\pi_{i-1}$ is started within the same $a$-configuration on $\pi_i$ (before $\bullet_{\pi_i}$ holds).
		 Formula \ref{eq:wrongBit2} states that $\veri$ started the counter at a time that shows that $\refu$ set $\mathit{errorStart}_{\pi_{i-1}}$ at a wrong location (\ie, did not point to a genuine error).
		 That is, the bit of the $a$-configuration where $\blacktriangle_{\pi_{i-1}}$ holds (for the first time) after $\mathit{errorStart}_{\pi_{i-1}}$ differs from the value of the $a$-configuration where $\blacktriangle_{\pi_{i-1}}$ holds (for the first time) after $\mathit{error}_{\pi_i}$.
	\end{itemize}
\end{enumerate}

\subsection{Lower Bound Proofs}

We use the counter construction to prove \refTheo{lb} and \refTheo{lb2}.

\lb*
\begin{proof}
	In the case where $k = 0$, the \HyperATLS{} formula is a \HyperLTL{} formula and we can reuse the $(l-1)$-\EXPSPACE{} lower bound shown by Rabe \cite{Rabe16}.
	So let us assume that $k \geq 1$.
	We distinguish if $l \geq 1$ or $l = 0$.
	
	\begin{itemize}
		\item If $l \geq 1$:
			The counter construction in \Cref{sec:counterConst1} and \Cref{sec:counterConst2} gives us a formula of the form
			\begin{align*}
				\textstyle\llangle \veri\rrangle \pi_k \ldots \llangle \veri\rrangle \pi_1\ldot \bigwedge_{i=1}^k \mathit{correct}_i
			\end{align*}
			and a game structure $\calG_\mathit{counter}$ (note that the size of $\calG_\mathit{counter}$ is independent of $n$), such that \refu{} can start the counter on path $\pi_k$ at any time and, once started, \veri{} needs to produce a correct counter, \ie,  $\blacktriangle_{\pi_k}$ holds every $\mathcal{C}(k, n)$ steps.
			Given a $\mathcal{C}(k, n)$-space-bounded Turing machine $\mathcal{T}$ and an input $w$ (with $|w| = n$), we can design a formula  of the form
			\begin{align}
				\exists \pi\ldot \llangle\veri \rrangle \pi_k \ldots \llangle \veri \rrangle \pi_1\ldot \psi_{(\calT, w)}\label{eq:finalConstruction}
			\end{align}
			and a game structure $\calG_{\calT}$ (the size of which is constant and does not depend on $w$), such that $\calG_\calT \models (\ref{eq:finalConstruction})$ iff $\calT$ accepts $w$.
			The idea of this encoding is similar to \cite{Rabe16}: 
			First, formula $\psi_{(\calT, w)}$ contains $\bigwedge_{i=1}^k \mathit{correct}_i$ as a conjunct to ensure that the counters on paths $\pi_1, \ldots, \pi_k$ are correct.
			In addition, the path $\pi$ should enumerate consecutive configurations of $\calT$ (each of which is $\mathcal{C}(k, n)$ steps long).
			The initial configuration should contain the input $w$ (which we simply hard-code in the formula).
			Using the yardstick (which \refu{} can start at any time), we can compare positions which are $\mathcal{C}(k, n)$-steps apart and -- as transitions of a TM are local -- enforce that $\pi$ encodes a valid accepting computation (see \cite[Lemma 5.6.3]{Rabe16} for details).
			Model checking of a formula of the form (\ref{eq:finalConstruction}) (with $k$ complex and $1$ simple quantifier) is thus $2k$-\EXPSPACE-hard (Note that $\mathcal{C}(k, n) \geq \mathit{Tower}_2(2k, n)$. We can scale the counter to an arbitrary base $c > 2$ by using $\lceil \log_2 c \rceil$ propositions for the counter).
			In cases of more than a single simple quantifier, we can construct a larger yardstick by adding the construction of Rabe \cite{Rabe16} (which extends the length by one exponent with each simple quantifier) to ours (which extends the length by two exponents with each complex quantifier).
			See \cite[Lemma 5.6.2]{Rabe16} for details.

		\item If $l = 0$:
			Similar to the previous case, we use our counter construction.
			Assume we are given a $2^{\mathcal{C}(k-1, n)}$-space-bounded Turing machine $\calT$ and an input $w$ (with $|w| = n$).
			We design a formula  
			\begin{align}
				\llangle \veri \rrangle \pi\ldot \llangle \veri \rrangle \pi_{k-1} \ldots \llangle \veri \rrangle \pi_1\ldot \psi_{(\calT, w)}
			\end{align}
			and game structure such that $\veri$ on path $\pi$ should produce a correct counter with $\mathcal{C}(k-1, n)$ many bits (similar to the $a$-counter) and in place of the $b$-counter output configurations of $\calT$ (so each configuration has length $2^{\mathcal{C}(k-1, n)}$).
			We use the yardstick of length $\mathcal{C}(k-1, n)$ on $\pi_{k-1}$ to verify the correctness of the $a$-counter on $\pi$ (if $k = 1$, we verify it directly using $\ltlN$s).
			We verify that consecutive configurations of $\calT$ are correct similar to the verification of the $b$-counter via the $\mathit{error}$ proposition (which now points to errors in the TM configurations as opposed to errors in the $b$-counter).
			Model checking a formula with $k$ complex quantifiers is thus $(2k-1)$-\EXPSPACE{}-hard.
			\qedhere
	\end{itemize}
\end{proof}

\lbb*
\begin{proof}
	If $k = 0$, we get an $(l-2)$-\EXPSPACE{} lower bound from \HyperLTL{} model checking hardness \cite{Rabe16} which is even better than the $(l-3)$-\EXPSPACE{} bound required.
	So let us assume that $k \geq 1$.
	We distinguish if $l \geq 1$ or $l = 0$.
	
	\begin{itemize}
		\item If $l \geq 1$:
			We first observe that we can construct a formula of the form
			\begin{align*}
				\textstyle\llangle \veri \rrangle \pi_k \ldots \llangle \veri \rrangle \pi_1\ldot \bigwedge_{i=1}^k \mathit{correct}'_i
			\end{align*}
			of constant size and a CGS $\calG$ (whose size depends on $n$) such that $\veri$ is required to output a yardstick of length $\mathcal{C}(k-1, n)$ on $\pi_k$.
			The construction is similar to the counter defined in \Cref{sec:counterConst1} and \Cref{sec:counterConst2} but modifies $\mathit{correct}_1$ (note that the size of $\mathit{correct}_1$ depends on $n$).  
			We ensure that $\pi_1$ no longer produces a yardstick of length $\mathcal{C}(1, n)$ but only of length $n$.
			We modify the game structure such that $n$ is hard-coded (\ie, $\blacktriangle$ occurs exactly every $n$ steps) and can ensure the correctness of the $n$-bit counter between two $\blacktriangle$s with a formula that does not depend on $n$.
			Similar to \refTheo{lb}, we can then encode the acceptance of a $\mathcal{C}(k-1, n)$-space bounded TM $\calT$ on input $w$  (with $|w|=n$) as a formula
			\begin{align*}
				\exists \pi\ldot \llangle \veri \rrangle \pi_k \ldots \llangle \veri \rrangle \pi_1\ldot \psi_{(\calT, w)}
			\end{align*}
			whose size does not depend on the size of input $w$ (the input is hard-coded in the game structure).
			Verification of a formula of the above form (with $k$ complex and $1$ simple quantifier) is thus $(2k-2)$-\EXPSPACE-hard in the size of the system. 
			In case of more than a single simple quantifier, we, again, use the construction of Rabe \cite{Rabe16} to extend the yardstick.
			
		\item If $l = 0$:
			This is analogous to the second case in the proof of \refTheo{lb}.
			To ensure that the size of the formula is independent of $n$, we again let $\pi_1$ only produce a counter of length $n$ (compared to $\mathcal{C}(1, n)$ in the proof of \refTheo{lb}).
			\qedhere
	\end{itemize}
	
\end{proof}

\section{Experimental Evaluation}\label{sec:proto}

As indicated by our lower bounds, \HyperATLS{} model checking for the full logic is not practical.
Instead, we focus on formulas of the form $[\llangle A_1 \rrangle \pi_1\ldots \llangle A_n\rrangle \pi_n]~\psi$ where $\psi$ is quantifier-free.
In terms of model checking complexity, this fragment is much cheaper than full \HyperATLS{}; it is $2$-\EXPTIME{}-complete in the size of the specification (as the fragment still subsumes \LTL{} realizability) and \PTIME{}-complete in the size of the system.
This fragment of \HyperATLS{} is expressive:
It subsumes alternation-free \HyperLTL{} specifications, many security specifications such as simulation-based security (\Cref{sec:simSec}), the game-based model checking approximation from \cite{CoenenFST19} (\Cref{sec:coenen}), and the asynchronous approach from \Cref{sec:async}.
Model checking a formula in the above fragment can be reduced to the solving of a parity game by building the product of the game structure with a deterministic parity automaton for $\psi$.
We account for the order of move selection by simulating a single step in the CGS with multiple intermediate steps where all agents in the same stage fix their move incrementally. 

We have implemented this construction in \texttt{hyperatlmc}, a prototype model checker for \HyperATLS{} formulas of the form $[\llangle A_1 \rrangle \pi_1\ldots \llangle A_n\rrangle \pi_n]~\psi$ where $\psi$ is quantifier-free.
Our tool uses \texttt{rabinizer4} \cite{KretinskyMSZ18} to compute deterministic parity automata and \texttt{pgsolver} \cite{FriedmannL09} to solve parity games. 

In this section, we give a simple operational semantics for a boolean programming language into CGSs (in \Cref{sec:semantics}). 
This allows us to check the (synchronous and asynchronous) security properties from \Cref{sec:examples1,sec:examples2} (which are stated at the level of CGSs) on programs.
Afterward, we report on experiments with \texttt{hyperatlmc} in \Cref{sec:experiments}.

\subsection{Compiling Programs into Game Structures}\label{sec:semantics}

\begin{figure}
	\small
	\begin{minipage}{0.33\textwidth}
		\vspace{2.5mm}
		\begin{prooftree}
			\AxiomC{}
			\UnaryInfC{$\langle x \leftarrow e, \mu \rangle \leadsto \langle \flat, \mu[x \mapsto \llbracket e \rrbracket(\mu)] \rangle$}
		\end{prooftree}
	\end{minipage}%
	\begin{minipage}{0.33\textwidth}
		\begin{prooftree}
			\AxiomC{$\llbracket e \rrbracket(\mu) = \top$}
			\UnaryInfC{$\langle \texttt{if}(e, P_1, P_2), \mu \rangle \leadsto \langle P_1, \mu \rangle$}
		\end{prooftree}
	\end{minipage}%
	\begin{minipage}{0.33\textwidth}
		\begin{prooftree}
			\AxiomC{$\llbracket e \rrbracket(\mu) = \bot$}
			\UnaryInfC{$\langle\texttt{if}(e, P_1, P_2), \mu \rangle \leadsto \langle P_2, \mu \rangle$}
		\end{prooftree}
	\end{minipage}
	
	\vspace{0.3cm}
	
	\begin{minipage}{0.33\textwidth}
		\begin{prooftree}
			\AxiomC{$b \in \mathbb{B}$}
			\AxiomC{$P \in \{L, H\}$}
			\BinaryInfC{$\langle x \leftarrow \texttt{Read}_P, \mu \rangle \leadsto \langle \flat, \mu[x \mapsto b] \rangle$}
		\end{prooftree}
	\end{minipage}%
	\begin{minipage}{0.33\textwidth}
		\vspace{3.7mm}
		\begin{prooftree}
			\AxiomC{}
			\UnaryInfC{$\langle P_1 \oplus P_2, \mu \rangle \leadsto \langle P_1, \mu\rangle$}
		\end{prooftree}
	\end{minipage}%
	\begin{minipage}{0.33\textwidth}
		\vspace{3.7mm}
		\begin{prooftree}
			\AxiomC{}
			\UnaryInfC{$\langle P_1 \oplus P_2, \mu \rangle \leadsto \langle P_2, \mu \rangle$}
		\end{prooftree}
	\end{minipage}
	
	\vspace{0.3cm}
	
	\begin{minipage}{0.5\textwidth}
		\begin{prooftree}
			\AxiomC{$\llbracket e \rrbracket(\mu) = \bot$}
			\UnaryInfC{$\langle \texttt{while}(e, P), \mu \rangle \leadsto \langle \flat, \mu \rangle$}
		\end{prooftree}
	\end{minipage}%
	\begin{minipage}{0.5\textwidth}
		\begin{prooftree}
			\AxiomC{$\llbracket e \rrbracket(\mu) = \top$}
			\UnaryInfC{$\langle \texttt{while}(e, P), \mu \rangle \leadsto \langle P;~\texttt{while}(e, P), \mu \rangle$}
		\end{prooftree}
	\end{minipage}
	
	\vspace{0.3cm}
	
	\begin{minipage}{0.33\textwidth}
		\begin{prooftree}
			\AxiomC{$\langle P_1, \mu \rangle \leadsto \langle \flat, \mu' \rangle$}
			\UnaryInfC{$\langle P_1 ; P_2, \mu \rangle \leadsto \langle P_2, \mu' \rangle$}
		\end{prooftree}
	\end{minipage}%
	\begin{minipage}{0.33\textwidth}
		\begin{prooftree}
			\AxiomC{$\langle P_1, \mu \rangle \leadsto \langle P_1', \mu' \rangle$}
			\AxiomC{$P_1' \neq \flat$}
			\BinaryInfC{$\langle P_1 ; P_2, \mu \rangle \leadsto \langle P_1';P_2, \mu' \rangle$}
		\end{prooftree}
	\end{minipage}%
	\begin{minipage}{0.33\textwidth}
		\vspace{3mm}
		\begin{prooftree}
			\AxiomC{}
			\UnaryInfC{$\langle \flat, \mu \rangle \leadsto \langle \flat, \mu \rangle$}
		\end{prooftree}
	\end{minipage}
	
	\caption{Small-step semantics for \texttt{bwhile}. A step has the form $\langle P, \mu \rangle \leadsto \langle P', \mu' \rangle$ where program $P$ in memory $\mu$ steps to program $P'$ and memory $\mu'$. For a boolean expression $e$ and memory $\mu$ we write $\llbracket e \rrbracket(\mu) \in \mathbb{B}$ for the value of $e$ in $\mu$ (defined as expected).} \label{fig:sem}
\end{figure}

To have a fixed language to express programs, we use a simple toy-language we call \texttt{bwhile}.
We endow \texttt{bwhile} programs with a direct semantics into a game structures over players $\agent_L, \agent_H$, and $\agent_N$ (\cf~\Cref{sec:examples1}), which allows us to apply the properties given in \Cref{sec:examples1,sec:examples2}.

We fix a finite set of program variables $\calX$ and define boolean expression as follows:
\begin{align*}
	e &:= x \mid \mathit{true} \mid \mathit{false} \mid e_1 \land e_2 \mid e_1 \lor e_2 \mid \neg e
\end{align*}
where $x \in \calX$. 

\noindent
\texttt{bwhile} programs are then generated by the following grammar:
\begin{align*}
	P &:= x \leftarrow e \mid x \leftarrow \texttt{Read}_H \mid x \leftarrow \texttt{Read}_L \mid \texttt{if}(e, P_1, P_2) \mid P_1 \oplus P_2 \mid \texttt{while}(e, P) \mid P_1 ; P_2 \mid \flat
\end{align*}
where $x \in \calX$.
Most language constructs are standard:
$x \leftarrow \texttt{Read}_H$ (resp.~$x \leftarrow \texttt{Read}_L$) reads the value of $x$ from a high-security (resp.~low security) source, $P_1 \oplus P_2$ is a nondeterministic choice between $P_1$ and $P_2$ and $\flat$ is the terminated program.
We endow our language with a standard small-step semantics operating on configurations of the form $\langle P, \mu \rangle$ where $P$ is a program and $\mu : \calX \to \mathbb{B}$ a memory.
The reduction steps are standard; we give them in \Cref{fig:sem} for completeness.
To obtain a game structure, we associated each program $P$ to a player $\llparenthesis P \rrparenthesis \in \{\agent_N, \agent_H, \agent_L\}$, where player $\llparenthesis P \rrparenthesis$ decides on the successor state of $P$.
We define $\llparenthesis P \rrparenthesis$ as follows:

\noindent
\begin{minipage}{0.25\textwidth}
	\begin{align*}
		\llparenthesis x \leftarrow \texttt{Read}_L \rrparenthesis &:= \agent_L\\
		\llparenthesis x \leftarrow e \rrparenthesis &:= \agent_N
 \end{align*}
\end{minipage}%
\begin{minipage}{0.25\textwidth}
	\begin{align*}
		\llparenthesis x \leftarrow \texttt{Read}_H \rrparenthesis &:= \agent_H \\
		\llparenthesis \texttt{if}(e, P_1, P_2) \rrparenthesis &:= \agent_N
	\end{align*}
\end{minipage}%
\begin{minipage}{0.27\textwidth}
	\begin{align*}
		 \llparenthesis P_1 \oplus P_2 \rrparenthesis &:= \agent_N\\
		\llparenthesis \texttt{while}(e, P)\rrparenthesis &:= \agent_N
	\end{align*}
\end{minipage}%
\begin{minipage}{0.25\textwidth}
	\begin{align*}
		 \llparenthesis P_1;P_2 \rrparenthesis &:= \llparenthesis P_1 \rrparenthesis \\
		\llparenthesis \flat \rrparenthesis &:= \agent_N
	\end{align*}
\end{minipage}

\vspace{2mm}
\noindent
Given a program $\dot{P}$, we define the game structure $\calG_{\dot{P}}$ over agents $\{\agent_N, \agent_H, \agent_L\}$ as follows:
The states of $\calG_{\dot{P}}$ are all configurations $\langle P, \mu\rangle$ where $P$ is a program and $\mu$ a memory.
The initial state is $\langle \dot{P}, \lambda \_. \bot\rangle$ (\ie, the initial memory assigns all variables to $\bot$).
In state $\langle P, \mu \rangle$, player $\llparenthesis P \rrparenthesis$ decides on a successor state from the set $\{\langle P', \mu' \rangle \mid \langle P, \mu \rangle \leadsto \langle P', \mu' \rangle\}$.\footnote{Note that for all constructs except $P_1 \oplus P_2, x \leftarrow \texttt{Read}_L, x \leftarrow \texttt{Read}_H$, and $P_1;P_2$ there is unique successor configuration (so the player is irrelevant).
	As expected, $\agent_L$ chooses the successor of a program $x \leftarrow \texttt{Read}_L$ and thereby fixes the next value of $x$ (and similarly for $\agent_H$ and $x \leftarrow \texttt{Read}_H$).
	In a non-deterministic branching $P_1 \oplus P_2$, player $\agent_N$ decides which branch to take. }
$\calG_{\dot{P}}$ is a turn-based game in the sense of \cite{AlurHK02}.
Note that the state-space of $\calG_{\dot{P}}$ is infinite (as there are infinitely many programs), but the reachable fragment is finite and computable. 
The atomic propositions in $\calG_{\dot{P}}$ are all variables from $\calX$ that are used in $\dot{P}$.
An atomic proposition (variable) $x \in \calX$ holds in state $\langle P, \mu \rangle$ iff $\mu(x) = \top$.

\begin{figure}
	
	\begin{tikzpicture}
		\node[align=center, draw, very thick] () at (0.5, -0.4) {\texttt{P1}};
		
		\node[align=center, draw, very thick] () at (0.5, -2.3) {\texttt{P2}};
		
		\node[align=center, draw, very thick] () at (4.5, -0.4) {\texttt{P3}};
		
		\node[align=center, draw, very thick] () at (8.5, -0.4) {\texttt{P4}};

		\node[align=left,anchor=north west] (n0) at (1,-0.1) {\myvar{o} $\leftarrow$ \myconst{true}\\
			\mycontrol{while}(\myconst{true})\\
			\makebox[0.3cm]{} \myvar{o} $\leftarrow$ $\neg$\myvar{o}};

		\node[align=left,anchor=north west] (n0) at (1,-2) {
			\myvar{o} $\leftarrow$ \myconst{true}\\
			\mycontrol{while}(\myconst{true})\\
			\makebox[0.3cm]{} \myvar{l} $\leftarrow$ \mycontrol{Read}$_L$\\
			\makebox[0.3cm]{} \myvar{o} $\leftarrow$ \myvar{l}
		};
		
		\node[align=left,anchor=north west] (n0) at (5,-0.1) {
			\myvar{o} $\leftarrow$ \myconst{true}\\
			\mycontrol{while}(\myconst{true})\\
			\makebox[0.3cm]{} \myvar{l} $\leftarrow$ \mycontrol{Read}$_L$\\
			\makebox[0.3cm]{} \mycontrol{if} $\star$ \mycontrol{then}\\
			\makebox[0.6cm]{} \myvar{o} $\leftarrow$ \myvar{l}\\
			\makebox[0.3cm]{} \mycontrol{else}\\
			\makebox[0.6cm]{} \myvar{o} $\leftarrow$ $\neg$\myvar{l}
		};
		
		\node[align=left,anchor=north west] (n0) at (9,-0.1) {
			\myvar{o} $\leftarrow$ \myconst{true}\\
			\mycontrol{while}(\myconst{true})\\
			\makebox[0.3cm]{} \mycontrol{if} $\star$ \mycontrol{then}\\
			\makebox[0.6cm]{} \myvar{h} $\leftarrow$ \mycontrol{Read}$_H$\\
			\makebox[0.6cm]{} \myvar{o} $\leftarrow$ \myvar{h}\\
			\makebox[0.3cm]{} \mycontrol{else}\\
			\makebox[0.6cm]{} \myvar{h} $\leftarrow$ \mycontrol{Read}$_H$\\
			\makebox[0.6cm]{} \myvar{o} $\leftarrow$ $\neg$\myvar{h}
		};
	\end{tikzpicture}
	
	\caption{Simple \texttt{bwhile} example programs that distinguish different information-flow policies. We write \texttt{if} $\star$ \texttt{then} $P_1$ \texttt{else} $P_2$ instead of $P_1 \oplus P_2$. }\label{fig:programs}
\end{figure}

\subsection{Experiments}\label{sec:experiments}

We applied \texttt{hyperatlmc} to small \texttt{bwhile} programs and checked synchronous and asynchronous information flow policies.

\begin{table}
	\caption{Model checking results for information-flow properties (expressed in \HyperATLS{}) on the small \texttt{bwhile} programs from \Cref{fig:programs}. 
		We give the model checking result (Res) (\cmark{} indicates that the property holds, and \xmark{} that it is violated) and time  taken by \texttt{hyperatlmc} in milliseconds ($t$). }\label{tab:res1}
	\small
	\renewcommand{\arraystretch}{1.3}
	\begin{tabular}{l@{\hspace{7mm}}ll@{\hspace{7mm}}ll@{\hspace{7mm}}ll@{\hspace{7mm}}ll}
		\toprule[1pt]
		& \multicolumn{2}{@{}c@{\hspace{7mm}}}{\textbf{(OD)}} & \multicolumn{2}{@{}c@{\hspace{7mm}}}{\textbf{(NI)}} & \multicolumn{2}{@{}c@{\hspace{8mm}}}{\textbf{(simSec)}} & \multicolumn{2}{@{}c@{\hspace{1mm}}}{\textbf{(aproxGNI$_3$)}} \\
		\cmidrule[1pt](l{-1mm}r{5mm}){2-3}
		\cmidrule[1pt](l{-1mm}r{5mm}){4-5}
		\cmidrule[1pt](l{-1mm}r{5mm}){6-7}
		\cmidrule[1pt](l{-1mm}){8-9}
		\textbf{Program} & \textbf{Res} & $\boldsymbol{t}$ & \textbf{Res} & $\boldsymbol{t}$ & \textbf{Res} & $\boldsymbol{t}$ & \textbf{Res} & $\boldsymbol{t}$ \\
		\cmidrule[1pt](l{0mm}r{6mm}){1-1}
		\cmidrule[1pt](l{-1mm}r{5mm}){2-3}
		\cmidrule[1pt](l{-1mm}r{5mm}){4-5}
		\cmidrule[1pt](l{-1mm}r{5mm}){6-7}
		\cmidrule[1pt](l{-1mm}){8-9}
		\texttt{P1} & \cmark & $12$ & \cmark & $11$ & \cmark & $12$ & \cmark & $31$ \\
		\texttt{P2} & \xmark & $41$ & \cmark & $33$ & \cmark & $34$ & \cmark & $112$\\
		\texttt{P3} & \xmark & $21$ & \xmark & $33$ & \cmark & $31$ & \cmark & $88$ \\
		\texttt{P4} & \xmark & $57$ & \xmark & $41$ & \xmark & $54$ & \cmark & $123$ \\
		\bottomrule[1pt]
	\end{tabular}
	
\end{table}

\paragraph{Information-Flow Policies}

We created a small benchmark of simple programs that distinguish different synchronous information-flow policies. See \Cref{fig:programs}.\footnote{We choose very simple programs to easily distinguish between the different security notions. Our tool \texttt{hyperatlmc} can handle more complex programs with larger bitwidths. }
We checked the following properties:
\textbf{(OD)} is the observational determinism property stating that the output is identical among all paths, \ie,$\forall \pi\ldot \forall \pi'\ldot \ltlG (o_\pi \leftrightarrow o_{\pi'})$.
\textbf{(NI)} is a simple formulation of non-interference due to Goguen and Meseguer  \cite{GoguenM82a} that states that the output is fully determined by the low-security inputs, \ie, $\forall \pi\ldot \forall \pi'\ldot \ltlG (l_\pi \leftrightarrow l_{\pi'}) \to \ltlG (o_\pi \leftrightarrow o_{\pi'})$.
\textbf{(simSec)} is simulation-based security as discussed in \Cref{sec:simSec}.
\textbf{(aproxGNI$_3$)} is the approximation of generalized non-interference with fixed lookahead of $3$ as discussed in  \Cref{sec:coenen}.
The results and running times for each instance (obtained using \texttt{hyperatlmc}) are given in \refTable{res1}.

\begin{table}
	\caption{Model checking results of synchronous and asynchronous properties (expressed in \HyperATLS{}) on small \texttt{bwhile} programs. 
		We give the model checking result (Res) (\cmark{} indicates that the property holds, and \xmark{} that it is violated) and time  taken by \texttt{hyperatlmc} in milliseconds ($t$). 
		Program \texttt{Q1} is the program from \Cref{fig:exProg}, and \texttt{Q2} is a slight modification that sets the output according to the low-security input but delays this update.}\label{tab:res2}
	\small
	\renewcommand{\arraystretch}{1.3}
	\begin{tabular}{l@{\hspace{7mm}}ll@{\hspace{7mm}}ll@{\hspace{7mm}}ll}
		\toprule[1pt]
		& \multicolumn{2}{@{}c@{\hspace{7mm}}}{\textbf{(OD)}} & \multicolumn{2}{@{}c@{\hspace{7mm}}}{\textbf{(OD$_{asynch}$)}} & \multicolumn{2}{@{}c@{\hspace{1mm}}}{\textbf{(NI$_{asynch}$)}} \\
		\cmidrule[1pt](l{-1mm}r{5mm}){2-3}
		\cmidrule[1pt](l{-1mm}r{5mm}){4-5}
		\cmidrule[1pt](l{-1mm}r{0mm}){6-7}
		\textbf{Program} & \textbf{Res} & $\boldsymbol{t}$ & \textbf{Res} & $\boldsymbol{t}$ & \textbf{Res} & $\boldsymbol{t}$  \\
		\cmidrule[1pt](l{0mm}r{6mm}){1-1}
		\cmidrule[1pt](l{-1mm}r{5mm}){2-3}
		\cmidrule[1pt](l{-1mm}r{5mm}){4-5}
		\cmidrule[1pt](l{-1mm}r{0mm}){6-7}
		\texttt{Q1} & \xmark & $54$ &\cmark & $97$ & \cmark & $121$ \\
		\texttt{Q2} & \xmark & $61$ &\xmark & $87$ & \cmark & $95$ \\
		\bottomrule[1pt]
	\end{tabular}
	
\end{table}

\paragraph{Asynchronous Hyperproperties}
Our model checker implements the transformation of a game structure to include an asynchronous scheduler (\cf~\refDef{stutter}). 
Using \texttt{hyperatlmc}, we checked synchronous observational-determinism \textbf{(OD)} and asynchronous versions of observational-determinism \textbf{(OD$_{asynch}$)} and non-interference \textbf{(NI)$_{asynch}$}.  
Note that while \textbf{(OD$_{asynch}$)} is expressible in the decidable fragment of \AHLTL, \textbf{(NI$_{asynch}$)} is not an admissible formula (and cannot be handled in \cite{BaumeisterCBFS21}). 
As non-interference only requires the outputs to align provided the inputs do, one needs to take care that the asynchronous scheduler does not ``cheat'' by deliberately misaligning inputs and thereby invalidating the premise of this implication. 
Our results are given in \refTable{res2}.

\section{Related Work}\label{sec:related}

\paragraph{The Landscape of Hyperproperties}

There has been a lot of recent interest in logics for hyperproperties. 
Most logics are obtained by extending standard temporal or first-order/second-order logics with either path quantification or by a special equal-level predicate~\cite{Finkbeiner017}.   
See \cite{CoenenFHH19} for an overview.
To the best of our knowledge, none of these logics can express strategic hyperproperties in multi-agent systems.
In \cite{Finkbeiner0SZ17,0008SZ18}, the authors study the verification of first-order \HyperLTL{} on multi-agent workflow systems specified as first-order transition systems.
The logic that is used (first-order \HyperLTL{}) does not reason about the strategic behavior in the multi-agent systems.

\paragraph{Hyperproperties in Multi-Agent Systems}

The approach taken in \HyperATLS{} of resolving the paths that are quantified in the prefix \emph{incrementally} is only one possible angle to express hyperproperties in multi-agent systems.
One could also envision a logic, that can state the existence of a strategy \emph{with respect to} a hyperproperty, \ie, state the existence of a strategy such that the set of plays compatible with this strategy satisfies a hyperproperty.
Model checking of the resulting logic would subsume \HyperLTL{} realizability, which is known to be undecidable even for simple alternation-free formulas \cite{FinkbeinerHLST18}.
The incremental approach in \HyperATLS{} is restrictive enough to maintain decidable model checking and powerful enough to express many properties of interest and subsume many existing logics (see \Cref{fig:expr}).

\paragraph{Epistemic Logics}

The relationship between epistemic logics and hyperlogics is interesting, as both reason about the flow of information in a system. 
\HyperLTL{} and \LTL{}$_\mathcal{K}$ (\LTL~extended with a knowledge operator \cite{FHMV1995}) have incomparable expressiveness \cite{BozzelliMP15}.
In \texttt{HyperQPTL} \cite{Rabe16,BeutnerF23b} -- which extends \HyperLTL~with additional propositional quantification -- the knowledge operator can be encoded by explicitly marking the knowledge positions via propositional quantification \cite[\S 7]{Rabe16}.
By allowing second-order quantification, one can even reason about \emph{common} knowledge in a system \cite{HalpernM90,BeutnerFFM23}.
Alternating-time temporal logic has also been extended with knowledge operators \cite{HoekW03a}. The resulting logic, \texttt{ATEL}, can express properties of the form ``if $\agent$ knows $\phi$, then she can enforce $\psi$ via a strategy.''
The natural extension of \texttt{ATEL} that allows for arbitrary nesting of quantification and temporal operators (\ie, an extension of \ATLS{} instead of \ATL{}), is incomparable to \HyperATLS{}.

\paragraph{Model Checking}

Decidable model checking is a crucial prerequisite for the effective use of a logic. Many of the existing (synchronous) hyperlogics admit decidable finite-state model checking, although mostly with non-elementary complexity \cite{ClarksonFKMRS14}. 
Most hyperproperties encountered in practice can be expressed with few (if any) quantifier alternations. 
Alternation-free \HyperLTL{} properties can be checked very efficiently by constructing the self-composition \cite{BartheDR11}, as, \eg, implemented in the \texttt{MCHyper} tool \cite{FinkbeinerRS15}.
Properties with quantifier alternations can be checked by using automata complementations or language inclusion checks, as, \eg, implemented in the \texttt{AutoHyper} tool \cite{BeutnerF23}.
For properties in the $\forall^*\exists^*$ fragment, efficient approximations, such as the game-based approach \cite{CoenenFST19,BeutnerF22b}, are applicable, even in infinite-state systems \cite{BeutnerF22a}.
For alternating-time temporal logic (in the non-hyper realm), model checking is efficient, especially when temporal operators cannot be nested as in \ATL{} \cite{AlurHK02,AlurHMQRT98}. 
In the presence of arbitrary nesting (as in \ATLS{}), model checking subsumes \LTL{} realizability \cite{PnueliR89}. 
This causes a jump in the model checking complexity to $2$-\EXPTIME-completeness \cite{AlurHK02}.
ATL model checking has also been investigated in the presence of imperfect information \cite{JamrogaA06, abs-1102-4225,BullingGJ15,BerthonMM17}, and imperfect recall \cite{Schobbens04}.
Strategy logic \cite{ChatterjeeHP10,MogaveroMPV14} (strictly) generalizes \ATLS{} by considering strategies as first class objects that can be quantified. 
Model checking of strategy logic is decidable, but nonelementary-hard \cite{ChatterjeeHP10,MogaveroMPV14}.

Our lower bounds demonstrate that the combination of strategic quantification and hyperproperties results in a logic that is algorithmically harder (for model checking) than non-strategic hyperlogics (such as \HyperLTL) or strategic (non-hyper) logics (such as \ATLS). 
The fragment of \HyperATLS{} supported by \texttt{hyperatlmc} is algorithmically cheaper than full \HyperATLS{}; it is $2$-\EXPTIME{}-complete in the size of the specification and \PTIME{}-complete in the size of the system.

\paragraph{Satisfiability}

The satisfiability of a formula (\ie, checking if a formula has a satisfying model) is relevant during the development of a specification.
It can be used as a sanity check (to ensure that the specification is not already contradictory) or to determine implications between different specifications. 
The hardness of \HyperLTL{} satisfiability can be characterized in terms of the structure of the quantifier prefix. 
Satisfiability is decidable (and \EXPSPACE{}-complete) for formulas in the $\exists^*\forall^*$ fragment and undecidable for all prefixes that contain a $\forall\exists$ alternation \cite{FinkbeinerH16}.
For \HyperCTLS{}, alternation-free formulas (where the quantifier structure is defined with respect to the scope of quantifiers in a negation-normal form) are decidable \cite{Hahn21}.
However, already $\exists^*\forall^*$ formulas lead to undecidability as quantification can occur at all points along a path (by placing quantification below a $\ltlg$) and create a comb-like structure that can ``simulate'' a $\forall\exists$ alternation  \cite{Hahn21}.
Fortin et al.~show that satisfiability of \HyperLTL{} and \HyperCTLS{} is highly undecidable; deciding satisfiability of a \HyperLTL{} formula is $\Sigma_1^1$-complete and deciding satisfiability of a \HyperCTLS{} formula is $\Sigma_1^2$-complete \cite{FortinKT021}.
By restricting the body of a formula and distinguishing between hyperproperty and trace property, one can identify classes of \HyperLTL{} within the $\forall^*\exists^*$ fragment for which satisfiability remains decidable \cite{BeutnerCFHK22}.
\ATLS{} satisfiability was studied by Schewe \cite{Schewe08} and found to coincide with the model-checking complexity ($2$-\EXPTIME{}-complete in the size of the specification).
This is surprising as for most branching-time temporal logics (such as  \CTL{} and \CTLS{}), satisfiability is strictly (at least exponentially) harder than model checking (in the size of the specification). 
As \HyperATLS{} subsumes \HyperCTLS{}, it inherits the $\Sigma_1^2$-hardness of \HyperCTLS{} satisfiability.\footnote{As alternating-time logics are evaluated over game structures, the satisfiability problem can either be stated as the search for a set of agents (containing the agents refereed to in the formula) \emph{and} game structure over those agents or the search for a game structure given a fixed set of agents (as part of the input). See \cite{WaltherLWW06} for details in the context of \ATL{}. We assume that the set of agents is provided with the input.}
Identifying fragments of \HyperATLS{} that are decidable (or sit below the general $\Sigma_1^2$-hardness) is interesting future work.

\paragraph{Asynchronous Hyperproperties}

Extending hyperlogics to express asynchronous properties has only recently started to gain momentum \cite{GutsfeldMO21,BaumeisterCBFS21,BozzelliPS21,BeutnerF22a}. 
Baumeister et al.~introduce \AHLTL{} by extending \HyperLTL{} with explicit trajectory quantification \cite{BaumeisterCBFS21}. 
Gutsfeld et al.~introduced a variant of the polyadic $\mu$-calculus, called $H_\mu$, and accompanying asynchronous automata that are able to express asynchronous hyperproperties \cite{GutsfeldMO21}.
Bozzelli et al.~present \HyperLTL{}$_S$ by extending \HyperLTL{} with new modalities that remove redundant (for example stuttering) parts of a trace \cite{BozzelliPS21}.
Finite-state model checking for the logics proposed in \cite{GutsfeldMO21,BaumeisterCBFS21,BozzelliPS21} is undecidable.
Observation-based HyperLTL \cite{BeutnerF22a}, can be seen as fragment of \HyperLTL{}$_S$ that is geared towards automated verification and admits decidable finite-state model checking. 
We can obtain decidable fragments of $H_\mu$ \cite{GutsfeldMO21} and $\texttt{HyperLTL}_S$ \cite{BozzelliPS21} by bounding the asynchronous offset by a constant $k$, \ie, asynchronous execution may not run apart (``diverge'') for more than $k$ steps. 
The (known) decidable fragment of \AHLTL{} \cite{BaumeisterCBFS21} can be encoded into \HyperATLS{}  (\Cref{sec:async}).

\section{Conclusion}

We have introduced \HyperATLS{}, a temporal logic to express hyperproperties in multi-agent systems.
Besides the obvious benefits of simultaneously reasoning about strategic choice and information flow, \HyperATLS{} 
provides a natural formalism to express \emph{asynchronous} hyperproperties.
Despite the added expressiveness, \HyperATLS{} model checking remains decidable for the entire logic.
Its expressiveness and decidability, as well as the availability of practical model checking algorithms, make it a very promising choice for model checking tools for hyperproperties.

\bibliographystyle{alphaurl}
\bibliography{references}

\appendix

\section{}\label{app:proofEquiv}

In this section, we show \refProp{MCequiv} for the most interesting case where $\varphi = \llangle A \rrangle \pi\ldot \psi$.
For the correctness of the construction in the other cases see \eg, \cite{MullerSS88,FinkbeinerRS15}.
For simplicity, we assume that $\calG$ is a CGS.
Let $\calA_\psi$ be the inductively constructed automaton for $\psi$ that, by the induction hypothesis, is $\calG$-equivalent to $\psi$.
We consider the construction in \Cref{fig:stratConstruct}.
We show that $\calA_\varphi$ is $\calG$-equivalent to $\varphi$ by showing both directions of the ``iff'' in the definition of $\calG$-equivalence separately. 

\begin{lem}\label{lem:dir1}
	For any $p_1, \ldots, p_n \in S^\omega$, if $\otimes(p_1, \ldots, p_n) \in \calL(\calA_\varphi)$ then $[\pi_i \mapsto p_i]_{i = 1}^n \models_{\calG} \varphi$
\end{lem}
\begin{proof}
	Let $(T, r)$ be an accepting run of $\calA_\varphi$ on $\otimes(p_1, \ldots, p_n)$.
	We use the disjunctive choices made in $(T, r)$ to construct strategies $F_A = \{f_\agent \mid \agent \in A\}$ where $f_\agent : S^+ \to \moves$.
	For each element $u \in S^+$ we define $f_\agent(u)$ for each $\agent \in A$ as follows (it is important to construct the response to $u$ together as the path identified next is not be unique).
	Let $u = u(0)\cdots u(k)$.
	We check if there exists a node $\tau$ in $(T, r)$ such that the nodes along $\tau$ are labeled by $u$, \ie, 
	\begin{align*}
		r(\epsilon), r(\tau[0, 0]), r(\tau[0, 1]), \ldots, r(\tau[0, |\tau|-1]) = q_{\mathit{init}},(u(1), \_), \ldots, (u(k), \_).
	\end{align*}
	Note that the offset is intentional, \ie, the first element $u(0)$ does not occur in $\tau$ (as $\calA_\varphi$ skips over $s_n^\circ$ in the first step).
	If no such node exists, we define $f_\agent(u)$ arbitrarily for all $\agent \in A$ (any play that is compatible with the strategy never reaches this situation). 
	Otherwise, let $r(\tau) = (u(k), q)$ where $q$ is a state of $\calA^\mathit{det}_\psi$ (or $r(\tau) = q_\mathit{init}$ if $|u| = 1$).
	By construction of $\tilde{\calA}_\varphi$ we have that the children of $\tau$ satisfy the formula
	\begin{align*}
		\bigvee\limits_{\sigma : A \to \moves} \bigwedge\limits_{\sigma' : \overline{A} \to \moves}   
		\big(\delta(u(k), \sigma  + \sigma'), \rho(q, [p_1(k), \ldots, p_n(k), u(k)])\big).
	\end{align*}
	(The case where $r(\tau) = q_\mathit{init}$ is analogous.)
	There must thus exist (at least one) $\sigma_u : A \to \moves$ such that for every $\sigma' : \overline{A} \to \moves$ there is a child of $\tau$ labeled with
	\begin{align*}
		(\delta(u(k), \sigma_u  + \sigma'), \rho(q, [p_1(k), \ldots, p_n(k), u(k)]).
	\end{align*}
	We define
	\begin{align*}
		f_\agent(u) := \sigma_u(\agent)
	\end{align*}
	for each $\agent \in A$.
	By assumption of $\sigma_u$ for any $\sigma' : \overline{A} \to \moves$, there is a successor of $\tau$ labeled by $\delta(u(k), \sigma_u  + \sigma')$, \ie, all adversarial moves lead to a node in $(T, r)$ if agents in $A$ play $\sigma_u$.
	
	It is, therefore, easy to see that for all $p \in \mathit{out}(\calG, s_n^\circ, F_A)$ (where $s_n^\circ := s_0$ if $n = 0$ and $s_n^\circ := p_n(0)$ otherwise), there exist a path in $(T, r)$ labeled with $q_\mathit{init} (p(1), q_1) (p(2), q_2) \cdots$.
	By definition of $\tilde{\rho}$, the sequence of automaton state $q_0, q_1, q_2, \ldots$ (where $q_0$ is the initial state of $\calA^\mathit{det}_\psi$) is the unique run of $\calA^\mathit{det}_\psi$ on $\otimes(p_1, \ldots, p_n, p)$.
	As $(T, r)$ is accepting this sequence of automata states is accepting, we thus get that $\otimes(p_1, \ldots, p_n, p) \in \calL(\calA^\mathit{det}_\psi) = \calL(\calA_\psi) $.
	By the induction hypothesis (from the proof of \refProp{MCequiv}) we have that $\calA_\psi$ is $\calG$-equivalent to $\psi$ and so $[\pi_i \mapsto p_i]_{i = 1}^n \cup [\pi \mapsto p] \models_{\calG} \psi$.
	As this holds for all $p \in \mathit{out}(\calG, s_n^\circ, F_A)$, $F_A$ is a winning set of strategies and $[\pi_i \mapsto p_i]_{i = 1}^n \models_{\calG} \varphi$ by the semantics of \HyperATLS{}.
\end{proof}

\begin{lem}\label{lem:dir2}
	For any $p_1, \ldots, p_n \in S^\omega$, if $[\pi_i \mapsto p_i]_{i = 1}^n \models_{\calG} \varphi$ then $\otimes(p_1, \ldots, p_n) \in \calL(\calA_\varphi)$
\end{lem}
\begin{proof}
	Let $F_A = \{f_\agent \mid \agent \in A\}$ be a winning strategy for the agents in $A$, \ie, for all $p \in \mathit{out}(\calG, s_n^\circ, F_A)$,  $[\pi_i \mapsto p_i]_{i = 1}^n \cup [\pi \mapsto p]\models_{\calG} \psi$.
	We construct an accepting run $(T, r)$ of $\calA_\varphi$ on $\otimes(p_1, \ldots, p_n)$. 
	We construct this tree incrementally by adding children to existing nodes.
	
	The root $\epsilon$ is labeled by $q_\mathit{init}$.
	Now let $\tau \in T$ be any node in the tree constructed so far and let
	\begin{align*}
		r(\epsilon),r(\tau[0,0]),r(\tau[0, 1]), \ldots, r(\tau[0, |\tau|-1]) = q_{\mathit{init}},(s_1, q_1), \ldots, (s_k, q_k)
	\end{align*}
	be the label of the nodes along $\tau$.
	We define the move vector $\sigma_\tau : A \to \moves$ via $\sigma_\tau (\agent) := f_\agent(s_n^\circ,s_1, \ldots, s_k)$ for each $\agent \in A$ (where $s_n^\circ := s_0$ if $n = 0$ and $s_n^\circ := p_n(0)$ otherwise).
	For each move vectors $\sigma' : \overline{A} \to \moves$ we add a new child of $\tau$ labeled with 
	\begin{align*}
		(\delta(s_k, \sigma_\tau  + \sigma'), \rho(q, [p_1(|\tau|), \ldots, p_n(|\tau|), s_k]).
	\end{align*}
	By construction of the transition function of $\calA_\varphi$, those children satisfy the transition relation $\rho'$ (see \Cref{fig:stratConstruct}).
	
	The constructed tree $(T, r)$ is thus a run on $\otimes(p_1, \ldots, p_n)$.
	We now claim that $(T, r)$ is accepting. Consider any infinite path in $(T, r)$ labeled by $q_\mathit{init} (s_1, q_1) (s_2, q_2), \ldots$. By construction of the tree, it is easy to see that the path $p = s_n^\circ, s_1, s_2, \ldots$ is contained in $\mathit{out}(\calG, s_n^\circ, F_A)$ as all children added to a node were added in accordance with $F_A$.
	As $F_A$ is winning and by the \HyperATLS{} semantics, we get that $[\pi_i \mapsto p_i]_{i = 1}^n \cupdot [\pi \mapsto p] \models_{\calG} \psi$.
	By induction hypothesis (from the proof of \refProp{MCequiv}) we get that $\calA_\psi$ is $\calG$-equivalent to $\psi$ so $\otimes(p_1, \ldots, p_n, p) \in \calL(\calA_\psi) = \calL(\calA_\psi^\mathit{det})$.
	By construction of $\calA_\varphi$ the automaton sequence $q_0, q_1, q_2, \ldots$ (where $q_0$ is the initial state of $\calA_\psi^\mathit{det}$) is the unique run of $\calA_\psi^\mathit{det}$ on $\otimes(p_1, \ldots, p_n, p)$ and therefore accepting.
	So $(T, r)$ is accepting, and it follows that $\otimes(p_1, \ldots, p_n) \in \calL(\calA_\varphi)$.
\end{proof}

\end{document}